\documentclass[11pt, a4paper, reqno]{amsart}
\usepackage{amsthm, amsmath, amsfonts, amssymb, appendix, dsfont, latexsym}
\usepackage{amsfonts, mathrsfs}
\usepackage{tikz}
\usepackage{braids}
\usepackage{ifthen}
\usetikzlibrary{arrows}

\makeatletter
\usepackage{delarray, a4, color}
\usepackage{euscript}
\usepackage[latin1]{inputenc}
\usepackage{enumitem}
\usepackage{hyperref}
\hypersetup{
   colorlinks,
   menucolor=black,
   linkcolor=black,
   citecolor=black,
   urlcolor=blue
}
\definecolor{darkred}{cmyk}{0,1,1,0.4}

\@namedef{subjclassname@2020}{\textup{2020} Mathematics Subject Classification}

\theoremstyle{plain}
\newtheorem{theorem}{Theorem}[section]
\newtheorem*{theorem*}{Theorem}
\newtheorem{lemma}[theorem]{Lemma}
\newtheorem{corollary}[theorem]{Corollary}
\newtheorem{proposition}[theorem]{Proposition}
\newtheorem{question}[theorem]{Question}

\theoremstyle{definition}
\newtheorem{definition}[theorem]{Definition}
\theoremstyle{remark}
\newtheorem{remark}[theorem]{Remark}
\newtheorem{remark*}[theorem]{Remark\textup{*}}
\newtheorem{example}[theorem]{Example}

\numberwithin{equation}{section}
\numberwithin{figure}{section}

\DeclareMathAlphabet{\mathpzc}{OT1}{pzc}{m}{it}
\def\eps{\varepsilon}

\def\C {\mathbb{C}}

\def\N {\mathbb{N}}

\def\R {\mathbb{R}}

\def\S {\mathbb{S}}
\def\Z {\mathbb{Z}}

\def\eD {\EuScript{D}}
\def\eN {\EuScript{N}}
\def\eP {\EuScript{P}}
\newcommand\1{{\ensuremath {\mathds 1} }}

\newcommand{\keyword}[1]{\textbf{#1}}

\newcommand{\0}{\mathbf{0}}
\newcommand{\balpha}{{\boldsymbol{\alpha}}}

\newcommand{\bA}{\mathbf{A}}

\newcommand{\bR}{\mathbf{R}}
\newcommand{\bX}{\mathbf{X}}

\newcommand{\be}{\mathbf{e}}

\newcommand{\bp}{\mathbf{p}}
\newcommand{\br}{\mathbf{r}}

\newcommand{\bx}{\mathbf{x}}

\newcommand{\bz}{\mathbf{z}}
\newcommand{\cA}{\mathcal{A}}

\newcommand{\cC}{\mathcal{C}}
\newcommand{\cD}{\mathcal{D}}
\newcommand{\cF}{\mathcal{F}}
\newcommand{\cH}{\mathcal{H}}

\newcommand{\cL}{\mathcal{L}}

\newcommand{\gu}{\mathfrak{u}}

\newcommand{\hp}{\hat{p}}
\newcommand{\hx}{\hat{x}}

\newcommand{\hT}{\hat{T}}
\newcommand{\hV}{\hat{V}}

\newcommand{\hbp}{\hat{\mathbf{p}}}

\newcommand{\sGL}{\textup{GL}}

\newcommand{\sSU}{\textup{SU}}

\newcommand{\sU}{\textup{U}}

\newcommand{\ssp}{\textup{p}}

\newcommand{\sx}{\textup{x}}

\newcommand{\sz}{\textup{z}}
\newcommand{\tcC}{\tilde{\mathcal{C}}}
\newcommand{\tX}{\tilde{X}}
\newcommand{\tY}{\tilde{Y}}

\DeclareMathOperator{\Ad}{\mathrm{Ad}}
\DeclareMathOperator{\infspec}{\mathrm{inf\, spec\,}}
\DeclareMathOperator{\curl}{\mathrm{curl}}
\DeclareMathOperator{\diag}{\mathrm{diag}}

\DeclareMathOperator{\Fib}{\mathrm{Fib}}
\DeclareMathOperator{\pr}{\mathrm{pr}}

\DeclareMathOperator{\sign}{\mathrm{sign}}
\DeclareMathOperator{\Span}{\mathrm{Span}}
\DeclareMathOperator{\supp}{\mathrm{supp}}

\newcommand{\mapsfrom}{\reflectbox{$\mapsto$}}

\newcommand{\dist}{\mathrm{dist}}

\newcommand{\sym}{\mathrm{sym}}
\newcommand{\asym}{\mathrm{asym}}

\newcommand{\rel}{\mathrm{rel}}
\newcommand{\vphi}{\varphi}

\newcommand{\bDelta}{{\mbox{$\triangle$}\hspace{-8.0pt}\scalebox{0.8}{$\triangle$}}}

\usepackage{mathtools}
\mathtoolsset{showonlyrefs=true}

\usepackage{placeins}

\makeatletter
\newcounter{mycounter}
\newcommand{\fs}[3][]{
  \begin{tikzpicture}[scale=0.3,font=\footnotesize,anchor=mid,baseline={([yshift=-1ex]current bounding box.center)}]
    \def\height{1.5}
    \def\offset{0.5}
    \ifthenelse{\equal{#1}{}}{\def\height{1}}{} 
    \setcounter{mycounter}{0}
    \@for\el:=#2\do{
      \ifthenelse{\equal{#1}{\value{mycounter}}}{
        \braid at (\value{mycounter},\height) s_1^{-1};
        \stepcounter{mycounter}
      }{
        \stepcounter{mycounter}
        \ifthenelse{\equal{#1}{\value{mycounter}}}{
        }{
          \draw (\value{mycounter}, \height) to (\value{mycounter}, 0);
        }
      }
      \node at (\value{mycounter}, \height+\offset) {$\el$};
    }
    \draw (0, 0) to (\value{mycounter}+1, 0);
    \setcounter{mycounter}{0}
    \@for\el:=#3\do{
      \node at (\value{mycounter}+0.5, -0.6) {$\el$};
      \stepcounter{mycounter}
    }
  \end{tikzpicture}
}
\newcommand{\fswide}[3][]{
  \begin{tikzpicture}[scale=0.3,font=\footnotesize,anchor=mid,baseline={([yshift=-.5ex]current bounding box.center)}]
    \def\height{1.75}
    \def\offset{0.5}
    \def\widthfactor{2}
    \ifthenelse{\equal{#1}{}}{\def\height{1}}{} 
    \setcounter{mycounter}{0}
    \@for\el:=#2\do{
      \ifthenelse{\equal{#1}{\value{mycounter}}}{
        \braid[width=\widthfactor cm, height=1.25 cm] at (\widthfactor*\value{mycounter},\height) s_1^{-1};
        \stepcounter{mycounter}
      }{
        \stepcounter{mycounter}
        \ifthenelse{\equal{#1}{\value{mycounter}}}{
        }{
          \draw (\widthfactor*\value{mycounter}, \height) to (\widthfactor*\value{mycounter}, 0);
        }
      }
      \node at (\widthfactor*\value{mycounter}, \height+\offset) {$\el$};
    }
    \draw (0, 0) to (\widthfactor*\value{mycounter}+\widthfactor*1, 0);
    \setcounter{mycounter}{0}
    \@for\el:=#3\do{
      \node at (\widthfactor*\value{mycounter} + \widthfactor*0.5, -0.6) {$\el$};
      \stepcounter{mycounter}
    }
  \end{tikzpicture}
}
\newcommand{\fswideflex}[4][]{
  \begin{tikzpicture}[scale=0.3,font=\footnotesize,anchor=mid,baseline={([yshift=-.5ex]current bounding box.center)}]
    \def\height{1.75}
    \def\offset{0.5}
    \def\widthfactor{#4}
    \ifthenelse{\equal{#1}{}}{\def\height{1}}{} 
    \setcounter{mycounter}{0}
    \@for\el:=#2\do{
      \ifthenelse{\equal{#1}{\value{mycounter}}}{
        \braid[width=\widthfactor cm, height=1.25 cm] at (\widthfactor*\value{mycounter},\height) s_1^{-1};
        \stepcounter{mycounter}
      }{
        \stepcounter{mycounter}
        \ifthenelse{\equal{#1}{\value{mycounter}}}{
        }{
          \draw (\widthfactor*\value{mycounter}, \height) to (\widthfactor*\value{mycounter}, 0);
        }
      }
      \node at (\widthfactor*\value{mycounter}, \height+\offset) {$\el$};
    }
    \draw (0, 0) to (\widthfactor*\value{mycounter}+\widthfactor*1, 0);
    \setcounter{mycounter}{0}
    \@for\el:=#3\do{
      \node at (\widthfactor*\value{mycounter} + \widthfactor*0.5, -0.6) {$\el$};
      \stepcounter{mycounter}
    }
  \end{tikzpicture}
}
\newcommand{\fswider}[3][]{
  \begin{tikzpicture}[scale=0.3,font=\footnotesize,anchor=mid,baseline={([yshift=-.5ex]current bounding box.center)}]
    \def\height{1.75}
    \def\offset{0.5}
    \def\widthfactor{3}
    \ifthenelse{\equal{#1}{}}{\def\height{1}}{} 
    \setcounter{mycounter}{0}
    \@for\el:=#2\do{
      \ifthenelse{\equal{#1}{\value{mycounter}}}{
        \braid[width=\widthfactor cm, height=1.25 cm] at (\widthfactor*\value{mycounter},\height) s_1^{-1};
        \stepcounter{mycounter}
      }{
        \stepcounter{mycounter}
        \ifthenelse{\equal{#1}{\value{mycounter}}}{
        }{
          \draw (\widthfactor*\value{mycounter}, \height) to (\widthfactor*\value{mycounter}, 0);
        }
      }
      \node at (\widthfactor*\value{mycounter}, \height+\offset) {$\el$};
    }
    \draw (0, 0) to (\widthfactor*\value{mycounter}+\widthfactor*1, 0);
    \setcounter{mycounter}{0}
    \@for\el:=#3\do{
      \node at (\widthfactor*\value{mycounter} + \widthfactor*0.5, -0.6) {$\el$};
      \stepcounter{mycounter}
    }
  \end{tikzpicture}
}
\newcommand{\fsfusedbraided}[5]{
  \begin{tikzpicture}[scale=0.3,font=\footnotesize,anchor=mid,baseline={([yshift=-.5ex]current bounding box.center)}]
    \def\height{3}
    \def\offset{0.5}
    \node at (1, \height+\offset) {$#2$};
    \node at (2, \height+\offset) {$#3$};
    \braid at (1, \height) s_1^{-1};
    \draw (0.5, 0) to (2.5, 0);
    \draw (1, 1.5) to [bend left=-30] (1.5, 1);
    \draw (2, 1.5) to [bend left=30] (1.5, 1);
    \draw (1.5, 1) to (1.5, 0);
    \node at (0.75, -0.6) {$#1$};
    \node at (2.25, -0.6) {$#4$};
    \node at (2, 0.6) {$#5$};
  \end{tikzpicture}
}
\newcommand{\fsfusedbraidedshort}[3]{
  \begin{tikzpicture}[scale=0.3,font=\footnotesize,anchor=mid,baseline={([yshift=-.5ex]current bounding box.center)}]
    \def\height{3}
    \def\offset{0.5}
    \node at (1, \height+\offset) {$#1$};
    \node at (2, \height+\offset) {$#2$};
    \braid at (1, \height) s_1^{-1};
    \draw (1, 1.5) to [bend left=-30] (1.5, 1);
    \draw (2, 1.5) to [bend left=30] (1.5, 1);
    \draw (1.5, 1) to (1.5, 0);
    \node at (2, 0.5) {$#3$};
  \end{tikzpicture}
}
\newcommand{\fsfused}[5]{
  \begin{tikzpicture}[scale=0.3,font=\footnotesize,anchor=mid,baseline={([yshift=-.5ex]current bounding box.center)}]
    \def\height{1.75}
    \def\offset{0.5}
    \node at (1, \height+\offset) {$#2$};
    \node at (2, \height+\offset) {$#3$};
    \draw (0.5, 0) to (2.5, 0);
    \draw (1, \height) to [bend left=-30] (1.5, 1);
    \draw (2, \height) to [bend left=30] (1.5, 1);
    \draw (1.5, 1) to (1.5, 0);
    \node at (0.75, -0.6) {$#1$};
    \node at (2.25, -0.6) {$#4$};
    \node at (2, 0.5) {$#5$};
  \end{tikzpicture}
}
\newcommand{\fsfusedshort}[3]{
  \begin{tikzpicture}[scale=0.3,font=\footnotesize,anchor=mid,baseline={([yshift=-1.5ex]current bounding box.center)}]
    \def\height{1.75}
    \def\offset{0.5}
    \node at (1, \height+\offset) {$#1$};
    \node at (2, \height+\offset) {$#2$};
    \draw (1, \height) to [bend left=-30] (1.5, 1);
    \draw (2, \height) to [bend left=30] (1.5, 1);
    \draw (1.5, 1) to (1.5, 0);
    \node at (2, 0.5) {$#3$};
  \end{tikzpicture}
}
\makeatother

\title{Exchange and exclusion in the non-abelian anyon gas}

\author[D. Lundholm]{Douglas LUNDHOLM}
\address{Uppsala University, Department of Mathematics - Box 480, SE-751 06 Uppsala, Sweden}
\email{douglas.lundholm@math.uu.se}

\author[V. Qvarfordt]{Viktor QVARFORDT}
\address{Stockholm Mathematics Centre (SMC), Stockholm, Sweden}
\email{viktor.qvarfordt@gmail.com }


\begin{document}

\begin{abstract}
We review and develop the many-body spectral theory of ideal anyons,
i.e.\ identical quantum particles in the plane whose exchange rules are governed by 
unitary representations of the braid group on $N$ strands. 
Allowing for arbitrary rank (dependent on $N$) and non-abelian representations, 
and letting $N \to \infty$,
this defines the ideal non-abelian many-anyon gas.
We compute exchange operators and phases for a common and wide class of representations
defined by fusion algebras,
including the Fibonacci and Ising anyon models.
Furthermore, we extend methods of statistical repulsion (Poincar\'e and Hardy inequalities)
and a local exclusion principle (also implying a Lieb--Thirring inequality)
developed for abelian anyons to arbitrary geometric anyon models,
i.e.\ arbitrary sequences of unitary representations of the braid group,
for which two-anyon exchange is nontrivial.
\end{abstract}

\maketitle
\vspace{-0.5cm}
\setcounter{tocdepth}{2}
\tableofcontents

\vspace{-1.0cm}

\section{Introduction}\label{sec:intro}

	Quantum statistics is fundamental to our understanding of the physical world.
	In three spatial dimensions this refers to the division of particles into
	either bosons ---
    such as Higgs, photons, gluons, and other force carriers,
	as well as certain atoms such as $^4$He;
	or fermions --- such as electrons, quarks, 
	and other ordinary matter particles including $^3$He atoms.
	Their collective behavior explains everyday phenomena such as 
	conduction vs.\ insulation, lasing, as well as the stability of large systems of matter such as planets
	and stars against electromagnetic or gravitational collapse.
	When quantum systems of many particles are confined to two dimensions, however,
	other possibilities emerge.
	Such intermediate quantum statistics of exchange phases \cite{LeiMyr-77,GolMenSha-81,GolMenSha-85}
	(as opposed to the exclusion principle \cite{Gentile-40,Gentile-42,Haldane-91}),
	first appeared in theory in the 1970-80's, 
	with associated particles now known as anyons \cite{Wilczek-82b,Wu-84b}.
	They filled a gap in the 
	logical argument used since the 1920's to derive the boson/fermion dichotomy
	\cite{Girardeau-65,Klaiber-68,Souriau-70,StrWil-70,LaiDeW-71,Froehlich-88}, and
	were later found to have an application within the fractional quantum Hall
	effect (FQHE), both in their abelian form \cite{Laughlin-99,AroSchWil-84}
	as well as in their non-abelian form \cite{MooRea-91,Wen-91,GurNay-97}.
	The latter possibility eventually spawned promising applications to quantum computing and 
	complexity theory \cite{Kitaev-03,Lloyd-02,FreKitLarWan-03}.
	The key concepts involved, 
    rooted in configuration spaces and braid groups,
    naturally intersect with many branches of pure mathematics, such as
    knot theory, quantum groups, Chern--Simons theory, and conformal field theory (CFT) 
	\cite{Witten-89,MooSei-89}.
	In a quantum gravity context, both abelian 
	\cite{tHooft-88,DesJac-88,Carlip-90,KayStu-91}
	and non-abelian \cite{PitRui-15}
	anyons appear, the former as a toy model for point masses in 2+1 dimensions 
	or cosmic strings in 3+1 dimensions, and
	the latter proposed for the degrees of freedom on the spherical
	event horizon of `normal' black holes.
	Reviews focusing on the physics of anyons are very numerous; see e.g.\ 
	\cite{DatMurVat-03,Forte-92,Froehlich-90,IenLec-92,Jackiw-90,Khare-05,Lerda-92,Lundholm-23,Myrheim-99,Nayak-etal-08,Ouvry-07,Stern-08,Wilczek-90}.
	For introductions to some of the mathematical techniques involved, we refer to
	\cite{DelRowWan-16,DoeStoTol-01,FroGab-90,Girardot-21,LieSeiSolYng-05,Lundholm-17,MacSaw-19,MunSch-95,Oddis-20,Rougerie-16,RowWan-18,Rougerie-21}.

    There has been a recent wave of renewed interest in anyons from both theory and experiment.
	On the theory side,
    the concrete emergence of anyons from underlying systems of bosons and fermions 
	has been studied in increasing detail;
    see e.g.
    \cite{Yakaboylu-etal-20,LamLunRou-22,LunRou-16,ValWesOhb-20,Yakaboylu-etal-19}.
	On the experimental side, signatures of individual anyons have finally been detected in FQHE setups \cite{Bartolomei-etal-20,Nakamura-etal-20}.
    A canonical approach is by measuring phase 
	interference for individual particles, which
    requires extremely sensitive equipment
    and even carries some analytical ambiguities \cite{Forte-91,Jain-07,ReaDas-24},
	while another, more robust, route would be to observe indirect effects of statistics in 
	density distributions for anyons \cite{CorDubLunRou-19,LevLunRou-25,Ataei-etal-25}.
	However, a fundamental problem \cite{CanJoh-94,MurSha-94,NayWil-94} 
	in this regard concerns the basic relationship between exchange phases 
	and exclusion (which in the fermionic case is attributed to Pauli \cite{Pauli-47}),
    i.e.\ the many-body spectral theory including the ground-state problem for the ideal homogeneous quantum gas,
	and this is only beginning to be addressed
	in a mathematically precise manner for anyons. 
    We refer to \cite{Lundholm-23} for a recent overview on the properties of 2D anyon gases from this perspective.

\smallskip

    {\bf Brief outline:}
    The main purpose of the present work is to extend some of the basic mathematical techniques developed 
    for the spectral theory of the ideal abelian anyon gas to the non-abelian context. 
    After an introduction and statement of results in the remainder of this section, 
    we first give in Section~\ref{sec:models} a compact review of the braid group representations we will be mainly concerned with, 
    and then in Section~\ref{sec:phases} show how to compute concrete exchange phases in these anyon models,
    including a few physically relevant examples.
    In Section~\ref{sec:ham}, we consider the Hamiltonian operator for an ideal but otherwise arbitrary geometric $N$-anyon model in the plane,
    and discuss the notion of statistics transmutation.
    In Section~\ref{sec:repulsion}, the connection between anyon models, exchange phases, and the exclusion principle is made mathematically precise in the form of Poincar\'e and Hardy inequalities.
	Finally, in Section~\ref{sec:gas}, some fundamental consequences for the ground-state energy of anyon gases are derived, including a Lieb--Thirring inequality for the inhomogeneous gas.

\subsection{Anyons --- abelian vs.\ non-abelian}

	The most direct route to quantum statistics is to start from an $N$-body
	Schr\"odinger wave function $\Psi\colon (\R^2)^N \to \C$,
	where $\sx = (\bx_1,\ldots,\bx_N) \in (\R^2)^N$ 
	denotes the positions of $N$ point particles in the plane.
	The square of its amplitude $|\Psi(\sx)|$ 
	is normalized $\int_{\R^{2N}} |\Psi|^2 = 1$ 
    and is interpreted as the
    probability density to find the particle labeled $j$ at the position $\bx_j \in \R^2$,
	at the same instant in time for $j = 1,\ldots,N$.
	In the case that the particles are all identical (indistinguishable), 
	this probability must be symmetric:
	\begin{equation}\label{eq:exchange-amp}
		|\Psi(\bx_1, \ldots, \bx_j, \ldots, \bx_k, \ldots, \bx_N)|^2
		= |\Psi(\bx_1, \ldots, \bx_k, \ldots, \bx_j, \ldots, \bx_N)|^2, 
		\quad j \neq k.
	\end{equation}
	This leaves the possibility for an exchange phase:
	\begin{equation}\label{eq:exchange-phase}
		\Psi(\bx_1, \ldots, \bx_j, \ldots, \bx_k, \ldots, \bx_N)
		= e^{i\alpha\pi} \Psi(\bx_1, \ldots, \bx_k, \ldots, \bx_j, \ldots, \bx_N), 
		\quad j \neq k.
	\end{equation}
	Further, by indistinguishability and logical consistency, it turns out that this phase must be independent
	of which pair of particles are exchanged 
    if none of the other particles interfere.
	Thus, if $\alpha = 0$ (or $\alpha \in 2\Z$) then $\Psi$ is symmetric and
	this defines a \keyword{bosonic state}, i.e.\ \keyword{bosons}, subject to Bose--Einstein statistics \cite{Bose-24,Einstein-24},
	while if $\alpha=1$ (or $\alpha \in 2\Z+1$)
	then $\Psi$ is antisymmetric and this defines
    a \keyword{fermionic state}, i.e.\ \keyword{fermions},
	subject to Fermi--Dirac statistics \cite{Fermi-26,Dirac-26}.
	In particular, bosons admit product states (Bose--Einstein condensate)
	\begin{equation}\label{eq:product-state}
		\Psi(\bx_1,\ldots,\bx_N) = \prod_{j=1}^N u(\bx_j),
	\end{equation}
	exhibiting 
    independent identical distribution in a one-body state $u \in L^2(\R^2)$,
	while fermions are necessarily correlated subject to 
	\keyword{Pauli's exclusion principle} \cite{Pauli-25,Pauli-47}:
	\begin{equation}\label{eq:Pauli}
		\Psi(\sx) = 0 \qquad \text{if $\bx_j = \bx_k$ for $j \neq k$}.
	\end{equation}
	Moreover, the space of fermionic states is spanned by Slater determinants (exterior products) of one-body states:
	\begin{equation}\label{eq:Slater-state}
		\Psi(\sx) = u_1 \wedge \ldots \wedge u_N (\sx) = (N!)^{-1/2} \det [u_k(\bx_j)].
	\end{equation}
	If $\alpha \notin \Z$ then \eqref{eq:exchange-phase} defines \keyword{anyons}
	(as in `any phase' \cite{Wilczek-82b})
	with \keyword{statistics parameter} $\alpha$,
	and requires that $\Psi$ is a multivalued function, 
	since a double exchange is not the identity
    (we assume here a continuous elementary exchange process with a positive orientation).
	As the free \emph{ideal} $N$-partice Hamiltonian operator
    we take the non-relativistic kinetic energy:
	$$
		\hT_0 := -\Delta_{\sx} = -\sum_{j=1}^N \Delta_{\bx_j}.
	$$
	It is defined to act on these (single- or multivalued) functions $\Psi$, 
	with sufficient regularity and suitable
	boundary conditions on the unit square 
	$Q_0 = [0,1]^2$, for example, to make the system finite
    (precise definitions will be given in Section~\ref{sec:ham}).
	The ground states of the form \eqref{eq:product-state} respectively \eqref{eq:Slater-state}
	immediately produce calculable ground-state energies $E_N = \infspec \hT_0$.
	For bosons (i.e.\ the \keyword{ideal Bose gas}), these are simply
	\begin{equation}\label{eq:Bose-energy}
		E_N^{(\alpha=0)} = N\lambda_0(-\Delta_{Q_0}^\eN) = 0,
		\quad \text{or} \quad
		\bar{E}_N^{(\alpha=0)} = N\lambda_0(-\Delta_{Q_0}^\eD) = 2\pi^2 N,
	\end{equation}
	while for  
	fermions (i.e.\ the \keyword{ideal Fermi gas}) 
	the semiclassical \keyword{Weyl's law} \cite{Weyl-12} yields
	\begin{equation}\label{eq:Fermi-energy}
		E_N^{(\alpha=1)} = \sum_{k=0}^{N-1} \lambda_k(-\Delta_{Q_0}^\eN) = 2\pi N^2 + o(N^2),
		\ \ \text{or} \ \ 
		\bar{E}_N^{(\alpha=1)} = \sum_{k=0}^{N-1} \lambda_k(-\Delta_{Q_0}^\eD) = 2\pi N^2 + o(N^2),
	\end{equation}
	as $N \to \infty$, where
	$\lambda_0 \le \lambda_1 \le \ldots$
	denote the eigenvalues of $-\Delta_{Q_0}^{\eN/\eD}$,
	the Neumann/Dirichlet Laplacian on $Q_0$, 
	ordered according to their multiplicity.
	On the other hand, the anyonic case presents a real difficulty as it is not 
	directly reducible to simple
	product states but actually turns out to be equivalent to a system of 
	\emph{interacting} bosons or fermions with complicated many-body correlations.
	In fact, the operator $\hT_0$ acting on multivalued $\Psi$ subject to \eqref{eq:exchange-phase}
	is equivalent to the magnetic operator
	$$
		\hT_\alpha := \sum_{j=1}^N \left( -i\nabla_{\bx_j} + \alpha\sum_{k \neq j} (\bx_j - \bx_k)^{-\perp} \right)^2,
		\quad
		\bx^{-\perp} := \frac{(-y,x)}{x^2 + y^2} \ \text{for} \ \bx = (x,y) \in \R^2 \setminus \{\0\},
	$$
	acting on bosonic $\Psi$. 
	Thus, mathematical techniques for interacting Bose gases are essential.
	
	If $\Psi$ takes values in some Hilbert space $\cF$,
	then the condition \eqref{eq:exchange-phase} may be replaced by
	\begin{equation}\label{eq:exchange-op}
		\Psi(\bx_1, \ldots, \bx_j, \ldots, \bx_k, \ldots, \bx_N)
		= U \Psi(\bx_1, \ldots, \bx_k, \ldots, \bx_j, \ldots, \bx_N), 
		\quad j \neq k,
	\end{equation}
	where $U \in \sU(\cF)$, the group of unitary operators on $\cF$.
	In general, there will be topological consistency conditions on the possibilities for
	exchange (considered properly as \emph{continuous} loops in the configuration 
	space of particle positions)
	and one must consider a representation of the corresponding braid group,
	i.e.\ a homomorphism
	\begin{equation}\label{eq:braid-rep}
		\rho\colon B_N \to \sU(\cF).
	\end{equation}
	We will later define precisely what we mean by this.
	In our context, where we consider \emph{ideal} anyons, the most general model
	we may consider, referred to as a \keyword{geometric anyon model},
	will be in one-to-one correspondence with such a representation $\rho$.
	The case \eqref{eq:exchange-phase} is a special case where each generator $\sigma_j$
	of $B_N$ is represented as the phase $\rho(\sigma_j) = e^{i\pi\alpha}$.
	
	Whereas phases are abelian, a representation \eqref{eq:braid-rep} such that $\rho(\sigma_j)$
	do not all commute is non-abelian
	(non-abelian anyons are also known as \keyword{nonabelions} or \keyword{plektons} in the literature).
	We will focus attention on a particular family of non-abelian representations
	which arise naturally from the perspective of  
	quantum field theory \cite{FreRehSch-89,FroGab-90}
	and which include a number of anyon models which have been proposed to be
	relevant in condensed matter contexts such as the FQHE.
	Some of these, such as one known as the Fibonacci anyon model, 
	have even been proposed as good candidates for topologically protected
	quantum computing \cite{FreKitLarWan-03,Kitaev-06,Nayak-etal-08,RowWan-18}.
	We refer to this general class of representations as  
	\keyword{algebraic anyon models}.

	Our main aim in this work is to 
	connect the relatively well-developed 
    algebraic theory of non-abelian anyons to 
	the hitherto relatively undeveloped
	many-body spectral theory and the analysis
	of corresponding geometrically defined Laplace operators $\hT_\rho$,
	in order to eventually 
	be able to compute or estimate 
	a physically essential 
	property such as the ground-state energy $E_N$ in the many-body (thermodynamic) limit 
	$N \to \infty$.

\subsection{Some mathematically rigorous results for abelian anyons}

	Apart from a few special systems such as for $2 \le N \le 4$ 
	\cite{LeiMyr-77,Wilczek-82b,AroSchWilZee-85,Sen-91,MurLawBraBha-91,SpoVerZah-91,SpoVerZah-92,Minor-93}, 
	strong external field \cite{Ouvry-07},
	as well as for a subspace of states in harmonic oscillator confinement for general $N$ 
	\cite{Wu-84b,Chou-91a,Chou-91b},
	the spectrum of the many-anyon Hamiltonian $\hT_\alpha$ remains unknown for $\alpha \notin \Z$ (even with a one-body potential added).
	It has been noted however that in a rotationally symmetric situation 
	such as the harmonic oscillator potential
	the ground state (g.s.)\ in the bosonic representation
	must have a very specific total angular momentum \cite{ChiSen-92}:
	$$
		L = -\alpha \binom{N}{2} + O(N^{3/2}),
	$$
	corresponding to a relative angular momentum $-\alpha$ for each pair of particles.
	This implies that the g.s.\ energy $E_N$ as a function of $\alpha$ will have to have level crossings
	between different suitable $L \in 2\Z$ and is therefore only likely 
	to be smooth on intervals 
	of short lengths, tending to zero as $N \to \infty$.
	
\begin{figure}[t]
	\begin{center}
	\begin{tikzpicture}
		\node [above right] at (0,0) {\includegraphics[scale=1.1, clip, trim=0pt 0pt 0pt 0pt]{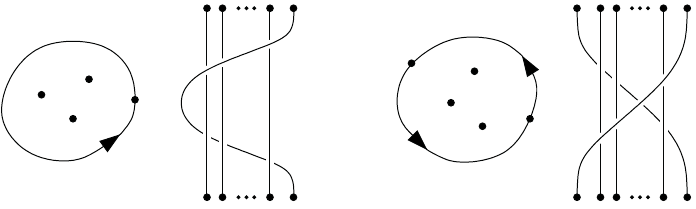}};
		\node [above right] at (2,0) {$e^{i2p\alpha\pi}$};
		\node [above right] at (8.5,0) {$e^{i(2p+1)\alpha\pi}$};
		\node [above right] at (1.2,1.9) {$p$};
		\node [above right] at (8.7,1.9) {$p$};
	\end{tikzpicture}
	\end{center}
	\caption{
		One- respectively two-particle loops of abelian anyons in the plane, with their
		respective phases and braid diagrams obtained by 
		projecting and ordering the particles on the horizontal axis
		and with time running upwards on the vertical axis.
		In each loop $p$ other particles are enclosed, and
		the total obtained phase is $\alpha\pi$ times the number of
		simple braids appearing in the diagram.
		}
	\label{fig:loops}
\end{figure}
	
	\keyword{Statistical repulsion}, generalizing the Pauli principle \eqref{eq:Pauli} for fermions, 
	manifests concretely in three ways: \underline{\emph{primarily}} as an effective scalar repulsion 
	between pairs of particles, making $|\Psi|^2$ smaller along the diagonals of the configuration space;
	\underline{\emph{secondarily}} as a non-trivial growth of the local Neumann energy with the number of particles;
	and \underline{\emph{tertiarily}} as a degeneracy pressure in the density 
	(such as a Thomas--Fermi profile compared to a condensed one-body profile \eqref{eq:product-state}).
	The first effect has been observed already for $N=2$ anyons in early works, 
	manifested as a 
	centrifugal-barrier repulsion due to the fractional relative angular momentum.
	In the context of the abelian many-anyon gas it was noted 
	by D.\,L.\ and Solovej in \cite{LunSol-13a,LunSol-13b} 
	and quantified by the \keyword{`fractionality'} of $\alpha$
	$$
		\alpha_N := \min\limits_{p \in \{0, 1, \ldots, N-2\}} \min\limits_{q \in \Z} |(2p+1)\alpha - 2q|,
	$$
	which entered as a coupling constant in a many-anyon Hardy inequality 
	\begin{equation}\label{eq:anyon-Hardy}
		\hT_\alpha \ge \frac{4\alpha_N^2}{N} \sum_{1 \le j < k \le N} |\bx_j-\bx_k|^{-2}
	\end{equation}
    (interpreted in the sense of non-negative quadratic forms).
	This type of inequality had been proven for fermions (for which $\alpha_N=1 \ \forall N$) by
	Hoffmann-Ostenhof et al.\ \cite[Theorem~2.8]{HofLapTid-08} 
	(improvements in higher dimensions were discussed in \cite{FraHofLapSol-24}).
	They had in fact proven a similar bound also for anyons \cite[Theorem~2.7]{HofLapTid-08} 
	although in this case with a much weaker constant, 
	replacing $4\alpha_N^2/N$ by
	$$
		C_{\alpha,N} = \min_{p \in \{1,\ldots,N-1\}} \left( \frac{1}{p} \min_{q \in \Z} |p\alpha-q| \right)^2,
	$$
	which vanishes e.g.\ for fermions.
	For this result they used a many-body version of a magnetic Hardy inequality 
	for an Aharonov--Bohm singular field due to Laptev and Weidl \cite{LapWei-98}
	and generalized by Balinsky \cite{Balinsky-03} to the case of multiple singularities.
	In a sense, it amounts to considering a single particle in the background field of the others fixed 
	(monodromy; cf.\ Figure~\ref{fig:loops}, left).
	On the other hand,
	the validity of the stronger inequality \eqref{eq:anyon-Hardy} for fermions boils down to  
	the Poincar\'e inequality in pairwise relative coordinates, namely
	that the energy for an antipodal-antisymmetric and nonzero function on the unit 
	circle $\S^1$ (or unit sphere $\S^{d-1}$) 
	cannot be zero:
	$$
		\int_0^{2\pi} |u'(\vphi)|^2 \,d\vphi \ge \int_0^{2\pi} |u(\vphi)|^2 \,d\vphi,
		\qquad \text{if} \quad u(\vphi+\pi) = -u(\vphi).
	$$
	Such a relative Poincar\'e inequality 
    (Wirtinger's inequality)
    was generalized to the anyonic setting 
	in \cite{LunSol-13a} 
	in the form of a symmetry adaptation of the magnetic inequality of Laptev and Weidl.
	For suitably multivalued functions on $\S^1$, this amounts to
	$$
		\int_0^{2\pi} |u'(\vphi)|^2 \,d\vphi \ge \min_{q \in \Z} |\alpha-2q|^2 \int_0^{2\pi} |u(\vphi)|^2 \,d\vphi,
		\qquad \text{if} \quad u(\vphi+\pi) = e^{i\alpha\pi} u(\vphi).
	$$
	Combining this angular result with the radial dependence of the derivative into a centrifugal-barrier repulsion,
    this provides then the first manifestation of statistical repulsion
	(note that it is rooted in simple pair exchange or \emph{half}-monodromy; cf.\ Figure~\ref{fig:loops}, right).

	Also important for the subsequent development, 
	a more powerful local version of the Hardy inequality was introduced
	(it will be given in Section~\ref{sec:repulsion-Hardy} in its improved form).
	In the case that the limiting fractionality
	$$
		\alpha_* := \inf_{N \ge 2} \alpha_{N} = \lim_{N \to \infty} \alpha_{N}
	$$
	is positive, which is true iff $\alpha$ is an odd-numerator rational number,
	this local version of the inequality \eqref{eq:anyon-Hardy} was used to derive local 
	energy estimates for $\hT_\alpha$
	(with Neumann b.c.\ on $Q_0^N$)
	\begin{equation}\label{eq:local-exclusion-alpha-star}
		E_N \ge C \alpha_N^2 (N-1)_+ \ge C \alpha_*^2 (N-1)_+ \quad \forall N,
	\end{equation}
	for an explicit numerical constant $C>0$.
	This provides a \keyword{local exclusion principle} 
	--- a secondary manifestation of statistical repulsion ---
	and it may be applied iteratively in a way which in the case of fermions 
	goes back already to
	Dyson and Lenard \cite{DysLen-67} in their ingenious
	proof of the thermodynamic stability of fermionic matter with Coulomb interaction 
	(see, e.g., \cite{LieSei-10,Lundholm-17}).
	Namely, the corresponding bound for fermions, 
	which follows immediately from \eqref{eq:Fermi-energy} 
	with $\lambda_k \ge \pi^2$ for $k \ge 1$, is
	$$
		E_N \ge \pi^2 (N-1)_+ \quad \forall N.
	$$
	It was shown in \cite{LunSol-13a,LunSol-14} that
	such a linear in $N$ bound \eqref{eq:local-exclusion-alpha-star} is sufficient to estimate the abelian anyon gas energy
	$$
		E_N \ge {\textstyle\frac{1}{4}} C \alpha_*^2 N^2 + o(N^2)
        \qquad
        \text{as $N \to \infty$.}
	$$
	The function $\alpha \mapsto \alpha_*$ appearing in these bounds 
	and supported on odd-numerator rationals 
	is a variant of the \keyword{Thomae} or \keyword{`popcorn' function}; 
	cf. Figure~\ref{fig:popcorn}.
	
\begin{figure}[t]
	\centering
	\scalebox{0.9}{
	\begin{tikzpicture}
		\node [above right] at (0,0) {\includegraphics[scale=0.65, trim=0.4cm 0cm 0cm 0cm]{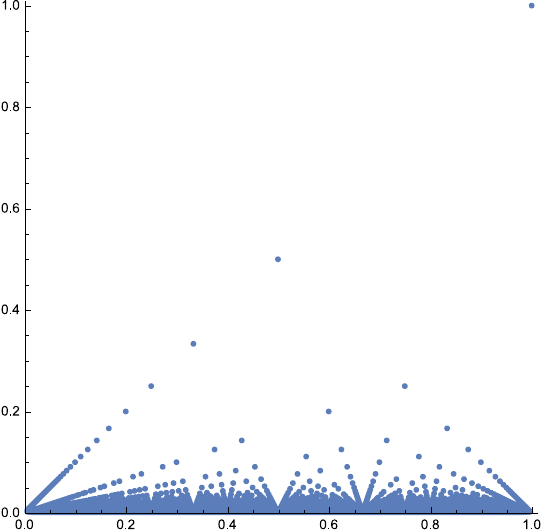}};
		\node [above right] at (5.70,0) {\scalebox{0.8}{$\alpha$\hspace{-20pt}}};
		\node [above right] at (-0.1,5.95) {\scalebox{0.8}{$\alpha_*$\hspace{-20pt}}};
	\end{tikzpicture}
	\begin{tikzpicture}
		\node [above right] at (0,0) {\includegraphics[scale=0.79]{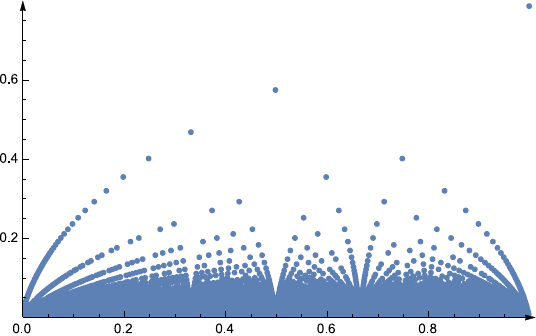}};
		\node [above right] at (7.23,.18) {\scalebox{0.8}{$\alpha$\hspace{-2pt}}};
		\node [above right] at (0,4.6) {\scalebox{0.8}{$f(j_{\alpha_*}'^2)$\hspace{-20pt}}};
		\node [above right] at (0.28,0.23) {\scalebox{6.8}[5.2]{\begin{tikzpicture}
			\draw[very thin,color=orange,domain=0:0.054,samples=50] plot (\x,{0.25*3.8*\x/sqrt(1+16.5*\x)});
			\draw[very thin,color=orange,domain=0.052:1] plot (\x,{0.25*0.147});
			\end{tikzpicture}}};
	\end{tikzpicture}
	}
	\caption{Left: The popcorn function $\alpha \mapsto \alpha_*$. 
		Right: A numerical lower bound to the function $\alpha \mapsto f(j_{\alpha_*}'^2)$ from \cite[Fig.~6]{LarLun-16}.
		The continuous orange curve indicates $\alpha \mapsto c(\alpha_2)$ for comparison (see text).}
	\label{fig:popcorn}
\end{figure}	
		
	These bounds were subsequently improved in works of 
	D.\,L.\ with Larson respectively Seiringer; 
	first in \cite{LarLun-16} considering a refinement of the
	many-anyon Hardy inequality to yield
	$$
		E_N \ge f\bigl(j_{\alpha_N}'^2\bigr) (N-1)_+,
	$$
	where the function $f\bigl(j_{\alpha}'^2\bigr) \simeq 4\pi\alpha$ 
	as $\alpha \to 0$ has an expected linear behavior for small $\alpha$
	(its definition and uniform bounds will be given below),
	and then with a scale-covariant method in \cite{LunSei-17} to
	$$
		E_N \ge c(\alpha_2) N, \qquad N \ge 2,
	$$
	where $c(\alpha_2) = \frac{1}{4}\min\left\{ f\bigl(j_{\alpha_2}'^2\bigr),0.147 \right\}$.
	The latter bound thus removes the dependence on $\alpha_*$, 
	although with a weaker constant; cf.\ Figure~\ref{fig:popcorn}.
	In fact, it was shown in \cite{LunSei-17} that
	the homogeneous ideal abelian anyon gas is subject to uniform linear bounds in 
	$\alpha \in [0,1]$
	(by periodicity and a conjugation symmetry $\Psi \mapsto \overline{\Psi}$, 
	$\alpha \mapsto -\alpha$,
	it suffices to study this range):
	
\begin{theorem}[{\keyword{Uniform bounds for the ideal abelian anyon gas} \cite{LunSei-17}}]\label{thm:intro-abelian}
	\mbox{}\\
	For any sequence of abelian $N$-anyon models $\rho_N\colon B_N \to \sU(1)$,
	$\rho_N(\sigma_j) = e^{i\alpha\pi}$, 
	with statistics parameters $\alpha = \alpha(N) \in [0,1]$
	and n-anyon exchange parameters $\alpha_n = \alpha_n(N) \in [0,1]$, 
	we have the following uniform bounds for the ground-state energy 
	$E_N = \infspec \hT_\alpha$ 
    (Neumann b.c. on $Q_0$, resp. $\bar{E}_N$ Dirichlet):
	\begin{equation}\label{eq:homogeneous-abelian}
		\frac{1}{4} C(\rho_N) N^2 \bigl(1 - O(N^{-1})\bigr)
		\le E_N \le \bar{E}_N 
		\le 2\pi^2 N^2 \bigl(1 + O(N^{-1/2})\bigr),
	\end{equation}
	where
	$$
		C(\rho_N) := \max \left\{ C_4(\rho_N), f\bigl( j_{\alpha_N}'^2 \bigr) \right\},
		\qquad  \nu/3 \le f( j_{\nu}'^2 ) \le 4\pi\nu(1+\nu)
		\ \ \forall \nu \in [0,1],
	$$
	and
	$$
		C_4(\rho_N) := \frac{1}{4} \min\{E_2,E_3,E_4\}
		\ \ge \ c(\alpha) := \frac{1}{4}\min\left\{ f(j_{\alpha}'^2),0.147 \right\}.
	$$
	Further, 
	there exist universal constants $C_2 \ge C_1 > 0$ such that,
	if $\alpha$ is independent of $N$,
	$$
		C_1 \alpha \le \liminf_{N \to \infty} E_N/N^2 
		\le \limsup_{N \to \infty} \bar{E}_N/N^2 \le C_2 \alpha.
	$$
\end{theorem}

	The latter limits are the energy per particle and density;
	thus, the abelian anyon gas with $\alpha \in (0,1]$
	exhibits extensivity in the energy like the Fermi gas \eqref{eq:Fermi-energy} 
	(and unlike the Bose gas \eqref{eq:Bose-energy}),
	as a consequence of the statistical repulsion between pairs of particles.
	
	Another development in \cite{LunSol-13a,LarLun-16,LunSei-17}
	concerns the \keyword{Lieb--Thirring inequality}:
	\begin{equation}\label{eq:intro-LT-kinetic}
		\langle \Psi,\hT_\alpha \Psi\rangle
		\ge C \alpha \int_{\R^2} \varrho_\Psi(\bx)^2 \, d\bx\,,
		\qquad \alpha \in [0,1],
	\end{equation}
	where the \keyword{one-body density} $\varrho_\Psi$ is the marginal of the probability distribution,
	\begin{equation}\label{eq:def-density-R2}
		\varrho_\Psi(\bx) := 
		\sum_{j=1}^N \int_{\R^{2(N-1)}} |\Psi(\bx_1, \ldots, 
		\bx_{j-1}, \bx, \bx_{j+1}, \ldots, \bx_N)|^2 \prod_{k \neq j}d\bx_k.
	\end{equation}
	Such an inequality was first proved for fermions by Lieb and Thirring \cite{LieThi-75,LieThi-76}
	with the goal of simplifying Dyson and Lenard's proof of stability of matter.
	It is a powerful combination of both the uncertainty principle and the exclusion
	principle, and manifests the tertiary form of statistical repulsion as a 
	degeneracy pressure in the density. Namely the r.h.s.\ of \eqref{eq:intro-LT-kinetic}
	is on the form of the \keyword{Thomas--Fermi approximation} \cite{Thomas-27,Fermi-27}
	\begin{equation}\label{eq:TF-approx}
        \inf_{\text{bosonic}\,\Psi : \int_{\R^{2N}}|\Psi|^2=1} \langle \Psi,\hT_\alpha\Psi \rangle
		\approx \inf_{\varrho \ge 0 : \int_{\R^2}\varrho=N} \int_{\R^2} 2\pi\alpha \varrho(\bx)^2 \,d\bx,
	\end{equation}
	which employs the extensive ideal Fermi gas energy locally at $\bx \in \R^2$
	(here with a mean-field motivated interpolation guess for intermediate $\alpha \in (0,1)$;
	cf.\ \cite{Sen-91,ChiSen-92,LiBhaMur-92,CorDubLunRou-19}).
	In \cite{LunSol-13a} it was observed that the local exclusion principle 
	\eqref{eq:local-exclusion-alpha-star} may be
	used to prove the Lieb-Thirring inequality, with a constant involving $\alpha_*^2$, 
	and the successive improvements in \cite{LarLun-16,LunSei-17}
    (see also \cite{RouYan-24})
	for the local energy led to corresponding improvements
	of the constant in \eqref{eq:intro-LT-kinetic}.
	Furthermore, we should mention that
	a possibility of minimizing the energy due to statistical repulsion 
	using certain clustering states (in a sense generalizing 
	\eqref{eq:product-state} and \eqref{eq:Slater-state} for particular $\alpha$)
	was discussed in \cite{LunSol-13b,Lundholm-16}, 
	although this still remains an open problem (see Remark~\ref{rem:Hardy-opt}).
	
	Another line of approach to understanding the many-anyon spectrum is to
	consider the limits $\alpha \to 0$ resp. $\alpha \to 1$, known as \keyword{almost-bosonic} resp. \keyword{almost-fermionic} anyons.
	In the corresponding mean-field (or \keyword{`average-field'}) theory 
	(see, e.g., \cite{CheWilWitHal-89,Wilczek-90,IenLec-92})
	of weakly magnetically interacting bosons, the Thomas--Fermi approximation \eqref{eq:TF-approx}
	has been rigorously justified in a particular limit via regularized 
	(extended, i.e.\ non-ideal) anyons
	\cite{LunRou-15,CorLunRou-16,CorLunRou-proc-17,Girardot-20,AtaLunNgu-24,AtaGirLun-25,Lundholm-24}.
    However, it brings a few surprises, such as a slightly bigger  
	constant than $2\pi$
	due to self-interactions and the emergence of a vortex lattice \cite{CorDubLunRou-19,Ataei-etal-25}.
    The almost-fermionic limit has also been made rigorous \cite{GirRou-21,GirRou-23}, and conjecturally extended to an interesting fractional dependence on $\alpha$ \cite{LevLunRou-25}.
	Note that this quantitative perturbative approach is typically not possible for non-abelian models
	(at least not the algebraic anyon models studied here,
	whose associated exchange phases are distributed at discrete 
	positions on the unit circle), and therefore it will not be our focus in this work.

\subsection{Main new results} 

	In this work we initiate a
    mathematically rigorous study of the ground-state properties
	of the ideal \emph{non}-abelian many-anyon gas, taking $N \to \infty$.
	Note that two-particle energies, second virial coefficients, and other pairwise statistics-dependent 
	properties have already been investigated in the past for certain non-abelian models, notably
	\cite{Verlinde-91,LeeOh-94} as well as \cite{ManTroMus-13a,ManTroMus-13b}
	for non-abelian Chern--Simons (NACS) particles,
	and these could be argued to be dominant for the \emph{dilute} gas.
	
	After dealing with the basic tasks of defining an $N$-anyon model in sufficient generality, 
	its kinetic energy operator and ground-state energy,
	the main aim is to find a non-abelian equivalent of the statistical repulsion
	to generalize the above results for abelian anyons.
	Again, we choose to take the route via a pairwise Poincar\'e inequality 
	and a many-anyon Hardy inequality,
	however this requires first to define and compute for an arbitrary sequence of anyon 
	models
	$\rho_N\colon B_N \to \sU(\cF_N)$ the
	two-anyon \keyword{exchange operators} $U_p$, 
	which are the unitary matrices in \eqref{eq:exchange-op}
	given that $p$ other anyons are encircled in the exchange of two anyons.
	It is a consequence of the uncertainty principle that 
	$U_0 \sim \rho_N(\sigma_j)$ does not suffice.
	
	Our first main observation, Theorem~\ref{thm:gen-exchange}, concerns the reduction of
	$U_p$ in any algebraic anyon model to an exchange involving only three objects,
	thus giving a recipe on how to compute it in general.
	We illustrate this by computing explicit exchange operators $U_p$ 
	and phases (eigenvalues $e^{i\beta\pi}$) for a 
	few models of special interest,
	including Fibonacci and Ising models, as well as the simplest non-trivial Burau model
	(which is not of the same algebraic type).
	The resulting data of interest concerning the exchange properties of the model
	are summarized by the \keyword{exchange parameters} for $p$ \emph{enclosed} particles:
	$$
		\beta_p := \min \{ \beta \in [0,1] : \text{$e^{i\beta\pi}$ or $e^{-i\beta\pi}$ is an eigenvalue of $U_p$} \},
		\quad p \in \{0,1\ldots,N-2\},
	$$
	as well as for $n \in \{2,3,\ldots,N\}$ \emph{involved} particles (worst case of up to $n-2$ enclosed):
	$$
		\alpha_n := \min_{p \in \{0,1,2,\ldots,n-2\}} \beta_p.
	$$

	In Section~\ref{sec:ham-space} we define the kinetic energy operator $\hT_\rho$ for an
	arbitrary representation $\rho$ \eqref{eq:braid-rep}.
	We then extend methods of local statistical repulsion 
	--- in all three of its manifestations ---
	to general anyon models. 
	In particular, we prove a many-anyon Hardy inequality
	\begin{equation}\label{eq:anyon-Hardy-nonab}
		\hT_\rho \ge \frac{4}{N} \max\left\{ \alpha_N^2, \frac{\alpha_2^2}{N-1} \right\} \sum_{1 \le j < k \le N} |\bx_j-\bx_k|^{-2},
	\end{equation}
	given in Theorem~\ref{thm:Hardy} in 
	a stronger but more technical local version.
	We also discuss a few counterexamples to such an
	inequality with any positive constant for models with $\alpha_2=0$.
	A local exclusion principle is obtained in Lemma~\ref{lem:local-exclusion},
	and as a consequence
	uniform lower as well as upper bounds for the homogeneous non-abelian anyon gas 
	ground-state energy are given in Theorem~\ref{thm:homogeneous-gas}
	and its Corollaries~\ref{cor:Fibonacci-gas}-\ref{cor:Clifford-gas}. 
	We summarize:
	
\begin{theorem}[{\keyword{Uniform bounds for the ideal non-abelian anyon gas}}]\label{thm:intro-nonabelian}
	\mbox{}\\
	For any sequence of $N$-anyon models $\rho_N\colon B_N \to \sU(\cF_N)$ 
	with n-anyon exchange parameters $\alpha_n = \alpha_n(N) \in [0,1]$,
	$2 \le n \le N$,
	we have the uniform bounds
	\begin{equation}\label{eq:homogeneous-nonabelian}
		\frac{1}{4} C(\rho_N) N^2 \bigl(1 - O(N^{-1})\bigr)
		\le E_N \le \bar{E}_N 
		\le 2\pi^2 N^2 \bigl(1 + O(N^{-1/2})\bigr),
	\end{equation}
	where
	$$
		C(\rho_N) := \max \left\{ C_4(\rho_N), f(j_{\alpha_N}'^2) \right\}
	$$
	and
	$$
		C_4(\rho_N) := \frac{1}{4} \min\{E_2,E_3,E_4\}
		\ge c(\alpha_2)
		\ge \frac{1}{4}\min\left\{\alpha_2/3,0.147\right\}.
	$$
	In particular:
	\begin{itemize}
	\item
	For Fibonacci anyons, the exchange parameters are $\alpha_2 = \beta_0 = 3/5$
	and $\alpha_N = \beta_1 = 1/5$ for $N \ge 3$,
	and hence \eqref{eq:homogeneous-nonabelian} holds with $C(\rho_N) \ge 1/15$ 
	(numerically $C(\rho_N) \gtrsim 0.35$) for $N \ge 3$.
	\item
	For Ising anyons, the exchange parameters are 
	$\alpha_N = \beta_0 = 1/8$ for all $N \ge 2$,
	and hence \eqref{eq:homogeneous-nonabelian} holds with $C(\rho_N) \ge 1/24$ 
	(numerically $C(\rho_N) \gtrsim 0.25$) for $N \ge 2$.
	\end{itemize}
\end{theorem}

    Although these bounds are probably far from optimal, they have a conceptual interest in that there is at least some explicit analytical control on the strength of the effective statistical repulsion, and that a potential difference may manifest comparing to the abelian case and between different non-abelian models.
    
	The inhomogeneous gas is covered by the Lieb--Thirring inequality, Theorem~\ref{thm:LT},
	which generalizes \eqref{eq:intro-LT-kinetic} to
    \begin{equation}\label{eq:LT-nonabelian}
		\langle \Psi,\hT_\rho \Psi\rangle
		\ge C \alpha_2 \int_{\R^2} \varrho_\Psi(\bx)^2 \, d\bx\,,
    \end{equation}
	for a universal constant $C>0$.
	An immediate application is the thermodynamic 
	stability of any Coulomb-interacting (2D embedded in 3D, and otherwise ideal) anyon gas 
	for which $\alpha_2(N) > 0$ is bounded below uniformly in $N$; 
	cf.\ \cite[Theorem~21]{LunSol-14} and \cite[Sec.~7.3.3]{Lundholm-17}.

\begin{figure}
  \centering
	\begin{tikzpicture}
		\draw [red!60!white] (0,0) 	ellipse (1.5 and 0.7);
		\node at (-4,0) {\large\bf algebraic};
		\node at (0,0) {\large\bf geometric};
		\node at (4,0) {\large\bf magnetic};
		\draw [arrows=->,thick] (-2.8,0) -- (-1.8,0);
		\draw [arrows=->,thick] (2.8,0.2) -- (1.8,0.2);
		\draw [arrows=->,thick,dashed] (1.8,-0.2) -- (2.8,-0.2);
		\node at (0,-1.3) {$\rho\colon B_N \to \sU(D)$};
		\node [darkred] at (0,1.5) {\small most general};
		\draw [arrows=->,darkred] (0,1.2) -- (0,0.8);
		\node [darkred] at (4,1.5) {\small most realistic};
		\draw [arrows=->,darkred] (4,1.2) -- (4,0.4);
		\node [darkred] at (-4,1.5) {\small most computational};
		\draw [arrows=->,darkred] (-4,1.2) -- (-4,0.4);
		\node at (3,-0.7) {\small statistics transmutation};
	\end{tikzpicture}
  \caption{Anyon models.}
  \label{fig:anyon-models}
\end{figure}

    \smallskip
    
	Due to the necessity in this field
	of employing tools from all corners of mathematics,
	we follow the style of a review and will introduce the necessary concepts as they arise.
	We will go from algebraic to geometric and magnetic models (see Figure~\ref{fig:anyon-models}) and,
	via functional inequalities of Poincar\'e, Hardy and Lieb--Thirring, 
	finally to the `physical/thermodynamic' application.
	In Section~\ref{sec:models} we recall the essentials of algebraic anyon models,
	which include models considered in FQHE contexts and of relevance to 
	topological quantum computing.
	Section~\ref{sec:phases} is devoted to defining and computing general exchange operators
	and phases for these models of interest.
	Previous work for abelian anyons has relied on the existence of a
	magnetic representation, however as we will discuss in Section~\ref{sec:ham-transmutation},
	it is not obvious that such a representation should exist for all
	non-abelian models. Therefore we take a different approach via covering
	spaces, as outlined e.g.\ in \cite{FroMar-89,MunSch-95}. 
	The technical machinery for this is recalled in 
	Sections~\ref{sec:ham-cspace}-\ref{sec:ham-space} and \ref{sec:ham-energy}.
	In Section~\ref{sec:repulsion} statistical repulsion is discussed,
	based on our generalization \eqref{eq:anyon-Hardy-nonab}
	of the many-anyon Hardy inequality for an arbitrary geometric anyon model.
	Finally, in Section~\ref{sec:gas} these tools are applied to prove bounds
	for the homogeneous anyon gas energy as well as the potentially inhomogeneous gas 
	by means of the Lieb--Thirring inequality.
	
	We use boldface to help the reader note whenever a new key concept arises.
	Although we have tried to uniformize and streamline the various notations found in the 
	literature, beware e.g.\ that the popular symbol $\sigma$ has multiple meanings 
	(spectrum, anyon, braid, $\ldots$).
	
	The results in this work are based in part on the MSc thesis of V.Q.\ \cite{Qvarfordt-17}.
    A first version of the manuscript appeared already in 2020, 
    although publication of the work was subsequently delayed 
    due to the Covid-19 global pandemic
    and V.Q.\ having left academia.

\medskip\noindent\textbf{Acknowledgments.} 
D.L.\ gratefully acknowledges financial support from the
Swedish Research Council 
(grant no.\ {2013-4734}, ``Spectral theory of quantum systems with exotic symmetries'';
grant no.\ {2021-05328}, ``Mathematics of anyons and intermediate quantum statistics'')
and the G\"oran Gustafsson Foundation (grant no.\ 1804).
Further, D.L.\ thanks Gustav Brage, Erik S\"onnerlind, Oskar Weinberger and Joel Wiklund
for fruitful discussions on braid group representations
in their BSc thesis projects at KTH Engineering Physics 
\cite{Weinberger-15,BraSon-18,Wiklund-18}.
We have also benefited significantly from discussions with 
John Andersson,
Eddy Ardonne, 
Tilman Bauer, 
Alexander Berglund,
Michele Correggi,
Gerald Goldin,
Hans Hansson,
Simon Larson, 
Jon Magne Leinaas, 
Tomasz Maciazek, 
Dan Petersen, 
Nicolas Rougerie,
Adam Sawicki, 
Robert Seiringer,
Jan Philip Solovej,
Yoran Tournois,
and
Susanne Viefers.
Finally, D.L.~thanks
the Department of Mathematics of Politecnico di Milano for its kind hospitality 
during the spring of 2025 intensive period ``Quantum Mathematics at Polimi'',
as well as Ask Ellingsen for useful discussions and comments on the manuscript.
The authors also wish to thank an anonymous referee for comments which have 
helped to clarify and improve our manuscript.


\section{Algebraic anyon models}\label{sec:models}

	A potential confusion for a newcomer to the vast literature on anyons is that, 
	while a choice of an \emph{ideal} $N$-anyon model is mathematically equivalent to a 
	choice of a braid group representation \eqref{eq:braid-rep}, 
	as we shall make precise in Section~\ref{sec:ham}, for any physical realization 
	it makes all the difference in the world whether
	the corresponding particles are treated as classical or quantum mechanical.
	That is, one might think of localizing the positions of the particles to a specific
	set of points and then perform braiding operations on them as if they were classical objects.
	This is convenient from the practical perspective of quantum computation
	and is sometimes the way that anyons are portrayed (abusing the terminology).
	We stress however that an anyonic object is a consequence of \keyword{identity} at a 
	fundamental ``quantum-mathematical'' level and can never be classical, 
	just like bosons and fermions cannot.
	Namely,
    the uncertainty principle forbids us to talk about precise positions 
	unless we have access to infinite energy,
	and demands that we consider in parallel all possibilities that cannot be logically excluded.
	One way to clarify this difference in perspectives is on one hand
	to consider a realization of a braid group representation 
	(such as an algebraic anyon model or a FQHE quasihole ansatz) 
	as a \emph{kinematical} framework,
	and on the other to require an associated
	geometric or magnetic anyon model in which also the \emph{dynamics} 
	(actual quantum Hamiltonian, Hilbert space etc.) of the particles --- actual anyons ---
	is specified (see \cite{Forte-91,LunRou-16,LamLunRou-22,Lundholm-23} for a related discussion).

	In more concrete terms, the algebraic data tell us how an internal state changes
	\emph{conditionally on a prescribed braid}: the particle positions are treated as
	external parameters and the braid is an operation performed on the corresponding
	fusion space.  A dynamical model instead promotes the positions to quantum
	variables.  Its wave function must therefore be defined simultaneously over all
	possible configurations, with the braid representation prescribing how the different local
	descriptions are glued together.  The kinetic-energy operator then determines
	which such wave functions have finite energy and how they evolve.  Thus the same
	braid representation may be used either as a kinematical gate set for pinned
	quasiparticles or as the statistics input of a mobile many-body system; the
	distinction is not in the braid matrices themselves, but in whether particle
	positions and their dynamics are part of the model.
	Our perspective is that anyons are defined in the latter sense.
	Hence, in this section we start with the kinematical description 
	and continue with the dynamics in later sections.
	The relevance of the uncertainty principle for the origin of the 
	exclusion principle for anyons has been discussed in some further detail in 
	\cite{EllLunMag-23,Lundholm-24}.
	
	For an \keyword{algebraic anyon model} we specify the following information:
	\begin{enumerate}[label=\textup{\arabic*.}]
	\item Particle types / labels / topological charges.
	\item Rules for fusion / splitting.
	\item Rules for braiding.
	\end{enumerate}
	
	Intuitively, \keyword{fusion} takes into account how any subsystem 
	of particles looks from afar,
	and is convenient when there is a flexible or large number of particles.
	The fully consistent machinery of fusion and braiding, 
	which was initially developed 
	in the context of CFT and Chern--Simons theory \cite{Witten-89,MooSei-89} 
	as well as algebraic QFT \cite{FreRehSch-89,FroGab-90},
	led eventually to the notion of a \keyword{unitary braided fusion category} 
	/ \keyword{modular tensor category} \cite[Appendix~E]{Kitaev-06} \cite{RowStoWan-09}.
	This is a highly technical subject and, luckily, over the years it has been simplified
	and destilled through a number of excellent lecture notes, reviews and theses
	into a form which is accessible to an average humanoid mathematical physicist.
	We have attempted to continue this effort and 
	give here only a very brief summary of the parts that are necessary for our purposes.
	Thus we follow the notations of 
	\cite{Preskill-04,Kitaev-06,Bonderson-07,Nayak-etal-08,Trebst-etal-08,FanGar-10,Wang-10,Pachos-12,DelRowWan-16,Simon-16,Tong-16,Garjani-17,FieSim-18}
	(although note differences in conventions!).

\subsection{Fusion}\label{sec:models-fusion}

\subsubsection{Fusion rules}

	We start with a finite set of \keyword{labels} $\cL = \{a,b,c,\ldots\}$, 
	whose elements are interpreted as the different \keyword{types of anyons} 
	(or the \keyword{topological charges})
	present in the model. 
	It must always contain a special element $1 \in \cL$ called the \keyword{vacuum}.
	The elements of $\cL$ are assumed to generate a 
	commutative, associative fusion algebra with unit $1$.
	This means that we may take formal linear combinations of labels and 
	have a fusion rule, i.e.\ a binary composition $\times$ 
	such that for $a,b \in \cL$,
	\begin{equation}\label{eq:fusion-algebra}
		a \times b = \sum_{c \in \cL} N_{ab}^c \,c,
	\end{equation}
	where the nonnegative integers $N_{ab}^c \in \N_0$ are called \keyword{fusion multiplicities}.
	If $N_{ab}^c \neq 0$
	we write $c \in a \times b$ and say that $c$ is a valid result of fusion,
	and $N_{ab}^c$ then counts the number of distinguishable ways that this type of fusion may occur,
	with each one of them called a \keyword{fusion channel}.
	The model is formally\footnote{This terminology is justified in \cite{RowWan-16}.} 
	called \keyword{non-abelian}
	if there exists some $a$ and $b$ in $\cL$ such that
	$$
		\sum_{c \in \cL} N_{ab}^c \ge 2,
	$$
	i.e.\ if there are several possible 
	ways to fuse $a$ and $b$, and otherwise, 
	i.e.\ if both the product and process of fusion are always unique, it is \keyword{abelian}.
	The typical models considered in the literature are actually \keyword{multiplicity free}, 
	i.e.\ $N_{ab}^c \in \{0,1\}$ for all $a,b,c \in \cL$,
	but instead there may be several 
	possible fusion products
	$c \in a \times b$, each defining a unique fusion channel.
	
	To each anyon type $a \in \cL$ is associated a \keyword{quantum dimension} 
	$d_a \ge 1$ 
	obtained as the largest
	eigenvalue of the matrix $N_a = [N_{ab}^c]_{b,c \in \cL}$.
	By the Perron-Frobenius theorem,
	these can also be found by 
	replacing \eqref{eq:fusion-algebra} 
	by the corresponding system of polynomial equations:
	\begin{equation}\label{eq:quantum-dimensions}
		d_a d_b = \sum_{c \in \cL} N_{ab}^c \,d_c.
	\end{equation}
	Furthermore, the fusion algebra \eqref{eq:fusion-algebra} 
	has the property that to each $a \in \cL$
	there is a unique inverse element, or \keyword{charge conjugate}, 
	$\bar{a} \in \cL$
	such that $1 \in a \times \bar{a}$.
	In particular, $\bar{1} = 1$.
	Also note that $1 \times a = a = a \times 1$ implies
	$N^c_{1a} = N^c_{a1} = \delta^c_a$.

\subsubsection{Fusion diagrams and spaces}

	Given a possible fusion $c \in a \times b$, i.e. $N_{ab}^c \neq 0$, 
	we construct a Hilbert space
	$V_{ab}^c$ by assigning an element $\langle a,b;c,\mu|$ of an abstract 
	orthonormal basis to each fusion channel, enumerated by an index
	$\mu \in \{1,2,\ldots,N_{ab}^c\}$:
	$$
		V_{ab}^c := \Span_\C \{ \langle a,b;c,\mu| \}_\mu.
	$$
	The dual space to this \keyword{fusion space} is called a \keyword{splitting space}:
	$$
		V^{ab}_c := \Span_\C \{ |a,b;c,\mu\rangle \}_\mu.
	$$
	We picture each fusion resp.\ splitting channel basis state 
	diagrammatically, with time running upwards, as:
	\begin{equation}
	    \begin{tikzpicture}[scale=0.35,font=\footnotesize,anchor=mid,baseline={([yshift=-.5ex]current bounding box.center)}]
	      \node at (0.7, -2.7) {$a$};
	      \node at (2.3, -2.7) {$b$};
	      \draw[->] (0.7, -2.3) to [bend right=-20] (1.25, -1.35);
	      \draw[->] (2.3, -2.3) to [bend right=20] (1.75, -1.35);
	      \draw (1.5, -1) circle (0.45);
	      \node at (1.503, -0.9) {$\mu$};
	      \draw[->] (1.5, -0.55) to (1.5, 0.6);
	      \node at (2, 0.05) {$c$};
	    \end{tikzpicture}
	    \qquad \text{respectively} \qquad
	    \begin{tikzpicture}[scale=0.35,font=\footnotesize,anchor=mid,baseline={([yshift=-.5ex]current bounding box.center)}]
	      \node at (0.7, 2.7) {$a$};
	      \node at (2.3, 2.7) {$b$};
	      \draw[<-] (0.7, 2.3) to [bend left=-20] (1.25, 1.35);
	      \draw[<-] (2.3, 2.3) to [bend left=20] (1.75, 1.35);
	      \draw (1.5, 1) circle (0.45);
	      \node at (1.503, 1.08) {$\mu$};
	      \draw[<-] (1.5, 0.55) to (1.5, -0.6);
	      \node at (2, -0.05) {$c$};
	    \end{tikzpicture}
		\qquad \text{or simply} \qquad
		\fs{b}{a,c}.
	\end{equation}
	In the last diagram we have suppressed the channel index $\mu$
	(anticipating multiplicity-free models)
	and the arrows,
	and allow time to flow from down/right to up/left
	(this allows for tremendous typographical simplifications upon composing diagrams).
	If one prefers, one may turn the diagrams (or time) around 
	and replace fusion spaces by splitting spaces, and vice versa.
	Beware that the conventions of the literature are mixed.
	Also note that by the properties of the inverse/conjugate/unit 
	we have the simple diagrams
	\begin{equation}
		\fs{1}{a,a}, \qquad
		\fs{a}{1,a}, \qquad
		\fs{\bar{a}}{a,1}, \qquad
		\fs{a}{\bar{a},1}.
	\end{equation}
	
	Fusions/splittings may be diagrammatically composed into more complicated processes;
	a typical one is shown in Figure~\ref{fig:fusions}, left.
	In case a single charge $d \in \cL$ splits consecutively into charges $a,b,c$,
	there are \emph{two} distinct ways to represent this by diagrams:
	\begin{equation}\label{eq:double-fusion-diagram}
		\fs{b,c}{a,e,d}
		\qquad \text{respectively} \qquad
		\fsfused{a}{b}{c}{d}{e},
	\end{equation}
	with corresponding composed splitting spaces defined by
	\begin{equation}\label{eq:fusion-space-decomp-3}
		V^{abc}_d := \bigoplus_{e} V^{ab}_e \otimes V^{ec}_d
		\qquad \text{respectively} \qquad
		\tilde{V}^{abc}_d := \bigoplus_{e} V^{ae}_d \otimes V^{bc}_e.
	\end{equation}
	Note that we must sum over all possible intermediate charges 
	$e \in a \times b$ resp. $e \in b \times c$,
	and that there could also be a multiplicity $N_{ab}^c = \dim V^{ab}_c$
	etc. in each splitting process.
	The diagrams \eqref{eq:double-fusion-diagram} 
	(possibly with additional suppressed labels 
	$\mu \in \{1,\ldots,N_{ab}^e\}$, $\nu \in \{1,\ldots,N_{ec}^d\}$, etc.)\  
	represent the basis elements of the corresponding space.
	All bases are naturally ordered according to a choice of order of the elements in $\cL$.
	
	In the case of further successive splitting we define the standard splitting space
	\begin{equation}\label{eq:fusion-space-decomp-n}
	\begin{aligned}
		V^{a_1 a_2 \cdots a_n}_c 
		&:= \bigoplus_{b_1,b_2,\ldots,b_{n-2}} V^{a_1a_2}_{b_1} \otimes V^{b_1 a_3}_{b_2} \otimes V^{b_2 a_4}_{b_3} \otimes \ldots \otimes V^{b_{n-2} a_n}_c \\
	  &= \left\{
	  \begin{tikzpicture}[scale=0.3,font=\footnotesize,anchor=mid,baseline={([yshift=-.5ex]current bounding box.center)}]
	    \node at (0, -0.6) {$a_1$};
	    \node at (1, 2) {$a_2$};
	    \node at (3, 2) {$a_3$};
	    \node at (10, 2) {$a_{n-1}$};
	    \node at (13, 2) {$a_n$};
	    \draw (1, 0) to (1, 1.5);
	    \draw (3, 0) to (3, 1.5);
	    \draw (10, 0) to (10, 1.5);
	    \draw (13, 0) to (13, 1.5);
	    \draw (-1, 0) to (5, 0);
	    \draw (7, 0) to (15, 0);
	    \node at (2, -0.7) {$b_1$};
	    \node at (4, -0.7) {$b_2$};
	    \node at (6, 0.2) {$\cdots$};
	    \node at (8.5, -0.7) {$b_{n-3}$};
	    \node at (11.5, -0.7) {$b_{n-2}$};
	    \node at (14, -0.7) {$c$};
	  \end{tikzpicture}
	  \;\;\;\bigg|\; \begin{array}{c}\text{for all possible intermediate} \\ \text{charges $b_1, b_2, \ldots, b_{n-2} \in \cL$}\end{array}
	  \right\}.
	\end{aligned}
	\end{equation}
	Again, in the case of nontrivial multiplicities 
	one must also consider all intermediate
	labels $\mu_1, \ldots, \mu_{n-2}$ of the respective spaces.
	Also note that by fusion with the vacuum we may always shift
	the diagrammatic representation:
	\begin{equation}
		\fswide{a_2,a_3}{a_1,b_1,b_2} \ldots = \fswide{a_1,a_2,a_3}{1,a_1,b_1,b_2} \ldots
	\end{equation}

\subsubsection{F-symbols: Associativity of fusion}

	The requirement of associativity of the fusion algebra,
	\begin{equation}\label{eq:F-assoc}
		(a \times b) \times c = a \times (b \times c),
	\end{equation}
	enforces natural isomorphisms on the fusion spaces
	$V^{abc}_d$ and $\tilde{V}^{abc}_d$.

\begin{definition}[F operator]\label{def:F}
	The \keyword{$F$ operator} (or \keyword{$F$ matrix})
	$F^{abc}_d\colon V^{abc}_d = \bigoplus_e V^{ab}_e \otimes V^{ec}_d \to \tilde{V}^{abc}_d = \bigoplus_e V^{ae}_d \otimes V^{bc}_e$ 
	is an isomorphism, diagrammatically represented by
	\begin{equation}\label{eq:F-matrix}
		F\colon \fs{b,c}{a,e,d} \mapsto \fsfused{a}{b}{c}{d}{e} = \sum_f F^{abc}_{d;fe} \fs{b,c}{a,f,d}.
	\end{equation}
	In the case of nontrivial fusion multiplicities one may write
	\begin{equation}\label{eq:F-matrix-mult}
		|a,e; d,\alpha\rangle \otimes |b,c; e,\beta\rangle
		= \sum_{f,\mu,\nu} \left[ F^{abc}_{d} \right]_{(f,\mu,\nu),(e,\alpha,\beta)}
		|a,b; f,\mu\rangle \otimes |f,c; d,\nu\rangle,
	\end{equation}
	where the components $F^{abc}_{d;fe} = [F^{abc}_d]_{f,e}$, 
	expressed in the standard bases, are known as the \keyword{F-symbols}.
	Canonical isomorphisms are represented trivially.
\end{definition}

	The following lemma will be useful when computing the $F$-symbols:

\begin{lemma}\label{lem:F1}
  Consider the splitting space $V^{abc}_d$ in a multiplicity-free model.
  When one of the particle types is trivial, 
  i.e.\ $a,b,c$ or $d$ equals $1$, 
  then $\dim V^{abc}_d = 1$, and furthermore if $a,b$ or $c$ equals $1$, 
  then the corresponding $F$-matrix $F^{abc}_d$ is trivial. 
  Explicitly that is
  \begin{equation}
    \begin{aligned}
      F^{1bc}_d &= F^{1bc}_{d;\, bd} = 1, \\
      F^{a1c}_d &= F^{a1c}_{d;\, ac} = 1, \\
      F^{ab1}_d &= F^{ab1}_{d;\, db} = 1.
    \end{aligned}
  \end{equation}
\end{lemma}
\begin{proof}
	For example, consider $a=1$:
	\begin{equation}
		V^{1bc}_d = V^{1b}_b \otimes V^{bc}_d \cong V^{bc}_d,
		\qquad
		\tilde{V}^{1bc}_d = V^{1d}_d \otimes V^{bc}_d \cong V^{bc}_d,
	\end{equation}
	where the intermediate labels are fixed by the constraints on fusion,
	and all spaces are one-dimensional and canonically isomorphic 
	on a single fusion/splitting channel by the property of the unit $1$.
	Thus, diagrammatically,
	\begin{equation}
	    \fs{b,c}{1,e,d} = \fs{b,c}{1,b,d} = \fsfused{1}{b}{c}{d}{d},
	\end{equation}
	which is exactly the relation given in Definition~\ref{def:F},
	with $F^{1bc}_{d;bd} = 1$.
	The case $b=1$ and $c=1$ is similar. For $d=1$,
	\begin{equation}
		V^{abc}_1 = V^{ab}_{\bar{c}} \otimes V^{\bar{c}c}_1,
		\qquad
		\tilde{V}^{abc}_1 = V^{a\bar{a}}_1 \otimes V^{bc}_{\bar{a}},
	\end{equation}
	and again all spaces are one dimensional,
	however, they may involve different fusion/splitting channels.
\end{proof}

\subsection{Braiding}\label{sec:models-braiding}

Now that we have defined the fusion spaces, we may consider braiding operators
on these spaces.
The simplest cases define the R and B symbols.

\subsubsection{R-symbols: Commutativity of fusion}

	One requires that
	the result of fusing $a$ with $b$ must be the same as fusing $b$ with $a$. 
	That is, fusion is commutative,
	\begin{equation}
	  a \times b = b \times a.
	\end{equation}
	This gives rise to a natural isomorphism between the corresponding 
	fusion/splitting spaces for each possible fusion channel.

\begin{definition}[R operator]\label{def:R}
  The \keyword{$R$ operator} $R^{ab}$ is an isomorphism on each fusion channel 
  $c \in a \times b$
  \begin{equation}\label{R-matrix}
    R^{ab} : V^{ba}_c \to V^{ab}_c,
  \end{equation}
  or in terms of the basis states, unitary matrices $R^{ab}_c \in \sU(N_{ab}^c)$,
  \begin{equation}\label{R-matrix-map}
    R^{ab}\colon |b,a; c,\mu\rangle \mapsto \sum_\nu [R^{ab}_c]_{\nu\mu} |a,b; c,\nu\rangle,
  \end{equation}
  diagrammatically represented by 
  \begin{equation}
    R^{ab} :
    \begin{tikzpicture}[scale=0.35,font=\footnotesize,anchor=mid,baseline={([yshift=-.5ex]current bounding box.center)}]
      \node at (0.7, 2.7) {$a$};
      \node at (2.3, 2.7) {$b$};
      \draw (0.7, 2.3) to [bend left=-20] (1.25, 1.35);
      \draw (2.3, 2.3) to [bend left=20] (1.75, 1.35);
      \draw (1.5, 1) circle (0.45);
      \node at (1.503, 1.08) {$\mu$};
      \draw (1.5, 0.55) to (1.5, -0.6);
      \node at (2, -0.05) {$c$};
    \end{tikzpicture}
    \mapsto
    \begin{tikzpicture}[scale=0.3,font=\footnotesize,anchor=mid,baseline={([yshift=-.5ex]current bounding box.center)}]
      \node at (0.7, 4.9) {$a$};
      \node at (2.3, 4.9) {$b$};
      \braid[width=1.6cm, height=2cm] at (0.7, 4.5) s_1^{-1};
      \draw (0.7, 2) to [bend left=-30] (1.25, 1.2);
      \draw (2.3, 2) to [bend left=30] (1.75, 1.2);
      \draw (1.5, 0.85) circle (0.45);
      \node at (1.503, 0.9) {$\mu$};
      \draw (1.5, 0.4) to (1.5, -0.75);
      \node at (2, -0.15) {$c$};
    \end{tikzpicture}
    = \sum_{\nu} \left[ R^{ab}_c \right]_{\nu\mu}
    \begin{tikzpicture}[scale=0.35,font=\footnotesize,anchor=mid,baseline={([yshift=-.5ex]current bounding box.center)}]
      \node at (0.7, 2.7) {$a$};
      \node at (2.3, 2.7) {$b$};
      \draw (0.7, 2.3) to [bend left=-20] (1.25, 1.35);
      \draw (2.3, 2.3) to [bend left=20] (1.75, 1.35);
      \draw (1.5, 1) circle (0.45);
      \node at (1.503, 1.02) {$\nu$};
      \draw (1.5, 0.55) to (1.5, -0.6);
      \node at (2, -0.05) {$c$};
    \end{tikzpicture}.
  \end{equation}
  In the case of trivial 
  fusion multiplicities we thus see that the $R$ operator is diagonal,
  \begin{equation}
    R^{ab} : \fsfused{}{a}{b}{}{c} \mapsto \fsfusedbraided{}{a}{b}{}{c} = R^{ab}_c \fsfused{}{a}{b}{}{c},
  \end{equation}
  with $R^{ab}_c \in U(1)$ the \keyword{R-symbols}.
  That is, there is no mixing of the $c$ label
  (and mixing of $a$ and $b$ cannot occur since these are fixed in the definition).
\end{definition}

	We require trivial braiding with the vacuum: 
	$R^{1b}_c = 1$ and $R^{a1}_c = 1$, for all $a,b,c \in \cL$.

\subsubsection{B-symbols: Braiding of standard fusion states}

	We note that
	\begin{equation}
	  \begin{aligned}
	    \fs[1]{b,c}{a,e,d}
	    &= \sum_f \left(F^{-1}\right)^{acb}_{d;fe} \fsfusedbraided{a}{b}{c}{d}{f} 
	    = \sum_f R^{bc}_f \left(F^{-1}\right)^{acb}_{d;fe} \fsfused{a}{b}{c}{d}{f} \\
	    &= \sum_g \sum_f F^{abc}_{d;gf} \, R^{bc}_f \left(F^{-1}\right)^{acb}_{d;fe} \fs{b,c}{a,g,d},
	  \end{aligned}
	\end{equation}
	where the ordering of the symbols is natural for matrix multiplication 
	(also with regard to any suppressed multiplicities).
	The above operation on splitting states in $V^{abc}_d$ defines the B operator.

	\begin{definition}[B operator]\label{def:B}
	  The \keyword{$B$ operator} $B^{abc}_d$ is an isomorphism on 
	  the splitting space $V^{abc}_d$, given by
	  $\fs[1]{b,c}{a,e,d} = \sum_g B^{abc}_{d;ge} \fs{b,c}{a,g,d}$,
	  with the \keyword{B-symbols}
	  \begin{equation}\label{eq:B-matrix}
		B^{abc}_{d;ge} := 
		\sum_f F^{abc}_{d;gf} \, R^{bc}_f \left(F^{-1}\right)^{acb}_{d;fe}.
	  \end{equation}
	  Symbolically we write this relationship as $B = F R F^{-1}$.
	\end{definition}

\begin{lemma}\label{lem:B1}
  Consider the fusion space $V^{abc}_d$ in a multiplicity-free model.
  When one of the particle types is trivial, 
  i.e.\ $a,b,c$ or $d$ equals $1$, then $\dim V^{abc}_d = 1$ and the corresponding 
  $B$-matrix $B^{abc}_d$ is one-dimensional,
  \begin{equation}
    \begin{aligned}
      B^{1bc}_d &= B^{1bc}_{d;bc} = R^{bc}_d, \\
      B^{a1c}_d &= B^{a1c}_{d;ad} = R^{1c}_c = 1, \\
      B^{ab1}_d &= B^{ab1}_{d;da} = R^{b1}_b = 1, \\
      B^{abc}_1 &= B^{abc}_{1;\bar{c}\bar{b}} = R^{bc}_{\bar{a}}. \\
    \end{aligned}
  \end{equation}
\end{lemma}
\begin{proof}
	For example, take $a=1$.  Fusion with the vacuum fixes the intermediate
	charge in the initial tree to be $b$, while the total charge condition fixes
	the channel in which $b$ and $c$ fuse to be $d$.  Both reassociations by
	$F$ are therefore between one-dimensional spaces and are trivial by
	Lemma~\ref{lem:F1}; the only non-trivial operation left in
	$B=FRF^{-1}$ is the exchange of $b$ and $c$, which contributes
	$R^{bc}_d$.  This proves the first identity.  The other cases follow in
	the same way, with the vacuum line fixing the indicated intermediate
	charges.
	Thus, diagrammatically,
  \begin{gather}
    \fs[1]{b,c}{1,c,d}
    = \fsfusedbraidedshort{b}{c}{d}
    = R^{bc}_d \fsfusedshort{b}{c}{d}
    = R^{bc}_d \fs{b,c}{1,b,d}, \\
    \fs[1]{b,c}{a,\parbox[c][0.8em][b]{0.5em}{$\overline{b}$},1}
    = \fsfusedbraided{a}{b}{c}{1}{\overline{a}}
    = R^{bc}_{\overline{a}} \fsfused{a}{b}{c}{1}{\overline{a}}
    = R^{bc}_{\overline{a}} \fs{b,c}{a,\overline{c},1},
  \end{gather}
  and so on.
\end{proof}

\subsection{Pentagon and hexagon equations}\label{sec:models-equations}

\begin{figure}[t]
	\begin{center}
	\begin{tikzpicture}
		\node [above right] at (0,0) {\includegraphics[scale=0.05]{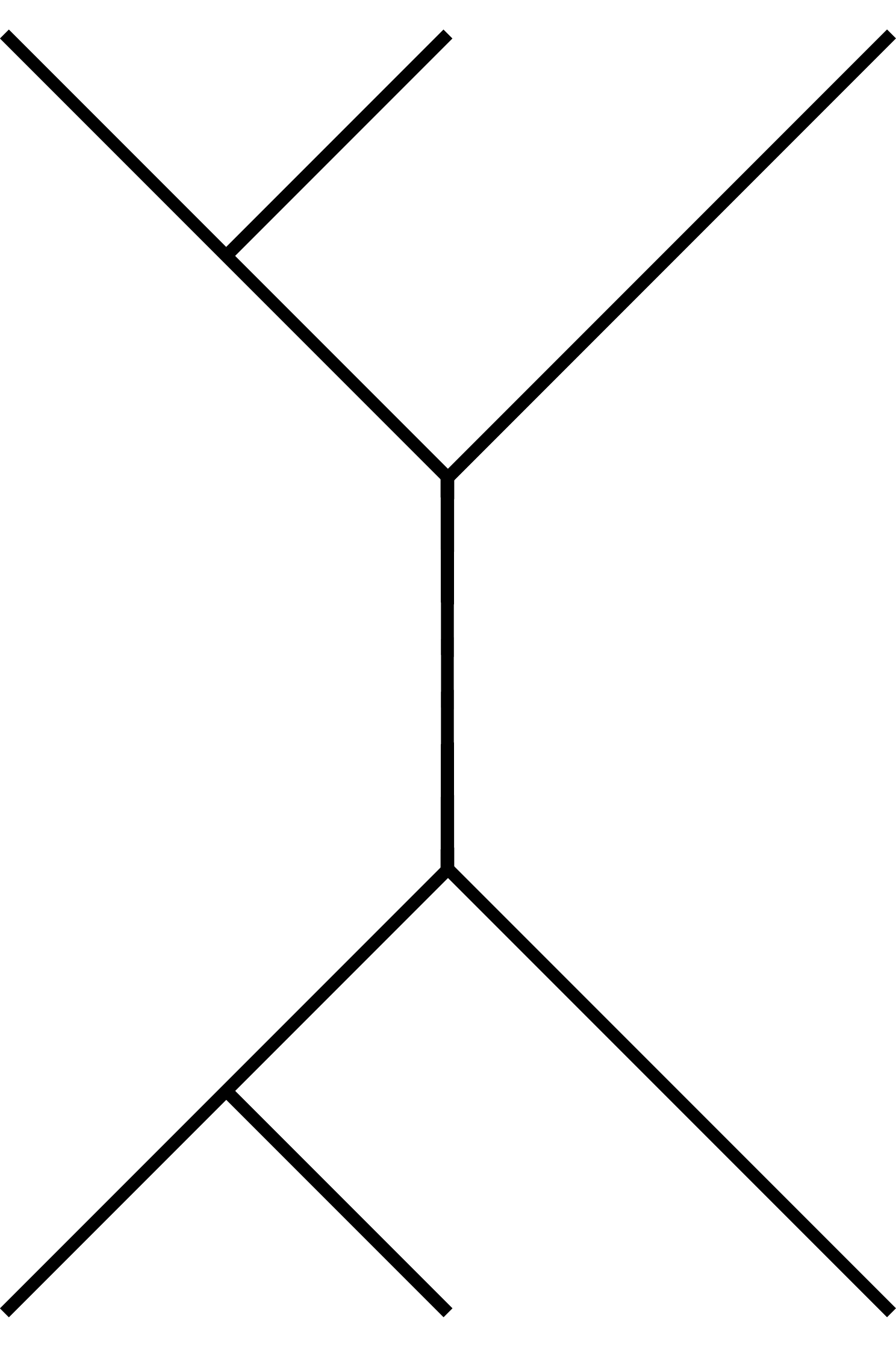}};
		\node [above right] at (-0.1,2) {\scalebox{0.8}{$a_1$\hspace{-20pt}}};
		\node [above right] at (0.4,2) {\scalebox{0.8}{$a_2\,\ldots$\hspace{-20pt}}};
		\node [above right] at (1.2,2) {\scalebox{0.8}{$a_n$\hspace{-20pt}}};
		\node [above right] at (-0.1,-0.3) {\scalebox{0.8}{$b_1$\hspace{-20pt}}};
		\node [above right] at (0.4,-0.3) {\scalebox{0.8}{$b_2\,\ldots$\hspace{-20pt}}};
		\node [above right] at (1.2,-0.3) {\scalebox{0.8}{$b_n$\hspace{-20pt}}};
		\node [above right] at (0.8,0.9) {\scalebox{0.8}{$c$\hspace{-20pt}}};
	\end{tikzpicture}
	\hspace{1.5cm}
	\begin{tikzpicture}
		\node [above right] at (0,0) {\includegraphics[scale=0.05]{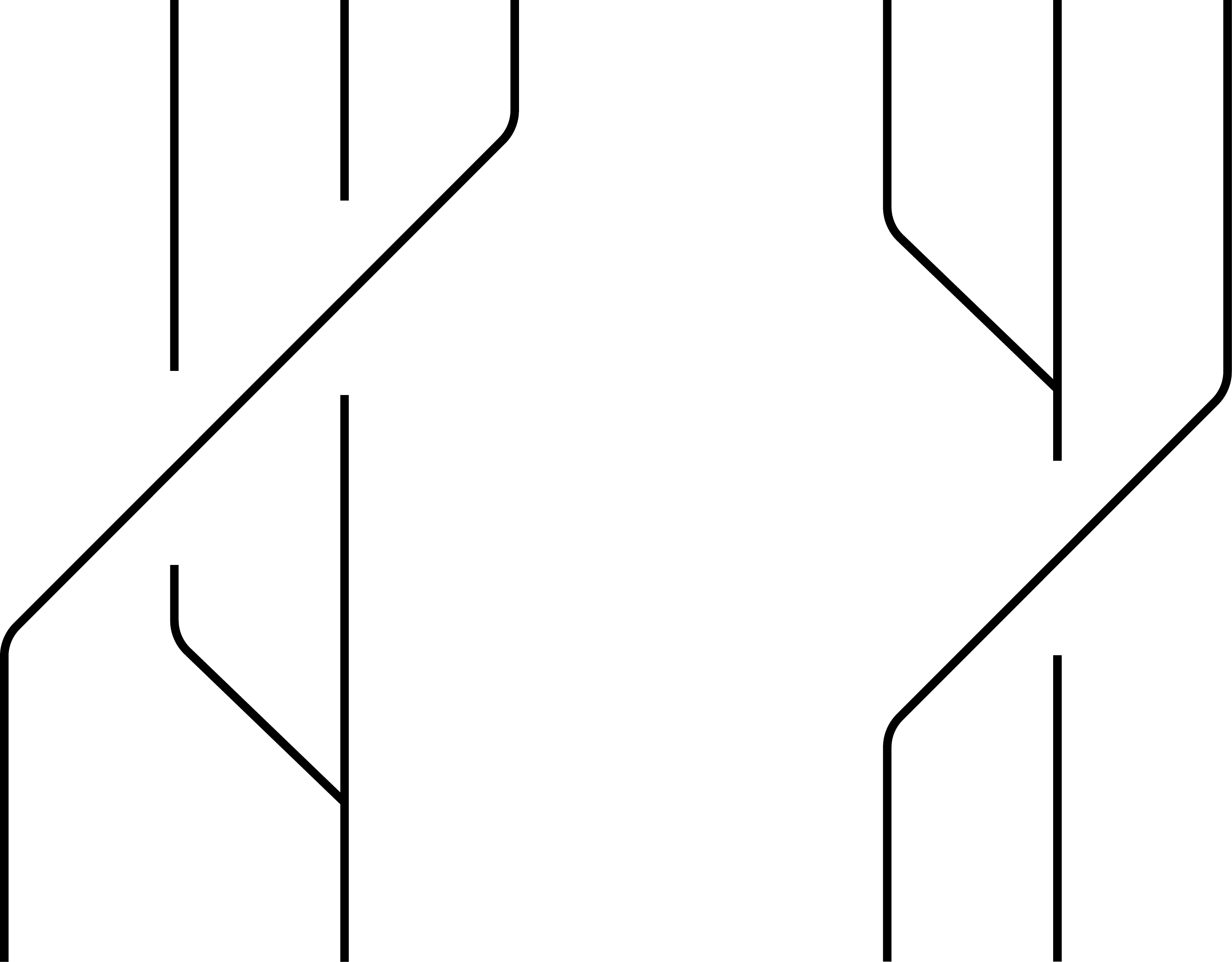}};
		\node [above right] at (1.2,1.0) {$=$};
		\node [above right] at (0.2,2.1) {\scalebox{0.8}{$a$\hspace{-20pt}}};
		\node [above right] at (0.55,2.1) {\scalebox{0.8}{$b$\hspace{-20pt}}};
		\node [above right] at (0.9,2.1) {\scalebox{0.8}{$d$\hspace{-20pt}}};
		\node [above right] at (-0.2,-0.3) {\scalebox{0.8}{$d$\hspace{-20pt}}};
		\node [above right] at (0.55,-0.3) {\scalebox{0.8}{$c$\hspace{-20pt}}};
		\node [above right] at (1.7,2.1) {\scalebox{0.8}{$a$\hspace{-20pt}}};
		\node [above right] at (2.05,2.1) {\scalebox{0.8}{$b$\hspace{-20pt}}};
		\node [above right] at (2.4,2.1) {\scalebox{0.8}{$d$\hspace{-20pt}}};
		\node [above right] at (1.7,-0.3) {\scalebox{0.8}{$d$\hspace{-20pt}}};
		\node [above right] at (2.05,-0.3) {\scalebox{0.8}{$c$\hspace{-20pt}}};
	\end{tikzpicture}
	\end{center}
	\caption{Some typical fusion/splitting/braiding diagrams.}
	\label{fig:fusions}
\end{figure}

	There are certain consistency conditions to be met upon combining
	fusions and braidings.
	Apart from the 
	geometrically more obvious relations that have been coded into
	the diagrammatic framework, such as e.g.\ Figures~\ref{fig:fusions} and \ref{fig:braids} 
	(for their formal origins, see e.g.\ \cite[Eq.~(4.3)]{MooSei-89}, 
	\cite[Eq.~(4.89)-(4.91)]{FroGab-90} or \cite[Eq.~(206)]{Kitaev-06}),
	there are also a few less obvious ones.

	One such family of constraints are the \keyword{Pentagon equations} 
	(Figure~\ref{fig:pentagon})
	\begin{equation}\label{eq:pentagon}
		F^{pzw}_{u;tq} F^{xyt}_{u;sp} 
		= \sum_{r \in \cL} F^{xyz}_{q;rp} F^{xrw}_{u;sq} F^{yzw}_{s;tr},
	\end{equation}
	and another the \keyword{Hexagon equations} (Figure~\ref{fig:hexagon})
	\begin{equation}\label{eq:hexagon}
		R^{xz}_p F^{xzy}_{u;qp} R^{yz}_q 
		= \sum_{r \in \cL} F^{zxy}_{u;rp} R^{rz}_u F^{xyz}_{u;qr},
	\end{equation}
	\begin{equation}\label{eq:hexagon-clockwise}
		\left(R^{xz}_p\right)^{-1} F^{xzy}_{u;qp} \left(R^{yz}_q\right)^{-1} 
		= \sum_{r \in \cL} F^{zxy}_{u;rp} \left(R^{rz}_u\right)^{-1} F^{xyz}_{u;qr}.
	\end{equation}

	Intuitively, the pentagon equation says that four anyons can be
	reassociated from one fusion tree to another without the result depending
	on which sequence of elementary $F$ moves is chosen: the two routes around
	the pentagon must give the same linear map.  The hexagon equations express
	the analogous compatibility between reassociation and exchange.  Braiding
	a charge past two others, either before or after changing the fusion tree,
	must produce the same map.  These coherence conditions make computations
	independent of arbitrary choices of parentheses and ensure that the local
	$F$- and $R$-moves assemble into consistent braid-group representations.

\begin{figure}[t]
	\begin{center}
	\includegraphics[scale=0.15]{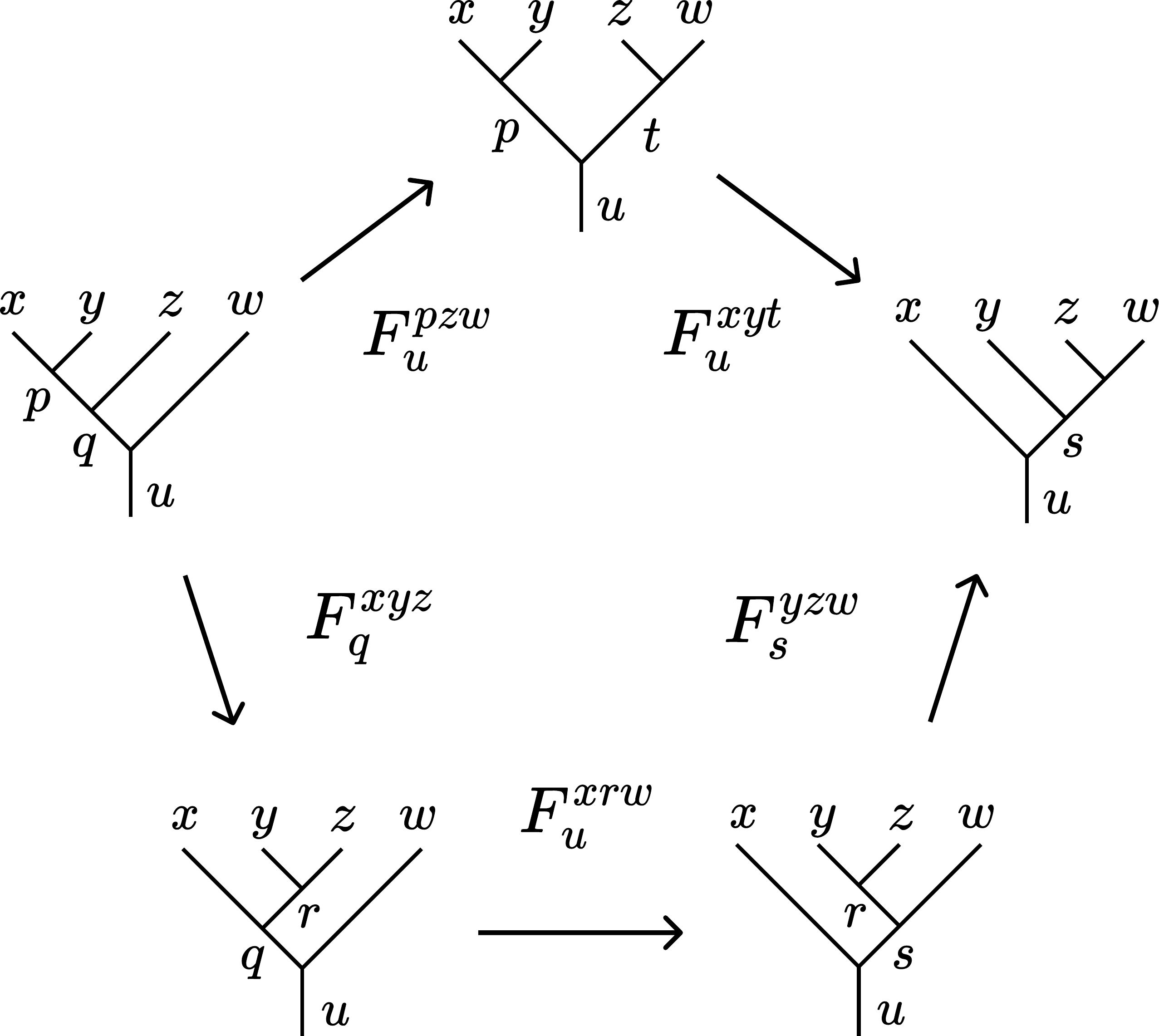}
	\ 
	\includegraphics[scale=0.15]{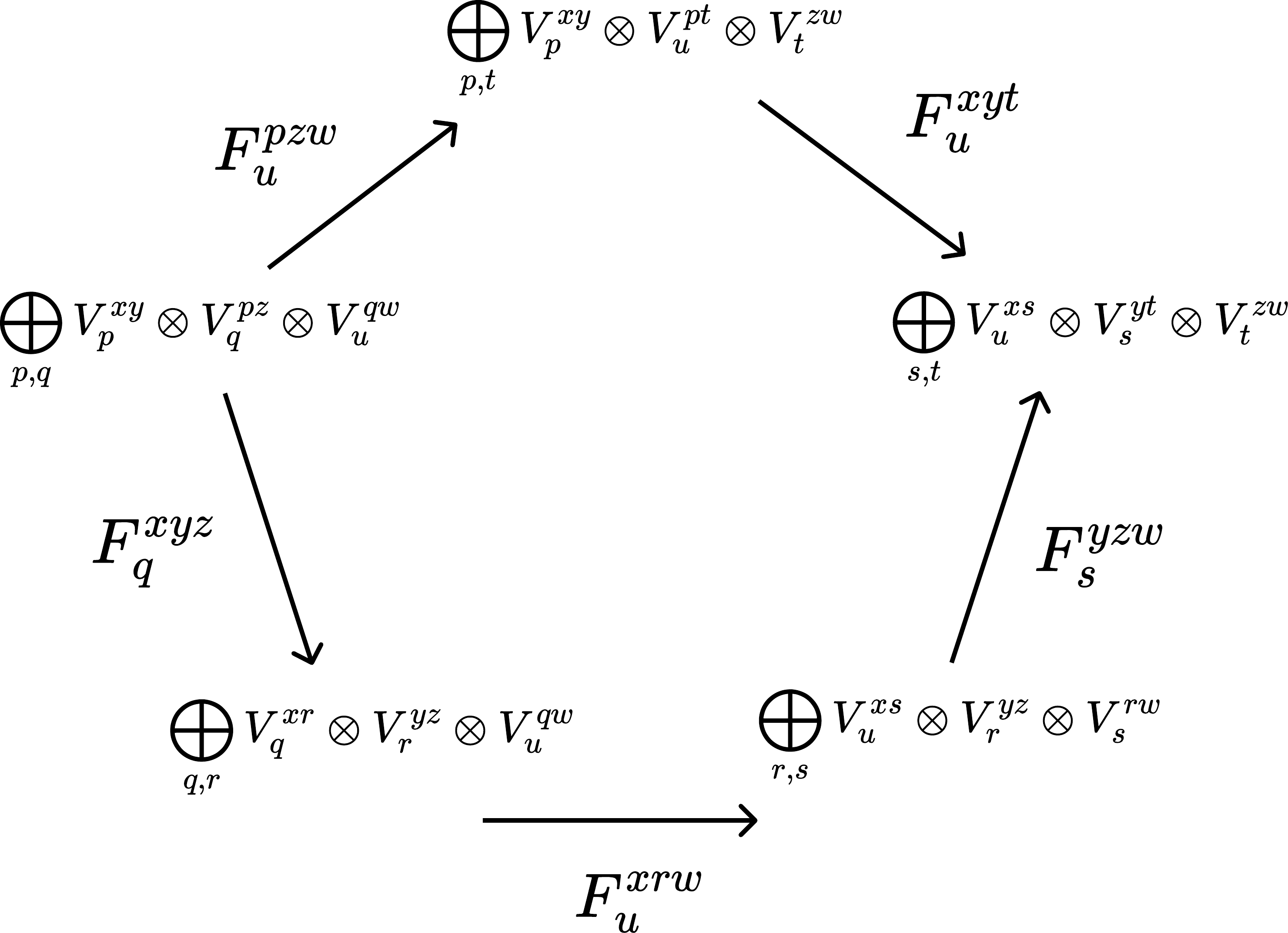}
	\end{center}
	\caption{Diagram corresponding to the pentagon equation \eqref{eq:pentagon}.}
	\label{fig:pentagon}
\end{figure}
\begin{figure}[t]
	\begin{center}
	\includegraphics[scale=0.2]{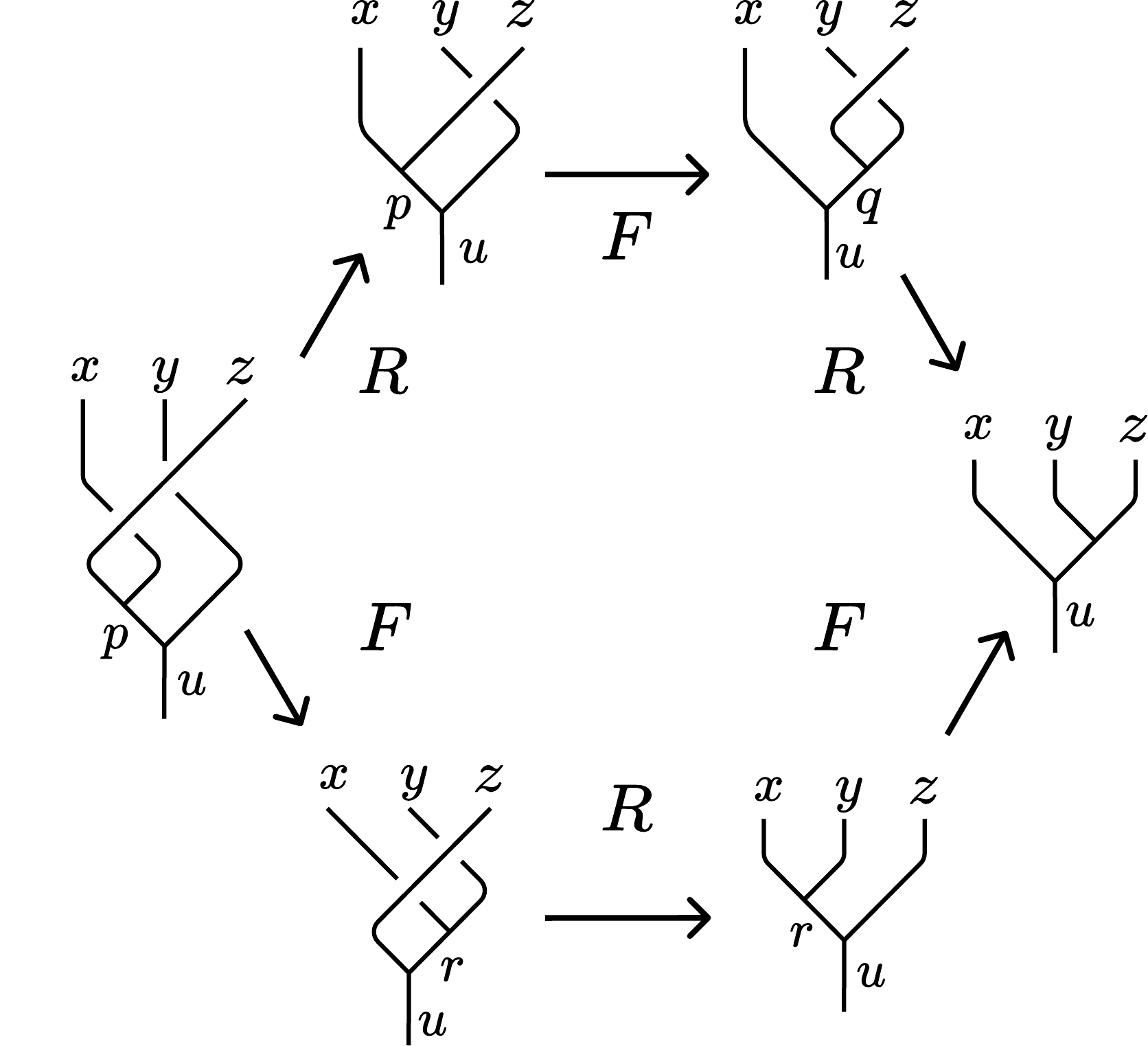}
	\hspace{1cm}
	\includegraphics[scale=0.2]{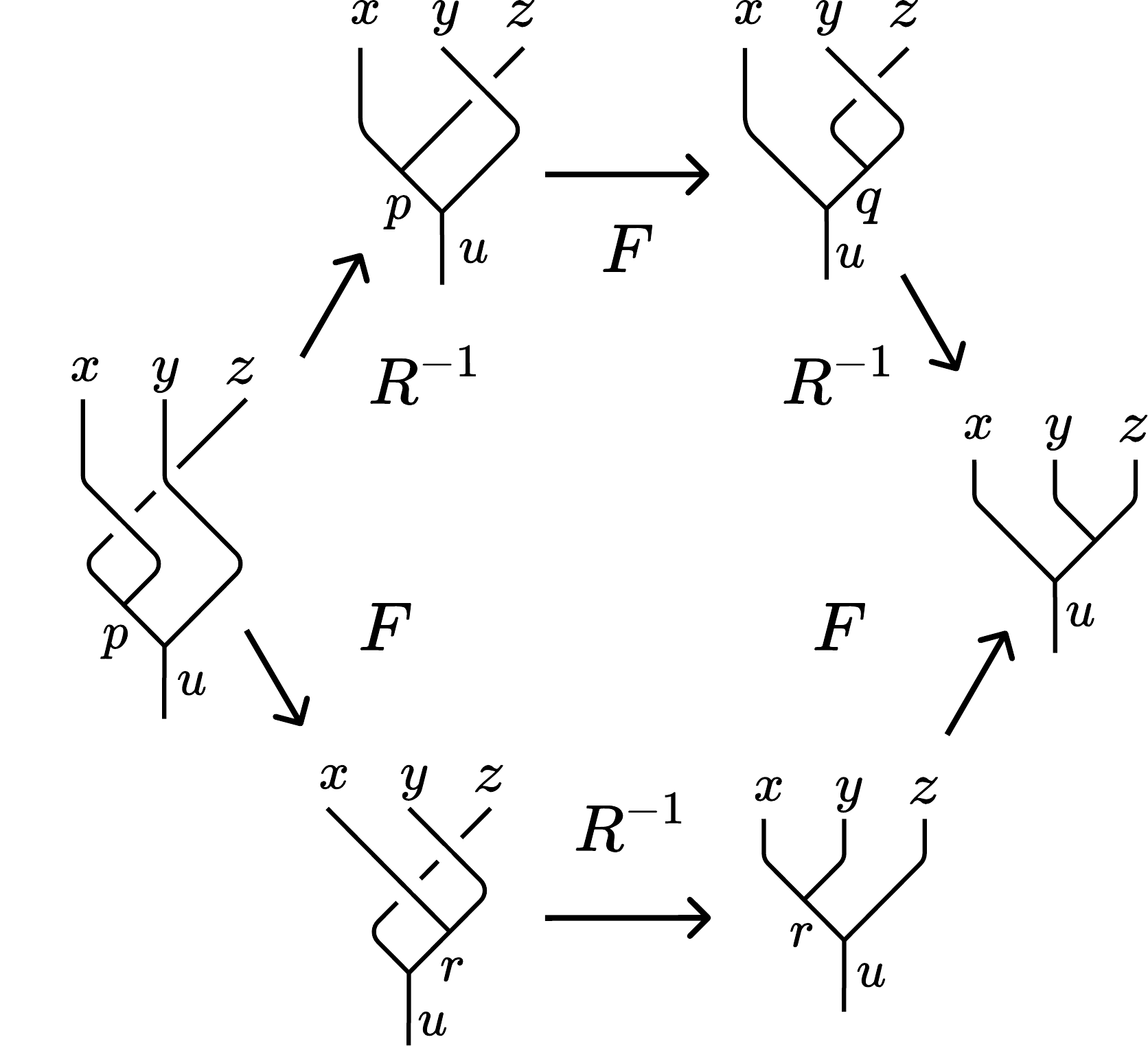}
	\end{center}
	\caption{Diagrams corresponding to the hexagon equations \eqref{eq:hexagon}-\eqref{eq:hexagon-clockwise}.}
	\label{fig:hexagon}
\end{figure}

	Furthermore, there is a remaining
	\keyword{gauge freedom} in the space of solutions to these equations; 
	see e.g. \cite[Section 2.5]{Bonderson-07} and \cite[Appendix~E.6]{Kitaev-06}. 
	Namely, for states $|\mu\rangle \in V^{ab}_c$ one may consider a change of basis
	\begin{equation}\label{eq:gauge-transformation}
		\widetilde{|\mu'\rangle} = \sum_{\mu} \bigl[ u^{ab}_c \bigr]_{\mu\mu'} |\mu\rangle,
	\end{equation}
	with $u^{ab}_c \in \sU(N_{ab}^c)$.
	In the multiplicity-free case these are simply phases, and amount to the transformations
	\begin{equation}\label{eq:F-R-gauge-freedom}
		\tilde{F}^{abc}_{d;ef} = \frac{u^{af}_d u^{bc}_f}{u^{ab}_e u^{ec}_d} F^{abc}_{d;ef},
		\qquad\qquad
		\tilde{R}^{ab}_c = \frac{u^{ba}_c}{u^{ab}_c} R^{ab}_c.
	\end{equation}
	With the totality of constraints it has been shown that there are only finitely
	many gauge equivalence classes of solutions with a given fusion algebra \eqref{eq:fusion-algebra}
	(known in category theory as Ocneanu rigidity; see \cite{Kitaev-06,EtiNikOst-05}).

	We now consider a few central examples of algebraic anyon models
	which solve the above constraints.

\subsection{Vacuum, bosons and fermions}\label{sec:models-bose-fermi}

	A simplest example of an algebraic anyon model is to take $\cL = \{1\}$,
	so that all symbols are trivial and there is only the vacuum state.
	This is the \keyword{trivial/vacuum model}.
	
	The next-to-simplest case is two symbols $\cL = \{1,\psi\}$,
	where $\bar\psi = \psi$ and $\psi \times \psi = 1$.
	Note that together with the default fusion relations 
	$1 \times 1 = 1$ and $1 \times \psi = \psi$, 
	these make up the relations of the group $\Z_2$.
	Here $F=1$ up to a gauge phase, 
	the pentagon equations are trivial while the hexagon equations reduce to
	\begin{equation}\label{eq:abelian-hexagon}
		R^{ac}_{a+c} R^{ab}_{a+b} = R^{a,(b+c)}_{a+b+c},
	\end{equation}
	where addition is that in $\Z_2$,
	with the only non-trivial equation
	$$
		(R^{\psi\psi}_1)^2 = 1.
	$$
	There are thus two such models:
	$R^{\psi\psi}_1 = \pm 1 = e^{in\pi}$, $n=0,1$, 
	denoted $\Z_2^{(n)}$
	and corresponding to \keyword{bosons} and \keyword{fermions} respectively.
	Note that there is also an accompanying grading in the fusion 
	whether we have even or odd numbers of particles,
	$$
	\psi^N = \begin{cases}1, & \text{$N$ even,}\\\psi, & \text{$N$ odd.}\end{cases}
	$$

\subsection{Abelian anyons}\label{sec:models-abelian}

	Recall that an abelian model is characterized by the fact that there is a unique
	result of fusion, $a \times b = c = b \times a$.
	
	As a generalization of the above $\Z_2$ abelian model, and
	in order to still only have a \emph{finite} collection of charges $\cL$,
	we consider an \keyword{elementary charge}
	$\alpha \in [0,2)$ such that $q = e^{i\pi\alpha}$ is a finite root of unity.
	Either $\alpha = \mu/\nu$ is a reduced fraction with $\mu \in \{0,2,\ldots,2(\nu-1)\}$
	and $\nu \ge 1$ odd, i.e. $q^\nu = 1$ an odd root of unity, 
	or $\alpha = \mu/\nu$ is a reduced fraction with $\mu \in \{1,3,\ldots,2\nu-1\}$ 
	and $\nu \ge 1$, i.e. $q^\nu = -1$ an even root of unity.
	We fix the denominator $\nu$ and define the fusion rules
	$$
		[\alpha] \times [\beta] = [\alpha + \beta \mod 2]
	$$
	or equivalently, with $\alpha = \mu/\nu$, $\beta = \lambda/\nu$,
	we may use the notation
	$$
		[\mu]_\nu \times [\lambda]_\nu = [\mu + \lambda \mod 2\nu]_\nu,
	$$
	the relations of $\Z_{2\nu}$.
	However in the case that $\nu$ is odd 
	and $\mu,\lambda$ even we obtain the subgroup $\Z_\nu$ of even integers.
	The conjugate charge is $\bar{\alpha} = -\mu/\nu \mod 2$ 
	i.e. $\overline{[\mu]_\nu} = [-\mu \mod 2\nu]_\nu$.

	Since $[\mu]_\nu^{\times{2\nu}} = [2\nu\mu]_\nu = [0]_\nu = 1$,
	iteration of the hexagon equation \eqref{eq:abelian-hexagon} yields
	$$
		(R^{[\mu]_\nu,[\mu]_\nu}_{[2\mu]_\nu})^{2\nu} 
		= R^{[\mu]_\nu,[\nu\mu]_\nu}_{[(\nu+1)\mu]_\nu} R^{[\mu]_\nu,[\nu\mu]_\nu}_{[(\nu+1)\mu]_\nu}
		= R^{[\mu]_\nu,[2\nu\mu]_\nu}_{[\mu]_\nu} = 1,
	$$
	and similarly to the previous case we have a family of models satisfying these equations
	--- all the $2\nu$:th roots of unity.
	Let $\cL=\{1,\balpha,\balpha^2,\ldots,\balpha^{2\nu-1}\}$ denote the one with
	$R^{\balpha\balpha}_{\balpha^2}=e^{i\alpha\pi}$, i.e. the elementary
	anyon $\balpha$ has the statistics parameter $\alpha$.
	In the case of even-numerator $\alpha$ we have $\balpha^\nu=1$ and a reduction to a
	group of $\nu$ roots of unity
	$\cL=\{1,\balpha,\balpha^2,\ldots,\balpha^{\nu-1}\}$.
	
	Note that the quantum dimension is $d_a = 1$ for all $a \in \cL$, 
	since $d_a d_{\bar a} = d_1 = 1$ and
	$$
		\prod_{a \in \cL} d_a = d_{[\sum_{a \in \cL} a]},
	$$
	and which would otherwise produce a contradiction.

\subsection{Fibonacci anyons}\label{sec:models-fib}

	The \keyword{Fibonacci anyon model} is determined by a single nontrivial anyon type
	$\cL = \{1,\tau\}$ 
	(with $\overline{\tau} = \tau$) 
	and the fusion rules
	\begin{equation}
		1 \times 1 = 1, \qquad
		1 \times \tau = \tau \times 1 = \tau, \qquad
		\tau \times \tau = 1 + \tau.
	\end{equation}
	The quantum dimension $d_\tau$ is the largest solution to
	$d_\tau^2 = 1 + d_\tau$. Hence, $d_\tau = \phi > 1$,
	where $\phi := (1+\sqrt{5})/2$ is the golden ratio.

	Only fusion or splitting processes involving chains of nontrivial 
	anyons $\tau$ are of interest,
	and the corresponding standard splitting spaces 
	$V^{\tau^n}_* = V^{\tau^n}_1 \oplus V^{\tau^n}_\tau$ 
	are spanned by elements of the form
	\begin{multline*}
		\fs{\tau}{1,\tau}; \quad
		\fs{\tau,\tau}{1,\tau,1},
		\fs{\tau,\tau}{1,\tau,\tau}; \quad
		\fs{\tau,\tau,\tau}{1,\tau,\tau,1}, 
		\fs{\tau,\tau,\tau}{1,\tau,1,\tau}, 
		\fs{\tau,\tau,\tau}{1,\tau,\tau,\tau}; \\
		\fs{\tau,\tau,\tau,\tau}{1,\tau,1,\tau,1}, 
		\fs{\tau,\tau,\tau,\tau}{1,\tau,\tau,\tau,1},
		\fs{\tau,\tau,\tau,\tau}{1,\tau,\tau,1,\tau},
		\fs{\tau,\tau,\tau,\tau}{1,\tau,1,\tau,\tau},
		\fs{\tau,\tau,\tau,\tau}{1,\tau,\tau,\tau,\tau};
		\quad \ldots;
	\end{multline*}
	with dimensions growing according to the Fibonacci sequence.
	Namely, the total charge can be either $1$ or $\tau$
	and if it is $\tau$ then the next-to-last intermediate charge can be either 
	$1$ or $\tau$, enumerating the states of $V^{\tau^{n-1}}_*$ recursively,
	while if the total charge is $1$ then the next-to-last must
	be $\tau$, while the second next-to-last can then be either $1$ or $\tau$,
	enumerating the states of $V^{\tau^{n-2}}_*$,
	and so on.
	
\begin{definition}\label{def:Fib}
	The $n$:th \keyword{Fibonacci number} $\Fib(n)$ is defined by the recurrence relation
	\begin{equation}
		\Fib(0) = 0, \quad
		\Fib(1) = 1, \quad
		\Fib(n) = \Fib(n-1) + \Fib(n-2).
	\end{equation}
\end{definition}

	For example:
	\begin{center}
		\renewcommand{\arraystretch}{1.5}
		\begin{tabular}{c|ccccccccccc}
		$n$       & $-3$ & $-2$ & $-1$ & $0$ & $1$ & $2$ & $3$ & $4$ & $5$ & $6$ & $7$ \\ \hline
		$\Fib(n)$ & $2$ & $-1$ & $1$ & $0$ & $1$ & $1$ & $2$ & $3$ & $5$ & $8$ & $13$
		\end{tabular}
	\end{center}
	Furthermore, there is the closed-form formula
	\begin{equation}
		\Fib(n) = \frac{\phi^n-(-\phi)^{-n}}{\sqrt{5}}.
	\end{equation}
	From the above observations we have
	$$
		\dim V^{1\tau^n}_1 = \Fib(n-1), \quad
		\dim V^{1\tau^n}_\tau = \Fib(n), \quad
		\dim V^{1\tau^n}_* = \Fib(n+1).
	$$
	
	See e.g.\ \cite{Trebst-etal-08} and \cite[Examples~3.9 and 3.12]{Garjani-17} 
	for explicit solutions of the constraints.
	The nontrivial pentagon equations read 
	\begin{equation}\label{eq:fib-pentagon}
		F^{\tau\tau c}_{\tau;da} F^{a\tau\tau}_{\tau;cb} 
		= F^{\tau\tau\tau}_{d;ce} F^{\tau e\tau}_{\tau;db} F^{\tau\tau\tau}_{b;ea},
	\end{equation}
	for all allowed $a,b,c,d,e$ internal edges of fusion trees,
	and the hexagon equations
	\begin{equation}\label{eq:fib-hexagon}
		R^{\tau\tau}_c F^{\tau\tau\tau}_{\tau;ca} R^{\tau\tau}_a 
		= \sum_b F^{\tau\tau\tau}_{\tau;cb} R^{\tau b}_\tau F^{\tau\tau\tau}_{\tau;ba},
	\end{equation}
	for $a,b,c$ internal edges of fusion trees.
	
	One can show 
	that the only potentially nontrivial F-symbols are
	$F^{\tau\tau\tau}_1$ and $F^{\tau\tau\tau}_\tau$
	(all others allowed by fusion equal 1),
	and that the pentagon equations and unitarity imply
	\begin{equation}\label{eq:F-fib}
		F^{\tau\tau\tau}_{1} = 1 
		\qquad \text{and} \qquad
		F^{\tau\tau\tau}_{\tau} = 
		\begin{bmatrix} \phi^{-1} & \eta \phi^{-1/2} \\ \bar\eta \phi^{-1/2} & -\phi^{-1} \end{bmatrix}
		=: F,
	\end{equation}
	(with a convenient abuse of notation),
	where the phase $|\eta|=1$ is a remaining gauge parameter 
	which we will put to $\eta=1$ for simplicity
	(cf. \cite{Garjani-17}).
	Note that with this convention $\det F = -1$
	and $F^2 = \1$, i.e. $F^{-1} = F = \bar{F}$, as matrices.
	Hence,
	\begin{equation}
	  \begin{aligned}
	    \fsfused{\tau}{\tau}{\tau}{1}{\tau} &=& \fs{\tau,\tau}{\tau,\tau,1}&, \\
	    \fsfused{\tau}{\tau}{\tau}{\tau}{1} &=&
	   	\phi^{-1} \ \fs{\tau,\tau}{\tau,1,\tau} &\ +& \phi^{-1/2} \fs{\tau,\tau}{\tau,\tau,\tau},\\
	    \fsfused{\tau}{\tau}{\tau}{\tau}{\tau} &=&
	   	\phi^{-1/2} \fs{\tau,\tau}{\tau,1,\tau} &\ -& \phi^{-1} \ \fs{\tau,\tau}{\tau,\tau,\tau},\\
	  \end{aligned}
	\end{equation}
	and the inverse transformations are found by a direct exchange of diagrams.

	There are two sets of solutions for the nontrivial R-symbols
	(with continued abuse of notation):
	$$
		R := R^{\tau\tau}_1 \oplus R^{\tau\tau}_\tau = \diag(R^{\tau\tau}_1,R^{\tau\tau}_\tau),
	$$
	with
	\begin{equation}\label{eq:R-fib-A}
		R^{\tau\tau}_{1} = e^{4\pi i/5}, \qquad
		R^{\tau\tau}_{\tau} = e^{-3\pi i/5},
	\end{equation}
	respectively their complex conjugates
	\begin{equation}\label{eq:R-fib-B}
		R^{\tau\tau}_{1} = e^{-4\pi i/5}, \qquad
		R^{\tau\tau}_{\tau} = e^{3\pi i/5}.
	\end{equation}
	Hence for the nontrivial B-symbols, $B = FRF$:
	\begin{equation}\label{eq:B-fib-1}
		B^{1\tau\tau}_* = R \ \ 
		(\text{i.e.}\ B^{1\tau\tau}_a = R^{\tau\tau}_a \1, \ a \in \{1,\tau\}),
		\qquad
		B^{\tau\tau\tau}_1 = R^{\tau\tau}_\tau,
	\end{equation}
	and (yet another abuse of notation)
	\begin{equation}\label{eq:B-fib-2}
		B := B^{\tau\tau\tau}_{\tau} = 
		\begin{bmatrix} B^{\tau\tau\tau}_{\tau;11} & B^{\tau\tau\tau}_{\tau;1\tau} \\ B^{\tau\tau\tau}_{\tau;\tau 1} & B^{\tau\tau\tau}_{\tau;\tau\tau} \end{bmatrix} =
		\begin{bmatrix} \phi^{-1} e^{-4\pi i/5} & \phi^{-1/2} e^{3\pi i/5} \\ \phi^{-1/2} e^{3\pi i/5} & -\phi^{-1} \end{bmatrix},
	\end{equation}
	i.e. diagrammatically
	\begin{equation}
	  \begin{aligned}
	    \fs[1]{\tau,\tau}{1,\tau,1} &= e^{4\pi i/5} \fs{\tau,\tau}{1,\tau,1}, \qquad 
	    \fs[1]{\tau,\tau}{1,\tau,\tau} = e^{-3\pi i/5} \fs{\tau,\tau}{1,\tau,\tau}, \qquad
	    \fs[1]{\tau,\tau}{\tau,\tau,1} = e^{-3\pi i/5} \fs{\tau,\tau}{\tau,\tau,1}, \\
	    \fs[1]{\tau,\tau}{\tau,1,\tau} &= \phi^{-1} e^{-4\pi i/5} \fs{\tau,\tau}{\tau,1,\tau} + \phi^{-1/2} e^{3\pi i/5} \fs{\tau,\tau}{\tau,\tau,\tau},\\
	    \fs[1]{\tau,\tau}{\tau,\tau,\tau} &= \phi^{-1/2} e^{3\pi i/5} \fs{\tau,\tau}{\tau,1,\tau} - \phi^{-1} \fs{\tau,\tau}{\tau,\tau,\tau}.
	  \end{aligned}
	\end{equation}
	One could also take all symbols complex conjugated, corresponding to the choice \eqref{eq:R-fib-B}.
	We will stick to the first choice \eqref{eq:R-fib-A} throughout,
	while the second \eqref{eq:R-fib-B} amounts to time reversal of all diagrams.
	Note further that $B^{-1} = B^\dagger = \bar{B}$ and 
	$B^{10} = FR^{10}F = \1$.
	
\subsection{Ising anyons}\label{sec:models-ising}

	The \keyword{Ising anyon model} is described by the set of charges
	$\cL = \{1,\sigma,\psi\}$ 
	(with $\overline{\sigma} = \sigma$ and $\overline{\psi} = \psi$),
	i.e. two nontrivial anyon types,
	with the nontrivial fusion rules
	\begin{equation}
	  \begin{aligned}
	    \sigma \times \sigma &= 1 + \psi, \\
	    \sigma \times \psi &= \psi \times \sigma = \sigma, \\
	    \psi \times \psi &= 1.
	  \end{aligned}
	\end{equation}
	The quantum dimensions are thus $d_\psi =1$ and $d_\sigma = \sqrt{2}$.
	The $\sigma$ and $\psi$ are called Ising anyon and Majorana fermion, respectively.
	Even though there is the additional particle type in this model, the possible 
	fusion/splitting basis states involving the Ising anyon $\sigma$
	are actually less numerous than 
	those of the Fibonacci anyon
	($d_\sigma < d_\tau$),
	the first few being

	\begin{multline*}
		\fs{\sigma}{1,\sigma}; \ \ 
		\fs{\sigma,\sigma}{1,\sigma,1},
		\fs{\sigma,\sigma}{1,\sigma,\psi}; \ \ 
		\fs{\sigma,\sigma,\sigma}{1,\sigma,1,\sigma}, 
		\fs{\sigma,\sigma,\sigma}{1,\sigma,\psi,\sigma}; \ \ 
		\fs{\sigma,\sigma,\sigma,\sigma}{1,\sigma,1,\sigma,1}, 
		\fs{\sigma,\sigma,\sigma,\sigma}{1,\sigma,\psi,\sigma,1},
		\fs{\sigma,\sigma,\sigma,\sigma}{1,\sigma,1,\sigma,\psi},
		\fs{\sigma,\sigma,\sigma,\sigma}{1,\sigma,\psi,\sigma,\psi};
		\ \ \ldots;
	\end{multline*}
	
	There are two (see e.g.\ \cite[Example~3.10]{Garjani-17}) solutions for the F-symbols:
	\begin{equation}\label{eq:F-ising-1}
		F^{\sigma\sigma\sigma}_{\sigma} = \pm \frac{1}{\sqrt{2}}
		\begin{bmatrix} 1 & 1 \\ 1 & -1 \end{bmatrix} =: F,
	\end{equation}
	\begin{equation}\label{eq:F-ising-2}
		F^{\sigma\psi\sigma}_{\psi;\sigma\sigma}
		= F^{\psi\sigma\psi}_{\sigma;\sigma\sigma} = -1,
	\end{equation}
	all other symbols being trivial.
	Note that also in this case $F^{-1} = F = \bar{F}$.
	There are actually eight different anyon models with the above fusion algebra,
	characterized by the value of the ``\keyword{topological spin}'' 
	(cf. \cite[Sec.~10.5-10.6]{Kitaev-06}, \cite{BonGurNay-11})
	$$
		\theta_\sigma := \frac{1}{\sqrt{2}}\left( R^{\sigma\sigma}_1 + R^{\sigma\sigma}_\psi \right)
		= e^{i2\pi(2j+1)/16}, \qquad j = 0,1,\ldots,7.
	$$
	We consider in this work
	one specific set of solutions for the $R$-symbols, with $\theta_\sigma = e^{\pi i/8}$:
	\begin{equation}\label{eq:R-ising}
		R^{\sigma\sigma}_{1} = e^{-\pi i/8}, \qquad
		R^{\sigma\sigma}_{\psi} = e^{3\pi i/8}, \qquad
		R^{\sigma\psi}_{\sigma} = R^{\psi\sigma}_{\sigma} = e^{-\pi i/2} = -i, \qquad
		R^{\psi\psi}_{1} = -1,
	\end{equation}
	and another is found by taking complex conjugates. 
	Corresponding diagrams are given in \cite[Table~1]{Kitaev-06}.
	Hence for the nontrivial B-symbols, we have directly from Lemma~\ref{lem:B1}
	\begin{equation}\label{eq:B-ising-1}
	\begin{aligned}
		B^{1\psi\psi}_1 = B^{1\psi\psi}_{1;\psi\psi} &= R^{\psi\psi}_1 = -1, \\
		B^{1\sigma\psi}_\sigma = B^{1\sigma\psi}_{\sigma;\sigma\psi} &= R^{\sigma\psi}_\sigma = -i, \\
		B^{1\psi\sigma}_\sigma = B^{1\psi\sigma}_{\sigma;\psi\sigma} &= R^{\psi\sigma}_\sigma = -i, \\
		B^{1\sigma\sigma}_{*} &= R^{\sigma\sigma}_*, \\
		B^{\psi\sigma\sigma}_1 = B^{\psi\sigma\sigma}_{1;\sigma\sigma} &= R^{\sigma\sigma}_\psi = e^{3\pi i/8}, \\
		B^{\sigma\psi\sigma}_1 = B^{\sigma\psi\sigma}_{1;\sigma\psi} &= R^{\psi\sigma}_\sigma = -i, \\
		B^{\sigma\sigma\psi}_{1} = B^{\sigma\sigma\psi}_{1;\psi\sigma} &= R^{\sigma\psi}_\sigma = -i, \\
	\end{aligned}
	\end{equation}
	and by simple computation one finds also
	\begin{equation}\label{eq:B-ising-2}
	\begin{aligned}
		B^{\psi\psi\psi}_\psi = B^{\psi\psi\psi}_{\psi;1 1} &= R^{\psi\psi}_1 = -1, \\
		B^{\sigma\psi\psi}_\sigma = B^{\sigma\psi\psi}_{\sigma;\sigma\sigma} &= R^{\psi\psi}_1 = -1, \\
		B^{\psi\sigma\psi}_\sigma = B^{\psi\sigma\psi}_{\sigma;\sigma 1} &= -R^{\sigma\psi}_\sigma = i, \\
		B^{\psi\psi\sigma}_\sigma = B^{\psi\psi\sigma}_{\sigma;1\sigma} &= -R^{\psi\sigma}_\sigma = i, \\
		B^{\psi\sigma\sigma}_\psi = B^{\psi\sigma\sigma}_{\psi;\sigma\sigma} &= R^{\sigma\sigma}_1 = e^{-\pi i/8}, \\
		B^{\sigma\psi\sigma}_\psi = B^{\sigma\psi\sigma}_{\psi;\sigma 1} &= -R^{\psi\sigma}_\sigma = i, \\
		B^{\sigma\sigma\psi}_\psi = B^{\sigma\sigma\psi}_{\psi;1\sigma} &= -R^{\sigma\psi}_\sigma = i, \\
	\end{aligned}
	\end{equation}
	as well as
	\begin{equation}\label{eq:B-ising-3}
		B  = FRF = B^{\sigma\sigma\sigma}_{\sigma} = \frac{1}{2}
		\begin{bmatrix} R^{\sigma\sigma}_1+R^{\sigma\sigma}_\psi & R^{\sigma\sigma}_1-R^{\sigma\sigma}_\psi \\ 
			R^{\sigma\sigma}_1-R^{\sigma\sigma}_\psi & R^{\sigma\sigma}_1+R^{\sigma\sigma}_\psi \end{bmatrix}
		= \frac{1}{\sqrt{2}}
		\begin{bmatrix} e^{\pi i/8} & e^{-3\pi i/8} \\ e^{-3\pi i/8} & e^{\pi i/8} \end{bmatrix}.
	\end{equation}
	Note again that $B^{-1} = B^\dagger = \bar{B}$, while $B^4 = FR^4F = -i\1$.


\section{Exchange operators and phases}\label{sec:phases}

	Here we show how to compute the operators and corresponding phases 
	associated with an exchange of two
	anyons in an arbitrary algebraic anyon model,
	and illustrate the procedure for abelian, Fibonacci and Ising anyons.
	We also consider a couple of other common representations 
	that arise in a different way.
	For an introduction to the braid group and its representations 
	we refer to \cite{Birman-74,KasTur-08}.

\subsection{Braid group representations}\label{sec:phases-reps}

\begin{definition}[Braid/permutation group]\label{def:braid-group}
    The \keyword{braid group} on $N$ strands, $B_N$, 
	may be defined as the group
	with generators $\sigma_1, \ldots, \sigma_{N-1}$
	subject to the relations
	\begin{subequations}
	\label{eq:braid-relations}
	  \begin{align}
	    \label{eq:braid-relation-1}
	    \sigma_j \sigma_k &= \sigma_k \sigma_j, &&\text{if } |j-k| \ge 2, \\
	    \label{eq:braid-relation-2}
	    \sigma_j \sigma_{j+1} \sigma_j &= \sigma_{j+1} \sigma_j \sigma_{j+1}, &&j = 1,2,\ldots,N-2,
	  \end{align}
	\end{subequations}
	and similarly for their inverses $\sigma_j^{-1}$.
	The \keyword{symmetric/permutation group} $S_N$ may be defined by the same generators and relations 
	as the braid group, with the additional relations
	$\sigma_j^2 = 1$ for $j=1,2,\ldots,N-1$,
	hence $\sigma_j^{-1} = \sigma_j$ in this case.
\end{definition}
	
\begin{figure}[t]\label{fig:braids}
	\centering
	\begin{tikzpicture}[scale=0.4,font=\footnotesize,anchor=mid,baseline={([yshift=-.5ex]current bounding box.center)}]
		\braid[number of strands=7] s_4^{-1};
		\node at (1, -2.0) {$1$};
		\node at (2, -2.0) {$2$};
		\node at (3, -2.0) {$\ldots$};
		\node at (4, -2.0) {$j$};
		\node at (6, -2.0) {$\ldots$};
		\node at (7, -2.0) {$N$};
	\end{tikzpicture}
	\qquad
	\qquad
	\begin{tikzpicture}[scale=0.4,font=\footnotesize,anchor=mid,baseline={([yshift=-.5ex]current bounding box.center)}]
		\braid[number of strands=4] s_1^{-1} s_2^{-1} s_1^{-1};
	\end{tikzpicture}
	\ = \ 
	\begin{tikzpicture}[scale=0.4,font=\footnotesize,anchor=mid,baseline={([yshift=-.5ex]current bounding box.center)}]
		\braid[number of strands=4] s_2^{-1} s_1^{-1} s_2^{-1};
	\end{tikzpicture}
	\qquad
	\qquad
	\begin{tikzpicture}[scale=0.4,font=\footnotesize,anchor=mid,baseline={([yshift=-.5ex]current bounding box.center)}]
		\braid[number of strands=4] s_3^{-1} s_1^{-1};
	\end{tikzpicture}
	\ = \ 
	\begin{tikzpicture}[scale=0.4,font=\footnotesize,anchor=mid,baseline={([yshift=-.5ex]current bounding box.center)}]
		\braid[number of strands=4] s_1^{-1} s_3^{-1};
	\end{tikzpicture}
	\caption{Braid diagrams corresponding to the generator $\sigma_j$ of $B_N$, 
		i.e. a counterclockwise exchange of particles/strands $j$ and $j+1$
		with time running upwards,
		and the relations $\sigma_1 \sigma_2 \sigma_1 = \sigma_2 \sigma_1 \sigma_2$
		respectively $\sigma_3 \sigma_1 = \sigma_1 \sigma_3$ of $B_4$.}
\end{figure}

	See Figure~\ref{fig:braids} for corresponding diagrams that encode these relations.
	An important example that we will return to frequently in this work is the
	exchange of two anyons, while enclosing exactly $p$ other anyons.
	This is represented by the group element and diagram (Figure~\ref{fig:loops}, right):
	\begin{equation}\label{eq:exchange-braid}
		\Sigma_p := \sigma_1 \sigma_2 \ldots \sigma_p \sigma_{p+1} \sigma_p \ldots \sigma_2 \sigma_1
		\qquad \leftrightarrow \qquad
	    \begin{tikzpicture}[scale=0.4,font=\footnotesize,anchor=mid,baseline={([yshift=-.5ex]current bounding box.center)}]
			\braid s_1^{-1} s_2^{-1} s_3^{-1} s_2^{-1} s_1^{-1};
			\node at (1, -6.25) {\tiny 1};
			\node at (2, -6.25) {\tiny 2};
			\node at (3, -6.25) {\tiny ...};
			\node at (4.2, -6.25) {\tiny p+2};
			\node at (1, 0.5) {$t$};
			\node at (2.1, 0.5) {$t_1..$};
			\node at (3.1, 0.5) {$t_p$};
			\node at (4, 0.5) {$t$};
		\end{tikzpicture}.
	\end{equation}
	
	In fact, the number of generators of the braid group may be reduced 
	to only two: 
	
\begin{lemma}\label{lem:braid-presentation}
	For any $N \ge 3$, $B_N$ is generated by the two elements
	$\sigma_1$ and
	$$
		\Theta_N := \sigma_1 \sigma_2 \ldots \sigma_{N-1},
	$$
	and furthermore
	\begin{equation}\label{eq:braid-presentation}
		\sigma_k = \Theta_N^{k-1} \sigma_1 \Theta_N^{1-k},
	\end{equation}
	for any $k \in \{1,2,\ldots,N-1\}$.
\end{lemma}
	The proof is straightforward by induction;
	see e.g.\ \cite[Theorem~3.5]{Weinberger-15}. 
	One should note however that the group relations expressed in terms of 
	$\{\sigma_1,\Theta_N\}$ are more complicated.

	As a consequence of this alternative presentation 
	we have also the following similarity property.
	We write as usual $\tilde U \sim U$ if $\tilde U = S^{-1}US$ for some isometry $S$.

\begin{lemma}\label{lem:similar-generators}
	For any $N \ge 2$, any representation $\rho\colon B_N \to \sGL(V)$ 
	and any pair $j,k \in \{1,2,\ldots,N-1\}$ of generators, we have similarity
	\begin{equation}\label{eq:similar-generators}
		\rho(\sigma_j) \sim \rho(\sigma_k).
	\end{equation}
	Furthermore, any unitary representation $\rho\colon B_N \to \sU(V)$
	is completely reducible into orthogonal blocks of irreducible unitary representations
	$\rho = \oplus_n \rho_n$, and if $\rho_n(\sigma_j)$ is scalar 
	(a scalar multiple of the identity)
	for some $j$ then so is $\rho_n$.
\end{lemma}
\begin{proof}
	This is a trivial consequence of the presentation \eqref{eq:braid-presentation}:
	$$
		\rho(\sigma_k) = \rho(\Theta_N^{k-1}) \rho(\sigma_1) \rho(\Theta_N^{k-1})^{-1},
	$$
	hence $\rho(\sigma_1) \sim \rho(\sigma_2) \sim \ldots \sim \rho(\sigma_{N-1})$.
	For the second statement, see e.g. 
	\cite[Propositions~4.2 and 4.7]{Weinberger-15}.
\end{proof}

\begin{definition}[Abelian vs.\ non-abelian representations]
	Naturally, we call a representation $\rho\colon B_N \to \mathrm{GL}(V)$ 
	\keyword{abelian} if $[\rho(\sigma_j),\rho(\sigma_k)] = 0$ 
	for all $j,k \in \{1,\ldots,N-1\}$,
	and \keyword{non-abelian} otherwise.
\end{definition}

	Now let us look for concrete representations $\rho(\sigma_j)$ 
	on the spaces of our algebraic anyon models.
	We will use the following notation for the standard fusion/splitting spaces:
	\begin{equation}
	\begin{aligned}
		V^{a,t^n}_c & =
		\Span \left\{ \fswide{t,t}{a,b_1,b_2} \cdots \fswider{t,t}{b_{n-2},b_{n-1},c} : \text{ for all possible $b_j$} \right\} \\
		V^{a,t^n}_* & =
		\Span \left\{ \fswide{t,t}{a,b_1,b_2} \cdots \fswider{t,t}{b_{n-2},b_{n-1},b_n} : \text{ for all possible $b_j$} \right\} \\
		V^{1,t^n}_* & =
		\Span \left\{ \fswide{t,t}{1,t,b_1} \cdots \fswider{t,t}{b_{n-3},b_{n-2},b_{n-1}} : \text{ for all possible $b_j$} \right\} = V^{t,t^{n-1}}_*\\
		V^{*,t^n}_* & =
		\Span \left\{ \fswide{t,t}{b_1,b_2,b_3} \cdots \fswider{t,t}{b_{n-1},b_{n},b_{n+1}} : \text{ for all possible $b_j$} \right\}
	\end{aligned}
	\end{equation}
	Note how the only thing that differs are the fixing of left and right charge sectors.
	Also note the chain of inclusions 
	$V^{1,t^n}_c \subseteq V^{1,t^n}_* \subseteq V^{*,t^n}_*$,
	and
	\begin{equation}\label{eq:splitting-sectors}
		V^{*,t^n}_* = \bigoplus_{\substack{a \in \cL\\c \in a\times t^n}} V^{a,t^n}_c.
	\end{equation}

\begin{definition}[\keyword{Splitting basis representation}]\label{def:splitting-rep}
	Denote $\rho_n(\sigma_j)$ the 
	representation of the braid generator
	$\sigma_j$ on the standard splitting space $V^{*,t^n}_*$,
	mapping a basis state to
	\begin{equation}\label{eq:splitting-reps}
	\begin{tikzpicture}[scale=0.4,font=\footnotesize,anchor=mid,baseline={([yshift=-.5ex]current bounding box.center)}]
		\braid[number of strands=7] s_4^{-1};
		\node at (1, 0.5) {$t$};
		\node at (2, 0.5) {$t$};
		\node at (3, 0.5) {$\ldots$};
		\node at (4, 0.5) {$t$};
		\node at (5, 0.5) {$t$};
		\node at (6, 0.5) {$\ldots$};
		\node at (7, 0.5) {$t$};
		\draw (0, -1.5) to (8, -1.5);
		\node at (0.5, -2.0) {$b_1$};
		\node at (1.5, -2.0) {$b_2$};
		\node at (3, -2.0) {$\ldots$};
		\node at (4.5, -2.0) {$b_{j+1}$};
		\node at (6, -2.0) {$\ldots$};
		\node at (7.9, -2.0) {$b_{n+1}$};
	\end{tikzpicture},
	\quad \text{ for all possible intermediate $b_k \in \cL$}, \ 
	k \in \{1,\ldots,n+1\}.
	\end{equation}
	Hence the $j$th anyon goes in front of the $j+1$th anyon, 
	while in $\rho_n(\sigma_j^{-1})$ the $j+1$th anyon goes in front of the $j$th anyon.
	The action for an arbitrary braid 
	$b = \sigma_{k_1}^{n_{k_1}} \ldots \sigma_{k_m}^{n_{k_m}}$
	is defined similarly by composition upwards
	(multiplication on the left corresponding to braiding later in time), 
	thus obtaining a $D_n$-dimensional unitary representation 
	$\rho_n\colon B_n \to \sU(V^{*,t^n}_*)$,
	with
	$$
		D_n = \dim V^{*,t^n}_* 
		= \sum_{b_2, \ldots, b_{n-1}} N_{b_1 t}^{b_2} N_{b_2 t}^{b_3} \ldots N_{b_n t}^{b_{n+1}}
	$$
	(thus, typically,
	$D_n \sim d_t^n$, as $n \to \infty$).
\end{definition}
	
	Subrepresentations may then also be obtained by restricting to 
	$V^{a,t^n}_c$ 
	for allowed values of $a = b_1$ and $c = b_{n+1}$,
	$\rho_n = \oplus \rho^{a,t^n}_c$.
	(Note that braiding does not change these values, cf. Lemma~\ref{lem:Up-braid} below.)

\begin{theorem}\label{thm:standard-rep-R}
	We have for the splitting basis representation $\rho_n$ with anyon type $t$
	\begin{equation}
		\rho_n(\sigma_j) \sim \bigoplus R^{tt},
	\end{equation}
	where the sum is over the labels $a,c \in \cL$ of subrepresentations $V^{a,t^n}_c$
	in the decomposition \eqref{eq:splitting-sectors} 
	and the intermediate possibilities $d$ of splitting spaces $V^{att}_d$.
	
\end{theorem}
\begin{proof}
	On each $V^{a,t^n}_c$ in \eqref{eq:splitting-sectors} use
	Lemma~\ref{lem:similar-generators} and the definition of the B-operator,
	$B \sim R$,
	$$
		\rho_c^{a,t^n}(\sigma_1) 
		= \bigoplus_{d \in a \times t \times t} B^{att}_d \otimes \1_c^{d,t^{n-2}}
		\sim \bigoplus_{d \in a \times t \times t} R^{tt}_* \otimes \1_c^{d,t^{n-2}}.
	$$
\end{proof}

	In the case that consecutive braids are considered, 
	the following diagram, allowing for an arbitrary central element $c$,
	will be particularly important.
	
\begin{lemma}\label{lem:Up-braid}
	Useful diagram:
  \begin{equation}\label{eq:exchange-t1-t2}
    \begin{tikzpicture}[scale=0.4,font=\footnotesize,anchor=mid,baseline={([yshift=-.5ex]current bounding box.center)}]
      \braid s_1^{-1} s_2^{-1} s_1^{-1};
      \node at (1, 0.5) {$t_1$};
      \node at (2, 0.5) {$c$};
      \node at (3, 0.5) {$t_2$};
      \draw (0, -3.5) to (4, -3.5);
      \node at (0.5, -4.25) {$a$};
      \node at (1.5, -4.25) {$b$};
      \node at (2.5, -4.25) {$d$};
      \node at (3.5, -4.25) {$e$};
    \end{tikzpicture} =
    \sum_{f,g,h} B^{act_2}_{d;fb} \, B^{ft_1t_2}_{e;gd} B^{at_1c}_{g;hf} \,
    \begin{tikzpicture}[scale=0.4,font=\footnotesize,anchor=mid,baseline={([yshift=-.5ex]current bounding box.center)}]
      \draw (1, 0) to (1, -1);
      \draw (2, 0) to (2, -1);
      \draw (3, 0) to (3, -1);
      \node at (1, 0.5) {$t_1$};
      \node at (2, 0.5) {$c$};
      \node at (3, 0.5) {$t_2$};
      \draw (0, -1) to (4, -1);
      \node at (0.5, -1.75) {$a$};
      \node at (1.5, -1.75) {$h$};
      \node at (2.5, -1.75) {$g$};
      \node at (3.5, -1.75) {$e$};
    \end{tikzpicture}.
  \end{equation}
\end{lemma}
\begin{proof}  
	We can replace
	$\fs[1]{c,\,t_2}{a,b,d} = \sum_f B^{act_2}_{d;fb} \fs{c,\,t_2}{a,f,d}$
	and continue:
  \begin{equation}
    \begin{aligned}
      \begin{tikzpicture}[scale=0.4,font=\footnotesize,anchor=mid,baseline={([yshift=-.5ex]current bounding box.center)}]
        \braid s_1^{-1} s_2^{-1} s_1^{-1};
        \node at (1, 0.5) {$t_1$};
        \node at (2, 0.5) {$c$};
        \node at (3, 0.5) {$t_2$};
        \draw (0, -3.5) to (4, -3.5);
        \node at (0.5, -4.25) {$a$};
        \node at (1.5, -4.25) {$b$};
        \node at (2.5, -4.25) {$d$};
        \node at (3.5, -4.25) {$e$};
      \end{tikzpicture} 
      &=
      \sum_{f} B^{act_2}_{d;fb}
      \begin{tikzpicture}[scale=0.4,font=\footnotesize,anchor=mid,baseline={([yshift=-.5ex]current bounding box.center)}]
        \braid s_1^{-1} s_2^{-1};
        \node at (1, 0.5) {$t_1$};
        \node at (2, 0.5) {$c$};
        \node at (3, 0.5) {$t_2$};
        \draw (0, -2.5) to (4, -2.5);
        \node at (0.5, -3.25) {$a$};
        \node at (1.5, -3.25) {$f$};
        \node at (2.5, -3.25) {$d$};
        \node at (3.5, -3.25) {$e$};
      \end{tikzpicture} 
      =
      \sum_{f} B^{act_2}_{d;fb}
      \sum_{g} B^{ft_1t_2}_{e;gd}
      \begin{tikzpicture}[scale=0.4,font=\footnotesize,anchor=mid,baseline={([yshift=-.5ex]current bounding box.center)}]
        \braid s_1^{-1};
        \draw (3, 0) to (3, -1.5);
        \node at (1, 0.5) {$t_1$};
        \node at (2, 0.5) {$c$};
        \node at (3, 0.5) {$t_2$};
        \draw (0, -1.5) to (4, -1.5);
        \node at (0.5, -2.25) {$a$};
        \node at (1.5, -2.25) {$f$};
        \node at (2.5, -2.25) {$g$};
        \node at (3.5, -2.25) {$e$};
      \end{tikzpicture} \\
      &=
      \sum_{f} B^{act_2}_{d;fb}
      \sum_{g} B^{ft_1t_2}_{e;gd}
      \sum_{h} B^{at_1c}_{g;hf}
      \begin{tikzpicture}[scale=0.4,font=\footnotesize,anchor=mid,baseline={([yshift=-.5ex]current bounding box.center)}]
        \draw (1, 0) to (1, -1);
        \draw (2, 0) to (2, -1);
        \draw (3, 0) to (3, -1);
        \node at (1, 0.5) {$t_1$};
        \node at (2, 0.5) {$c$};
        \node at (3, 0.5) {$t_2$};
        \draw (0, -1) to (4, -1);
        \node at (0.5, -1.75) {$a$};
        \node at (1.5, -1.75) {$h$};
        \node at (2.5, -1.75) {$g$};
        \node at (3.5, -1.75) {$e$};
      \end{tikzpicture}.
    \end{aligned}
  \end{equation}
\end{proof}

\subsection{Exchange operators}\label{sec:phases-ops}

	We may now generalize the simple exchange of two anyons,
	allowing for other anyons to be topologically enclosed in the exchange loop.
	Given a representation $\rho\colon B_n \to \sU(\cF)$, 
	with $n \ge p+2$, one may thus consider the 
	\keyword{two-anyon exchange operator}\footnote{\cite{FanGar-10} call the matrices
	$\rho(\sigma_j)$ `elementary braiding operations' (EBO), 
	and we may refer to these as \emph{elementary} exchange operators,
	however we have not
	seen a name convention given to $U_p$ in the literature.}
	corresponding to the braid \eqref{eq:exchange-braid},
	\begin{equation}\label{eq:def-U_p}
		U_p := \rho(\Sigma_p) = \rho(\sigma_1) \rho(\sigma_2) \ldots \rho(\sigma_p) \rho(\sigma_{p+1}) \rho(\sigma_p) \ldots \rho(\sigma_2) \rho(\sigma_1).
	\end{equation}
	However we will also define a more general class of exchange operators
	acting on the standard splitting states.
	
\begin{definition}[Exchange operator]
  Let $U_{t,c,t}$ denote counterclockwise exchange of a pair of anyons of type 
  $t$ around one anyon of type $c$ acting on any 
  corresponding standard splitting state, that is,
  according to \eqref{eq:exchange-t1-t2},
  \begin{equation}\label{eq:def-U_tct}
    U_{t,c,t} :
    \fs{t,c,t}{a,b,d,e}
    \mapsto
    \begin{tikzpicture}[scale=0.4,font=\footnotesize,anchor=mid,baseline={([yshift=-.5ex]current bounding box.center)}]
      \braid s_1^{-1} s_2^{-1} s_1^{-1};
      \node at (1, 0.5) {$t$};
      \node at (2, 0.5) {$c$};
      \node at (3, 0.5) {$t$};
      \draw (0, -3.5) to (4, -3.5);
      \node at (0.5, -4.25) {$a$};
      \node at (1.5, -4.25) {$b$};
      \node at (2.5, -4.25) {$d$};
      \node at (3.5, -4.25) {$e$};
    \end{tikzpicture} =
    \sum_{f,g,h} B^{act}_{d;fb} \, B^{ftt}_{e;gd} \, B^{atc}_{g;hf} \,
    \begin{tikzpicture}[scale=0.4,font=\footnotesize,anchor=mid,baseline={([yshift=-.5ex]current bounding box.center)}]
      \draw (1, 0) to (1, -1);
      \draw (2, 0) to (2, -1);
      \draw (3, 0) to (3, -1);
      \node at (1, 0.5) {$t$};
      \node at (2, 0.5) {$c$};
      \node at (3, 0.5) {$t$};
      \draw (0, -1) to (4, -1);
      \node at (0.5, -1.75) {$a$};
      \node at (1.5, -1.75) {$h$};
      \node at (2.5, -1.75) {$g$};
      \node at (3.5, -1.75) {$e$};
    \end{tikzpicture}.
  \end{equation}
	More generally, if the exchange is done in such a way as to enclose 
	$p$ anyons of type $t_1, \ldots t_p$, 
	we write $U_{t,\{t_1,\ldots,t_p\},t}$,
 \begin{equation}\label{eq:def-U_tpqt}
    U_{t,\{t_1,\ldots,t_p\},t} :
	\fswideflex{t,t_1,t_2}{a_1,a_2,a_3,a_4}{2}
	\ldots
	\fswideflex{t_p,t}{a_{p+1},a_{p+2},a_{p+3}}{3}
    \mapsto
    \begin{tikzpicture}[scale=0.4,font=\footnotesize,anchor=mid,baseline={([yshift=-.5ex]current bounding box.center)}]
      \braid s_1^{-1} s_2^{-1} s_3^{-1} s_2^{-1} s_1^{-1};
      \node at (1, 0.5) {$t$};
      \node at (2, 0.5) {$t_1..$};
      \node at (3, 0.5) {$t_p$};
      \node at (4, 0.5) {$t$};
      \draw (0, -5.5) to (5, -5.5);
      \node at (0.5, -6.25) {$a_1$};
      \node at (1.5, -6.25) {$a_2$};
      \node at (2.7, -6.25) {$\ldots$};
      \node at (4.7, -6.25) {$a_{p+3}$};
    \end{tikzpicture}.
  \end{equation}
	As usual, if an anyon $t$ is repeated $p$ times we write $t^p$,
	and if the anyon type $t$ is understood we denote for brevity 
	$U_p := U_{t,t^p,t}$.
	This is the representation of \eqref{eq:def-U_p} obtained from 
	the splitting basis representation $\rho_{n=p+2}$ in Definition~\ref{def:splitting-rep}.
\end{definition}
	
	We shall only need to compute $U_p$ up to isometries ---
	in fact we only need to know its eigenvalues.
	Therefore, as it turns out, the order of the inner anyons is not important in
	$U_{t,\{t_1,\ldots,t_p\},t}$, only their fusion products.

\begin{theorem}[{Reduction of exchange operator \cite{Qvarfordt-17}}]\label{thm:gen-exchange}
  The exchange operator 
  is given up to similarity by
  \begin{equation}
    U_{t,\{t_1,\ldots,t_p\},t} \ \sim \ \bigoplus_{c \in t_1 \times \ldots \times t_p} U_{t,c,t}
  \end{equation}
  where the direct sum is taken over all fusion channels of the fusion 
  $t_1 \times t_2 \times \ldots \times t_p$. 
  That is, $c$ is a possible result of the fusion $t_1 \times \ldots \times t_p$, 
  counted with multiplicity.
\end{theorem}

\begin{proof}
	We make a change of basis at the base of the braid 
	and use the diagrammatic rules.
	Namely, using a sequence of F-moves, where $t_2$ moves to the same tree as $t_1$,
	$t_3$ to the same as $t_1$ and $t_2$, etc:
\begin{equation}\label{eq:F U_p basis}
  \begin{tikzpicture}[scale=0.18,font=\footnotesize,anchor=mid]
    \node at (-58, -4) {$a_1$};
    \node at (-55, -4) {$a_2$};
    \node at (-43, -4) {$a_{p+3}$};
    \draw (-57, -1) to (-57, -3); 
    \node at (-57, 0) {$t$};
    \draw (-54, -1) to (-54, -3); 
    \node at (-54, 0) {$t_1$};
    \node at (-51, -1.5) {$\cdots$};
    \draw (-48, -1) to (-48, -3); 
    \node at (-48, 0) {$t_p$};
    \draw (-45, -1) to (-45, -3); 
    \node at (-45, 0) {$t$};
    \draw (-60, -3) to (-42, -3); 
    \node[font=\large] at (-38, -1.5) {$\xmapsto{F^{-1}}$};
    \draw (-30, -1) to (-30, -3); 
    \node at (-30, 0) {$t$};
    \draw (-28, 1) to [bend left=-30] (-27, -1);
    \draw (-26, 1) to [bend left=30]  (-27, -1);
    \draw (-27, -1) to (-27, -3); 
    \node at (-28, 2) {$t_1$};
    \node at (-26, 2) {$t_2$};
    \draw (-24, -1) to (-24, -3); 
    \node at (-24, 0) {$t_3$};
    \node at (-21, -1.5) {$\cdots$};
    \draw (-18, -1) to (-18, -3); 
    \node at (-18, 0) {$t_p$};
    \draw (-15, -1) to (-15, -3); 
    \node at (-15, 0) {$t$};
    \draw (-33, -3) to (-12, -3); 
    \node[font=\large] at (-8, -1.5) {$\xmapsto{F^{2-p}}$};
    \draw (-2, 6) to [bend left=-30] (-1, 4);
    \draw ( 0, 6) to [bend left=30]  (-1, 4);
    \draw (-1, 4) to [bend left=-30] (0, 2);
    \draw (1, 4) to [bend left=30]  (0, 2);
    \draw (0, 2) to [bend left=-30] (1, 0);
    \draw (2, 2) to [bend left=30]  (1, 0);
    \draw (1, 0) to (1, -3); 
    \node at (1+0.75, -0.75) {$c$};
    \node[font=\scriptsize] at (3, 3) {$t_p$};
    \node[font=\scriptsize] at (2.4, 5.4) {$t_{p-1}$};
    \node[font=\scriptsize] at (1, 7.5) {$t_{p-2}$};
    \node[rotate=20] at (-2.25, 7) {$\vdots$};
    \draw (-5, -3) to (7, -3); 
    \draw (-2.5, 0) to (-2.5, -3); 
    \draw (4.5, 0) to (4.5, -3); 
    \node at (-2.5+0.75, -0.75) {$t$};
    \node at (4.5+0.75, -0.75) {$t$};
  \end{tikzpicture}
\end{equation}
	The states are then enumerated by the intermediate total charge
	$c \in t_1 \times t_2 \times \ldots \times t_p$.
	In other words we have decomposed the space
	\begin{equation}
		V^{a_1 t t_1 \ldots t_p t}_{a_{p+3}}
		\cong \bigoplus_{c \,\in\, t_1 \times t_2 \times \ldots \times t_p} 
		\widetilde{V}^{t_1 \ldots t_p}_c \otimes V^{a_1 t c t}_{a_{p+3}},
	\end{equation}
	where $\widetilde{V}$ is the space of states formed of successive splitting
	of the intermediate charge $c$ (cf. \eqref{eq:fusion-space-decomp-3})
	into the anyons $t_p,t_{p-1}, \ldots, t_1$, forming a standard splitting tree
	(in this case also known as a `staircase configuration').
	
	We may then use the diagrammatic identities of
	Figure~\ref{fig:fusions} 
	(which actually translate to the hexagon equations 
	\cite[p.196-197]{MooSei-89}, \cite[p.81-82]{Kitaev-06})
	to place the braid at the base state, thereby reducing the operation
	to $U_{t,c,t}$ acting only on the factor $V^{a_1,tct}_{a_{p+3}}$.
	Finally we may use the inverse isomorphism F-moves to transform the
	resulting unbraided state into the original standard form.
	Hence, in our shorthand notation,
	$$
		U_{t,\{t_1,\ldots,t_p\},t} = (F)^{p-1} 
			\bigl( \oplus_c U_{t,c,t} \bigr) (F^{-1})^{p-1}.
	$$
	
	The procedure is similar to the one depicted in \cite[Figure~33]{MooSei-89}
	corresponding to the first $p$ simple braids or first B-move of \eqref{eq:def-U_tct}.
\end{proof}

	We illustrate this theorem with some important examples of anyon models.

\subsection{Abelian anyons}\label{sec:phases-abelian}

	In an arbitrary abelian anyon model with elementary exchange phase
	$\rho(\sigma_j) = e^{i\alpha\pi}$,
	we have the one-dimensional exchange operator (complex phase factor)
	$$
		U_p = U_{\balpha,\balpha^p,\balpha} 
		= \rho(\sigma_1) \ldots \rho(\sigma_p) \rho(\sigma_{p+1}) \rho(\sigma_p) \ldots \rho(\sigma_1)
		= e^{i(1+2p)\alpha\pi}.
	$$
	An alternate derivation is using $B^{abc}_{a+b+c} = R^{bc}_{b+c}$:
	$$
		U_p = U_{\balpha,\balpha^p,\balpha} 
		= B^{1 \balpha^p \balpha}_{\balpha^{p+1};\balpha^p \balpha}
			B^{\balpha^p \balpha \balpha}_{\balpha^{p+2};\balpha^{p+1} \balpha^{p+1}} 
			B^{1 \balpha \balpha^p}_{\balpha^{p+1};\balpha \balpha^p}
		= R^{\balpha^p \balpha}_{\balpha^{p+1}} R^{\balpha \balpha}_{\balpha^2} R^{\balpha \balpha^p}_{\balpha^{p+1}} 
		= e^{i(1+2p)\alpha\pi}.
	$$
	Note that in Theorem~\ref{thm:gen-exchange} 
	we have always a unique result of fusion $\balpha^p = \balpha^{\times p}$.

\begin{figure}
  \centering
  \begin{tikzpicture}[scale=2]
    \draw[->] (-1.25,0) -- (1.25,0) node[below] {Re};
    \draw[->] (0,-1.25) -- (0,1.25) node[right] {Im};
    \draw (0,0) circle (1);
    \node[font=\large] at ({cos(deg(1*pi/5))}, {sin(deg(1*pi/5))}) {$\bullet$};
    \node[above right] at ({cos(deg(1*pi/5))}, {sin(deg(1*pi/5))}) {$e^{i\pi/5}$ {\small ($p \equiv 3$)}};
    \node[font=\large] at ({cos(deg(3*pi/5))}, {sin(deg(3*pi/5))}) {$\bullet$};
    \node[above left] at ({cos(deg(3*pi/5))}, {sin(deg(3*pi/5))}) {$e^{i3\pi/5}$ {\small ($p \equiv 0$)}};
    \node[font=\large] at ({cos(deg(5*pi/5))}, {sin(deg(5*pi/5))}) {$\bullet$};
    \node[above left] at ({cos(deg(5*pi/5))}, {sin(deg(5*pi/5))}) {$e^{i\pi}$ {\small ($p \equiv 2$)}};
    \node[font=\large] at ({cos(deg(7*pi/5))}, {sin(deg(7*pi/5))}) {$\bullet$};
    \node[below left] at ({cos(deg(7*pi/5))}, {sin(deg(7*pi/5))}) {$e^{i7\pi/5}$ {\small ($p \equiv 4$)}};
    \node[font=\large] at ({cos(deg(9*pi/5))}, {sin(deg(9*pi/5))}) {$\bullet$};
    \node[below right] at ({cos(deg(9*pi/5))}, {sin(deg(9*pi/5))}) {$e^{i9\pi/5}$ {\small ($p \equiv 1$)}};
  \end{tikzpicture}
  \caption{Exchange phases $U_p$, $p$ modulo $5$, on the complex unit circle 
	for abelian anyons with $\alpha=3/5$.}
  \label{fig:abelian-eigenvals}
\end{figure}

	More generally we could also have considered a \emph{reducible} abelian model:
	$$
		\rho(\sigma_j) = S^{-1} \diag(e^{i\alpha_1\pi},\ldots,e^{i\alpha_D\pi}) S,
		\quad j = 1,\ldots,N-1, \quad S \in \sU(D),
	$$
	i.e. $\rho \sim \rho^{\alpha_1} \oplus \cdots \oplus \rho^{\alpha_D}$,
	for which
	$$
		U_p = S^{-1} \diag(e^{i(1+2p)\alpha_1\pi},\ldots,e^{i(1+2p)\alpha_D\pi}) S,
		\quad p \ge 0.
	$$
	
	Note that this formula also provides a way to test if a given model is non-abelian:
	
\begin{proposition}\label{prop:test-non-abelian}
	Let $\rho\colon B_N \to \sU(D)$ be a representation with
	$\rho(\sigma_j) \sim \diag(e^{i\alpha_1\pi},\ldots\\,e^{i\alpha_D\pi})$
	(the similarity here is not necessarily uniform in $j$).
	If, for some $p \ge 1$,
	$$
		U_p \nsim \diag(e^{i(1+2p)\alpha_1\pi},\ldots,e^{i(1+2p)\alpha_D\pi})
	$$
	then $\rho$ is non-abelian, i.e. there exist $j \neq k$ so that
	$[\rho(\sigma_j),\rho(\sigma_k)] \neq 0$.
\end{proposition}

	The proof is immediate from the above observation for arbitrary abelian $\rho$.

\subsection{Fibonacci anyons}\label{sec:phases-fib}

	In the case of Fibonacci anyons the exchange operator is
	\begin{equation}
		U_p = U_{\tau,\tau^p,\tau}.
	\end{equation}
	Since $p$ anyons may fuse either to $1$ or $\tau$ with multiplicities
	\begin{equation}
		\tau^p = \Fib(p-1) \cdot 1 + \Fib(p) \cdot \tau,
	\end{equation}
	by Theorem~\ref{thm:gen-exchange} we have
	\begin{equation}
		U_{\tau,\tau^p,\tau} \sim U_{\tau,1,\tau}^{\oplus \Fib(p-1)} \oplus U_{\tau,\tau,\tau}^{\oplus \Fib(p)}.
	\end{equation}
	We compute the respective blocks in this decomposition:

\subsubsection{\texorpdfstring{$U_{\tau,1,\tau}$}{U\_{\tau,1,\tau}}}

	Begin by observing $\fs{\tau,1,\tau}{} = \fs{\tau,\tau}{}$.
	This exchange operator acts on the space with (ordered) basis
	\begin{equation}
	  \left\{
	    \fs{\tau,\tau}{1,\tau,1}, \ 
	    \fs{\tau,\tau}{1,\tau,\tau}, \ 
	    \fs{\tau,\tau}{\tau,\tau,1}, \ 
	    \fs{\tau,\tau}{\tau,1,\tau},
	    \fs{\tau,\tau}{\tau,\tau,\tau}
	  \right\}.
	\end{equation}
	By \eqref{eq:B-fib-1}--\eqref{eq:B-fib-2} we find
	(let us again drop some indices for simplicity)
	\begin{equation}
	\begin{aligned}
      U_0 = \rho_2(\sigma_1) &=
      \begin{bmatrix}
        B^{1\tau\tau}_{1;\tau\tau} & & & & \\
        & B^{1\tau\tau}_{\tau;\tau\tau} & & & \\
        & & B^{\tau\tau\tau}_{1;\tau\tau} & & \\
        & & & B^{\tau\tau\tau}_{\tau;11} & B^{\tau\tau\tau}_{\tau;1\tau} \\
        & & & B^{\tau\tau\tau}_{\tau;\tau1} & B^{\tau\tau\tau}_{\tau;\tau\tau}
      \end{bmatrix} 
      =
      \begin{bmatrix}
        R^{\tau\tau}_1 & & & & \\
        & R^{\tau\tau}_\tau & & & \\
        & & R^{\tau\tau}_\tau & \\
        & & & B_{1 1} & B_{1 \tau} \\
        & & & B_{\tau 1} & B_{\tau \tau}
      \end{bmatrix} \\
      &= R^{\tau\tau}_1 \oplus R^{\tau\tau}_\tau \oplus R^{\tau\tau}_\tau \oplus B.
	\end{aligned}
	\end{equation}
	Taking the first conjugation convention \eqref{eq:R-fib-A}, 
	we observe by $B \sim R$ that 
	\begin{equation}\label{eq:spec-B-fib}
		\mathrm{spec}(B) =:
		\sigma(B) = \sigma(R) = \{ R^{\tau\tau}_1, R^{\tau\tau}_\tau \}
		= \{e^{4\pi i/5}, e^{-3\pi i/5}\},
	\end{equation}
	and thus
	\begin{equation}\label{eq:spec-U0-fib}
		\sigma(U_0) 
		= \bigl\{ R^{\tau\tau}_1 \text{(mult.\,2)}, R^{\tau\tau}_\tau \text{(mult.\,3)} \bigr\}
		= \bigl\{ e^{4\pi i/5} \text{(mult.\,2)}, e^{-3\pi i/5} \text{(mult.\,3)} \bigr\}.
	\end{equation}

\subsubsection{\texorpdfstring{$U_{\tau,\tau,\tau}$}{U\_{\tau,\tau,\tau}}}

	This exchange operator acts on the space with (ordered) basis
	\begin{equation}
	  \left\{
	    \fs{\tau,\tau,\tau}{1,\tau,\tau,1}, \ 
	    \fs{\tau,\tau,\tau}{1,\tau,1,\tau},
	    \fs{\tau,\tau,\tau}{1,\tau,\tau,\tau}, \ 
	    \fs{\tau,\tau,\tau}{\tau,1,\tau,1},
	    \fs{\tau,\tau,\tau}{\tau,\tau,\tau,1}, \ 
	    \fs{\tau,\tau,\tau}{\tau,1,\tau,\tau},
	    \fs{\tau,\tau,\tau}{\tau,\tau,1,\tau},
	    \fs{\tau,\tau,\tau}{\tau,\tau,\tau,\tau}
	  \right\}.
	\end{equation}
	Grouping the basis into different charge sectors, we obtain
	\begin{alignat*}{10}
		\rho_3(\sigma_1) &= R^{\tau\tau}_\tau &{}\oplus{}& R &{}\oplus{}& B &{}\oplus{}&
	    \begin{bmatrix}
	      B_{11} & & B_{1\tau} \\
	      & R^{\tau\tau}_{\tau} \\
	      B_{\tau1} & & B_{\tau\tau}
	    \end{bmatrix} \\
	    \rho_3(\sigma_2) &= R^{\tau\tau}_\tau &{}\oplus{}& B &{}\oplus{}& R &{}\oplus{}&
	    \begin{bmatrix}
	      R^{\tau\tau}_\tau \\
	      & B_{11} & B_{1\tau} \\
	      & B_{\tau1} & B_{\tau\tau}
	    \end{bmatrix}
	\end{alignat*}
	\begin{align*}
      U_1 = \rho_3(\sigma_1) \rho_3(\sigma_2) \rho_3(\sigma_1) 
      = \left( R^{\tau\tau}_\tau \right)^3 \oplus \left( RBR \right) \oplus \left( BRB \right)
      \oplus M, 
	\end{align*}
	where $BRB \sim RBR = RFRFR$ and
	\begin{equation}
		M = \begin{bmatrix}
			R^{\tau\tau}_\tau \left(B_{11}\right)^2+B_{1\tau} B_{\tau1} B_{\tau\tau} & B_{1\tau} B_{\tau 1} R^{\tau\tau}_\tau & B_{1\tau} \left(B_{\tau\tau}\right)^2+B_{11} B_{1\tau} R^{\tau\tau}_\tau \\
			B_{1\tau} B_{\tau 1} R^{\tau\tau}_\tau & B_{11} \left(R^{\tau\tau}_\tau\right)^2 & B_{1\tau} B_{\tau\tau} R^{\tau\tau}_\tau \\
			B_{\tau 1} \left(B_{\tau\tau}\right)^2+B_{11} B_{\tau 1} R^{\tau\tau}_\tau & B_{\tau 1} B_{\tau\tau} R^{\tau\tau}_\tau & \left(B_{\tau\tau}\right)^3+B_{1\tau} B_{\tau 1} R^{\tau\tau}_\tau \\
		\end{bmatrix}.
	\end{equation}
	
	In this case we obtain (as can be verified simply by computer algebra)
	\begin{align*}
		(R^{\tau\tau}_\tau)^3 &= e^{\pi i/5}, \\
		\sigma(RBR) = \sigma(BRB) &= \{ e^{4\pi i/5}, e^{-\pi i/5} \}, \\
		\sigma(M) &= \{ e^{4\pi i/5}, e^{\pi i/5}, e^{-\pi i/5} \},
	\end{align*}
	and hence
	\begin{equation}\label{eq:spec-U1-fib}
		\sigma(U_1) 
		= \bigl\{ 
			e^{4\pi i/5} \text{(mult.\,3)}, 
			e^{\pi i/5} \text{(mult.\,2)},
			e^{-\pi i/5} \text{(mult.\,3)}
			\bigr\}.
	\end{equation}

\begin{corollary}[{\cite{Qvarfordt-17}}]\label{cor:Up-fib}
  The eigenvalues of $U_p$ for Fibonacci anyons are
  (cf.\ Figure~\ref{fig:fib-eigenvals})
  \begin{equation}
      \sigma(U_p) =
      \left\{
      \begin{array}{rl}
        e^{i4\pi/5}  &\text{\rm with multiplicity }\,\; \Fib(p+3), \\
        e^{i\pi/5}   &\text{\rm with multiplicity }\,\; 2\Fib(p), \\
        e^{-i\pi/5}  &\text{\rm with multiplicity }\,\; 3\Fib(p), \\
        e^{-i3\pi/5} &\text{\rm with multiplicity }\,\; 3\Fib(p-1)
      \end{array}
      \right\}.
  \end{equation}
  Reducing to $U_p|_{V^{\tau^{2+p}}_*}$
  yields the same eigenvalues with multiplicity 
  $\Fib(p+1)$, $\Fib(p)$, $\Fib(p)$, respectively $\Fib(p-1)$.
\end{corollary}

\begin{figure}
  \centering
  \begin{tikzpicture}[scale=2]
    \draw[->] (-1.25,0) -- (1.25,0) node[below] {Re};
    \draw[->] (0,-1.25) -- (0,1.25) node[right] {Im};
    \draw (0,0) circle (1);
    \node[font=\large] at ({cos(deg(pi/5))}, {sin(deg(pi/5))}) {$\bullet$};
    \node[above right] at ({cos(deg(pi/5))}, {sin(deg(pi/5))}) {$e^{i\pi/5}$ {\small ($p \ge 1$)}};
    \node[font=\large] at ({cos(-deg(pi/5))}, {sin(-deg(pi/5))}) {$\bullet$};
    \node[below right] at ({cos(-deg(pi/5))}, {sin(-deg(pi/5))}) {$e^{-i\pi/5}$ {\small ($p \ge 1$)}};
    \node[font=\large] at ({cos(deg(4*pi/5))}, {sin(deg(4*pi/5))}) {$\bullet$};
    \node[above left]  at ({cos(deg(4*pi/5))}, {sin(deg(4*pi/5))}) {$e^{i4\pi/5}$ {\small ($p \ge 0$)}};
    \node[font=\large] at ({cos(-deg(3*pi/5))}, {sin(-deg(3*pi/5))}) {$\bullet$};
    \node[below left]  at ({cos(-deg(3*pi/5))}, {sin(-deg(3*pi/5))}) {$e^{-i3\pi/5}$ {\small ($p \neq 1$)}};
  \end{tikzpicture}
  \caption{Eivenvalues of $U_p$ on the complex unit circle for Fibonacci anyons
  	(in one convention, another is obtained by complex conjugation).}
  \label{fig:fib-eigenvals}
\end{figure}

	Note that the above expressions are valid for all $p \ge 0$ 
	since $\Fib(p)$ is defined also for negative $p$.

\subsection{Ising anyons}\label{sec:phases-ising}

	Let us compute the exchange matrices
	$U_{\sigma,\{\sigma^p,\psi^q\},\sigma}$ resp. $U_{\psi,\{\sigma^p,\psi^q\},\psi}$.

	The fusion multiplicities are
	\begin{equation}
	    \sigma^{2n+1} = 2^{n} \cdot \sigma, \qquad
	    \sigma^{2n}   = 2^{n-1} \cdot (1 + \psi),
	\end{equation}
	and hence by Theorem~\ref{thm:gen-exchange} we have
	(note that these act on differently sized spaces)
	\begin{equation}
	  \begin{aligned}
	    U_{\sigma,\sigma^{2n+1},\sigma} &\sim U_{\sigma,\sigma,\sigma}^{\oplus 2^n}, \\
	    U_{\sigma,\sigma^{2n},\sigma} &\sim U_{\sigma,1,\sigma}^{\oplus 2^{n-1}} \oplus U_{\sigma,\psi,\sigma}^{\oplus 2^{n-1}}, \\
	    U_{\psi,\sigma^{2n+1},\psi} &\sim U_{\psi,\sigma,\psi}^{\oplus 2^n}, \\
	    U_{\psi,\sigma^{2n},\psi} &\sim U_{\psi,1,\psi}^{\oplus 2^{n-1}} \oplus U_{\psi,\psi,\psi}^{\oplus 2^{n-1}}.
	  \end{aligned}
	\end{equation}
	The exchanges involving $\psi$ are easiest to compute as they are all diagonal:
	\begin{equation}
	  \begin{aligned}
	    U_{\psi,1,\psi} &= R^{\psi\psi}_1 \1 = -\1, \\
	    U_{\psi,\psi,\psi} &= R^{\psi\psi}_1 \1 = -\1, \\
	    U_{\psi,\sigma,\psi} &= R^{\sigma\psi}_\sigma R^{\psi\psi}_1 R^{\sigma\psi}_\sigma \1 = \1, \\
	  \end{aligned}
	\end{equation}
	i.e. for all valid labels $a,b,c,d \in \{1,\sigma,\psi\}$,
	\begin{equation}\label{eq:psi-braiding}
		\fs[1]{\psi,\psi}{a,b,c} = - \fs{\psi,\psi}{a,b,c}, \qquad 
		\begin{tikzpicture}[scale=0.4,font=\footnotesize,anchor=mid,baseline={([yshift=-.5ex]current bounding box.center)}]
			\braid s_1^{-1} s_2^{-1} s_1^{-1};
			\node at (1, 0.5) {$\psi$};
			\node at (2, 0.5) {$\psi$};
			\node at (3, 0.5) {$\psi$};
			\draw (0, -3.5) to (4, -3.5);
			\node at (0.5, -4.25) {$a$};
			\node at (1.5, -4.25) {$b$};
			\node at (2.5, -4.25) {$c$};
			\node at (3.5, -4.25) {$d$};
		\end{tikzpicture} =
		-\fs{\psi,\psi,\psi}{a,b,c,d}, \qquad
		\begin{tikzpicture}[scale=0.4,font=\footnotesize,anchor=mid,baseline={([yshift=-.5ex]current bounding box.center)}]
			\braid s_1^{-1} s_2^{-1} s_1^{-1};
			\node at (1, 0.5) {$\psi$};
			\node at (2, 0.5) {$\sigma$};
			\node at (3, 0.5) {$\psi$};
			\draw (0, -3.5) to (4, -3.5);
			\node at (0.5, -4.25) {$a$};
			\node at (1.5, -4.25) {$b$};
			\node at (2.5, -4.25) {$c$};
			\node at (3.5, -4.25) {$d$};
		\end{tikzpicture} =
		+\fs{\psi,\sigma,\psi}{a,b,c,d}.
	\end{equation}
	For the ones involving braiding of $\sigma$'s we obtain off-diagonal matrices however:

\subsubsection{\texorpdfstring{$U_{\sigma,1,\sigma}$}{U\_{\sigma,1,\sigma}}}

	This exchange operator acts on the space 
	$\fs{\sigma,1,\sigma}{} = \fs{\sigma,\sigma}{}$
	with (ordered) basis
	\begin{equation}
	  \left\{
	    \fs{\sigma,\sigma}{1,\sigma,1},
	    \fs{\sigma,\sigma}{1,\sigma,\psi},
	    \fs{\sigma,\sigma}{\sigma,1,\sigma},
	    \fs{\sigma,\sigma}{\sigma,\psi,\sigma},
	    \fs{\sigma,\sigma}{\psi,\sigma,1},
	    \fs{\sigma,\sigma}{\psi,\sigma,\psi}
	  \right\}.
	\end{equation}
	By \eqref{eq:B-ising-1}--\eqref{eq:B-ising-3} we have immediately
	\begin{equation}
	  \begin{aligned}
	    U_{\sigma,1,\sigma} = \rho_2(\sigma_1)
	    &= \begin{bmatrix}
	      R^{\sigma\sigma}_{1} \\
	      & R^{\sigma\sigma}_{\psi} \\
	      & & B^{\sigma\sigma\sigma}_{\sigma;11} & B^{\sigma\sigma\sigma}_{\sigma;1\psi} \\
	      & & B^{\sigma\sigma\sigma}_{\sigma;\psi 1} & B^{\sigma\sigma\sigma}_{\sigma;\psi\psi} \\
	      & & & & R^{\sigma\sigma}_\psi \\
	      & & & & & R^{\sigma\sigma}_1  
	    \end{bmatrix}
	    = R \oplus B \oplus \tilde{R},
	  \end{aligned}  
	\end{equation}
	where $\tilde{R} := R_\psi \oplus R_1$.
	Since $\sigma(B) = \sigma(R) = \sigma(\tilde{R}) = \{R_1,R_\psi\}$,
	we find
	\begin{equation}
		\sigma( U_{\sigma,1,\sigma} )
		= \bigl\{ 
			R_1 \text{(mult.\,3)}, 
			R_\psi \text{(mult.\,3)}
			\bigr\}
		= \bigl\{ 
			e^{-\pi i/8} \text{(mult.\,3)}, 
			e^{3\pi i/8} \text{(mult.\,3)}
			\bigr\}.
	\end{equation}

\subsubsection{\texorpdfstring{$U_{\sigma,\sigma,\sigma}$}{U\_{\sigma,\sigma,\sigma}}}

	This exchange operator acts on the space with (ordered) basis
	\begin{equation}
	  \begin{aligned}
	    \left\{
	      \fs{\sigma,\sigma,\sigma}{1,\sigma,1,\sigma},
	      \fs{\sigma,\sigma,\sigma}{1,\sigma,\psi,\sigma},
	      \fs{\sigma,\sigma,\sigma}{\sigma,1,\sigma,1},
	      \fs{\sigma,\sigma,\sigma}{\sigma,\psi,\sigma,1},
	      \fs{\sigma,\sigma,\sigma}{\sigma,1,\sigma,\psi},
	      \fs{\sigma,\sigma,\sigma}{\sigma,\psi,\sigma,\psi},
	      \fs{\sigma,\sigma,\sigma}{\psi,\sigma,1,\sigma},
	      \fs{\sigma,\sigma,\sigma}{\psi,\sigma,\psi,\sigma}
	    \right\}.
	  \end{aligned}
	\end{equation}
	Again, by grouping into charge sectors,
	
\centerline{
  \begin{minipage}{\linewidth}
    \begin{align*}
      \rho_3(\sigma_1)
      &= \begin{bmatrix}
        R^{\sigma\sigma}_1 \\
        & R^{\sigma\sigma}_\psi \\
        & & B^{\sigma\sigma\sigma}_{\sigma;11} & B^{\sigma\sigma\sigma}_{\sigma;1\psi} \\
        & & B^{\sigma\sigma\sigma}_{\sigma;\psi 1} & B^{\sigma\sigma\sigma}_{\sigma;\psi\psi} \\
        & & & & B^{\sigma\sigma\sigma}_{\sigma;11} & B^{\sigma\sigma\sigma}_{\sigma;1\psi} \\
        & & & & B^{\sigma\sigma\sigma}_{\sigma;\psi 1} & B^{\sigma\sigma\sigma}_{\sigma;\psi\psi} \\
        & & & & & & R^{\sigma\sigma}_\psi \\
        & & & & & & & R^{\sigma\sigma}_1 
      \end{bmatrix} 
      = R \oplus B \oplus B \oplus \tilde{R}, \\
      \rho_3(\sigma_2)
      &= \begin{bmatrix}
        B^{\sigma\sigma\sigma}_{\sigma;11} & B^{\sigma\sigma\sigma}_{\sigma;1\psi} \\
        B^{\sigma\sigma\sigma}_{\sigma;\psi 1} & B^{\sigma\sigma\sigma}_{\sigma;\psi\psi} \\
        & & R^{\sigma\sigma}_1 \\
        & & & R^{\sigma\sigma}_\psi \\
        & & & & R^{\sigma\sigma}_\psi \\
        & & & & & R^{\sigma\sigma}_1 \\ 
        & & & & & & B^{\sigma\sigma\sigma}_{\sigma;11} & B^{\sigma\sigma\sigma}_{\sigma;1\psi} \\
        & & & & & & B^{\sigma\sigma\sigma}_{\sigma;\psi 1} & B^{\sigma\sigma\sigma}_{\sigma;\psi\psi}
      \end{bmatrix}
      = B \oplus R \oplus \tilde{R} \oplus B.
    \end{align*}
  \end{minipage}
}
	\noindent
	Hence
	\begin{equation}
		U_{\sigma,\sigma,\sigma} 
		= \rho_3(\sigma_1) \rho_3(\sigma_2) \rho_3(\sigma_1) 
		= RBR \oplus BRB \oplus B\tilde{R}B \oplus \tilde{R}B\tilde{R},
	\end{equation}
	and using that $BRB \sim RBR = RFRFR$ and that (by computation)
	\begin{equation}
		\sigma(BRB) = \sigma(B\tilde{R}B) = \sigma(\tilde{R}B\tilde{R}) = \{ e^{-\pi i/8}, e^{7\pi i/8} \},
	\end{equation}
	we obtain
	\begin{equation}
		\sigma( U_{\sigma,\sigma,\sigma} )
		= \bigl\{ 
			e^{-\pi i/8} \text{(mult.\,4)}, 
			e^{7\pi i/8} \text{(mult.\,4)}
			\bigr\}.
	\end{equation}

\subsubsection{\texorpdfstring{$U_{\sigma,\psi,\sigma}$}{U\_{\sigma,\psi,\sigma}}}

	This exchange operator acts on the space with (ordered) basis
	\begin{equation}
	  \begin{aligned}
	    \left\{
	      \fs{\sigma,\psi,\sigma}{1,\sigma,\sigma,1},
	      \fs{\sigma,\psi,\sigma}{1,\sigma,\sigma,\psi},
	      \fs{\sigma,\psi,\sigma}{\psi,\sigma,\sigma,1},
	      \fs{\sigma,\psi,\sigma}{\psi,\sigma,\sigma,\psi},
	      \fs{\sigma,\psi,\sigma}{\sigma,1,\psi,\sigma},
	      \fs{\sigma,\psi,\sigma}{\sigma,\psi,1,\sigma}
	    \right\}.
	  \end{aligned}
	\end{equation}
	Here we use Lemma~\ref{lem:Up-braid} on each of these basis states, 
	and \eqref{eq:B-ising-1}-\eqref{eq:B-ising-3}, to find
    \begin{align*}
      U_{\sigma,\psi,\sigma}
      &= \diag \bigl( 
      	B^{1\psi\sigma}_{\sigma;\psi\sigma} B^{\psi\sigma\sigma}_{1;\sigma\sigma} B^{1\sigma\psi}_{\sigma;\sigma\psi}, 
      	B^{1\psi\sigma}_{\sigma;\psi\sigma} B^{\psi\sigma\sigma}_{\psi;\sigma\sigma} B^{1\sigma\psi}_{\sigma;\sigma\psi},
      	B^{\psi\psi\sigma}_{\sigma;1\sigma} B^{1\sigma\sigma}_{1;\sigma\sigma} B^{\psi\sigma\psi}_{\sigma;\sigma 1},
      	B^{\psi\psi\sigma}_{\sigma;1\sigma} B^{1\sigma\sigma}_{\psi;\sigma\sigma} B^{\psi\sigma\psi}_{\sigma;\sigma 1}
      	\bigr) \oplus \tilde{B} \\
      &= \diag \bigl( e^{-5\pi i/8}, e^{7\pi i/8}, e^{7\pi i/8}, e^{-5\pi i/8} \bigr) \oplus \tilde{B},
    \end{align*}
	where the last two states yield the off-diagonal part
	$$
		\tilde{B} := \begin{bmatrix}
			B^{\sigma\psi\sigma}_{\psi;\sigma 1} B^{\sigma\sigma\sigma}_{\sigma;\psi\psi} B^{\sigma\sigma\psi}_{\psi;1\sigma} & B^{\sigma\psi\sigma}_{\psi;\sigma 1} B^{\sigma\sigma\sigma}_{\sigma;1\psi} B^{\sigma\sigma\psi}_{1;\psi\sigma} \\
			B^{\sigma\psi\sigma}_{1;\sigma\psi} B^{\sigma\sigma\sigma}_{\sigma;\psi 1} B^{\sigma\sigma\psi}_{\psi;1\sigma} & B^{\sigma\psi\sigma}_{1;\sigma\psi} B^{\sigma\sigma\sigma}_{\sigma;1 1} B^{\sigma\sigma\psi}_{1;\psi\sigma}
		\end{bmatrix}
		= \frac{1}{\sqrt{2}} \begin{bmatrix}
			e^{-7\pi i/8} & e^{-3\pi i/8} \\
			e^{-3\pi i/8} & e^{-7\pi i/8}
		\end{bmatrix},
	$$
	with $\sigma(\tilde{B}) = \{ e^{-5\pi i/8}, e^{7\pi i/8} \}$.
	Hence
	\begin{equation}
		\sigma( U_{\sigma,\psi,\sigma} )
		= \bigl\{ 
			e^{-5\pi i/8} \text{(mult.\,3)}, 
			e^{7\pi i/8} \text{(mult.\,3)}
			\bigr\}.
	\end{equation}
	
	In summary,

\begin{figure}
  \centering
  \begin{tikzpicture}[scale=2]
    \draw[->] (-1.25,0) -- (1.25,0) node[below] {Re};
    \draw[->] (0,-1.25) -- (0,1.25) node[right] {Im};
    \draw (0,0) circle (1);
    \node[font=\large] at ({cos(-deg(5*pi/8))}, {sin(-deg(5*pi/8))}) {$\bullet$};
    \node[below left] at ({cos(-deg(5*pi/8))}, {sin(-deg(5*pi/8))}) {$e^{-i5\pi/8}$ {\small (even $p \ge 2$)}};
    \node[font=\large] at ({cos(-deg(pi/8))}, {sin(-deg(pi/8))}) {$\bullet$};
    \node[below right] at ({cos(-deg(pi/8))}, {sin(-deg(pi/8))}) {$e^{-i\pi/8}$ {\small (all $p$)}};
    \node[font=\large] at ({cos(deg(3*pi/8))}, {sin(deg(3*pi/8))}) {$\bullet$};
    \node[above right] at ({cos(deg(3*pi/8))}, {sin(deg(3*pi/8))}) {$e^{i3\pi/8}$ {\small (even $p \ge 0$)}};
    \node[font=\large] at ({cos(deg(7*pi/8))}, {sin(deg(7*pi/8))}) {$\bullet$};
    \node[above left] at ({cos(deg(7*pi/8))}, {sin(deg(7*pi/8))}) {$e^{i7\pi/8}$ {\small (all $p \ge 1$)}};
  \end{tikzpicture}
  \caption{Eivenvalues of $U_p$ on the complex unit circle for Ising anyons.}
  \label{fig:ising-eigenvals}
\end{figure}

\begin{corollary}\label{cor:Up-ising}
	The eigenvalues of $U_p = U_{\sigma,\sigma^p,\sigma}$ for Ising anyons are
	(cf.\ Figure~\ref{fig:ising-eigenvals})
	\begin{equation}
		\sigma( U_{\sigma,1,\sigma} )
		= \bigl\{ e^{-\pi i/8}, e^{3\pi i/8} \bigr\},
		\qquad \text{each with multiplicity $3$},
	\end{equation}
	and for $n \ge 1$
	\begin{equation}
		\sigma( U_{\sigma,\sigma^{2n},\sigma} )
		= \bigl\{ e^{-5\pi i/8}, e^{-\pi i/8}, e^{3\pi i/8}, e^{7\pi i/8}
			\bigr\},
		\qquad \text{each with multiplicity $3 \cdot 2^{n-1}$},
	\end{equation}
	respectively, for $n \ge 0$
	\begin{equation}
		\sigma( U_{\sigma,\sigma^{2n+1},\sigma} )
		= \bigl\{ e^{-\pi i/8}, e^{7\pi i/8} \bigr\},
		\qquad \text{each with multiplicity $2^{n+2}$}.
	\end{equation}
	Reducing to $U_p|_{V^{\sigma^{2+p}}_*}$
	yields the same eigenvalues with multiplicity 
	$1$, $2^{n-1}$, respectively $2^n$.
\end{corollary}

\subsection{Clifford anyons}\label{sec:phases-clifford}

	Another anyonic model that has been considered in the literature is
	obtained by taking a composition of Ising anyons 
	with an abelian factor of $\alpha = -1/8$:
	$$
		\rho^{\text{Clifford}}(\sigma_j) := \rho^{\alpha = -1/8}(\sigma_j) \rho^{\text{Ising}}(\sigma_j).
	$$
	Because the resulting representation may alternatively 
	be defined using Clifford algebras and spinors \cite{NayWil-96,Ivanov-01},
	these anyons may be referred to as \keyword{Clifford anyons}
	(or ``half-quantum vortices from Majorana fermions'', see e.g. \cite[Ch.\,4.2.2]{Tong-16},
	and \cite[Theorem~4.8]{DelRowWan-16}).
	
	Using 
	$$
		U_p = U_p^{\text{Clifford}} = e^{-(2p+1) \pi i/8} U_p^{\text{Ising}},
	$$ 
	we find
	\begin{align*}
		\sigma( U_0 ) &= \bigl\{ e^{-\pi i/4}, e^{\pi i/4} \bigr\}, \\
		\sigma( U_1 ) &= \bigl\{ e^{-\pi i/2}, e^{\pi i/2} \bigr\} = \{ \pm i \}, \\
		\sigma( U_2 ) &= \bigl\{ e^{-3\pi i/4}, e^{-\pi i/4}, e^{\pi i/4}, e^{3\pi i/4} \bigr\}, \\
		\sigma( U_3 ) &= \bigl\{ \pm 1 \bigr\}, \\
		\sigma( U_4 ) &= \bigl\{ e^{-3\pi i/4}, e^{-\pi i/4}, e^{\pi i/4}, e^{3\pi i/4} \bigr\}, \\
		\sigma( U_5 ) &= \bigl\{ \pm i \bigr\},
	\end{align*}
	and so on, with $\sigma(U_{4k+l}) = \sigma(U_l)$, $k \ge 1$, $l = 1,2,3,4$.
	The first few of these $U_p$ have also been computed in \cite{BraSon-18,Wiklund-18} 
	using representations of Clifford algebras.

\subsection{Burau anyons}\label{sec:phases-burau}

	The standard (unreduced) \keyword{Burau representation} \cite{Burau-35}
	is simply a deformation of the defining representation of permutations on $\C^N$ 
	with deformation parameter $z \in \C \setminus \{0\}$:
	$$
		\rho: B_N \to \mathrm{GL}(\C^N)
	$$
	\begin{equation}\label{eq:Burau-sigma}
		\rho(\sigma_j) = \1_{j-1} \oplus \begin{bmatrix} 1-z & z \\ 1 & 0 \end{bmatrix} \oplus \1_{N-j-1}.
	\end{equation}
	Thus it has rank $N$, 
	but may, depending on $z$, be reduced either once or twice to an irreducible 
	representation of rank $N-1$ or $N-2$.
	At the outset these are not unitary but for $z \in \sU(1)$ close enough to 1 they may be unitarized 
	\cite{Squier-84}.
	For more details we refer to e.g. \cite{Weinberger-15}, 
	\cite[Sec.~2.3]{DelRowWan-16}.
	Note that these are different than the algebraic anyon models considered above
	since their dimension grows linearly with $N$
    (rather than exponentially).
	
	As a simplest example we may consider $N=3$, reduced and unitarized:
	$$
		\rho: B_3 \to \sU(2)
	$$
	\begin{align*}
		\rho(\sigma_1) &= \frac{1}{2} \begin{bmatrix} -w^2+w+1 & -w\sqrt{w+w^{-1}-1}\sqrt{w+w^{-1}+1} \\ -w\sqrt{w+w^{-1}-1}\sqrt{w+w^{-1}+1} & -w^2-w+1 \end{bmatrix}, 
		\\
		\rho(\sigma_2) &= \frac{1}{2} \begin{bmatrix} -w^2+w+1 & +w\sqrt{w+w^{-1}-1}\sqrt{w+w^{-1}+1} \\ +w\sqrt{w+w^{-1}-1}\sqrt{w+w^{-1}+1} & -w^2-w+1 \end{bmatrix},
	\end{align*}
	where $w=\sqrt{z}=e^{i\pi\alpha}$ and $|\alpha| < 1/3$; see \cite{Weinberger-15}.
	We find that
	$$
		\sigma(U_0) = \sigma(\rho(\sigma_1)) = \sigma(\rho(\sigma_2)) = \{1,-w^2\}
	$$
	while
	$$
		U_1 = \rho(\sigma_1\sigma_2\sigma_1) = \rho(\sigma_2\sigma_1\sigma_2) = \diag(w^3,-w^3).
	$$
	Proposition~\ref{prop:test-non-abelian} then verifies that $\rho$ is necessarily
	non-abelian for these $w$ (except possibly at $w=1$), 
	but it may also be verified directly (also at $w=1$) by the identity
	$$
		\rho(\sigma_1\sigma_2) - \rho(\sigma_2\sigma_1) = \begin{bmatrix} 0 & w^2\sqrt{w+w^{-1}-1}\sqrt{w+w^{-1}+1} \\ -w^2\sqrt{w+w^{-1}-1}\sqrt{w+w^{-1}+1} & 0 \end{bmatrix}.
	$$

	In fact, a central result in the theory of braid group representations is 
	that the rank of any non-abelian representation must grow at least
	linearly with $N$, and that the Burau representations are of the smallest rank:
	
\begin{theorem}[{Formanek \cite{Formanek-96}}]\label{thm:small-reps-abelian}
	Let $\rho\colon B_N \to \sU(D)$ be a unitary braid group representation,
	and $N>6$.
	If $D < N-2$, then $\rho$ is abelian. 
	Furthermore, if $D = N-2$ or $D = N-1$ then $\rho$ is either abelian or of 
	Burau type (i.e.\ equivalent to the product of an abelian and a reduced Burau representation).
\end{theorem}
	The proof is given in \cite[Theorem~23]{Formanek-96} for irreducible
	representations $\rho\colon B_N \to \mathrm{GL}(\C^D)$, 
	and the above formulation follows from 
	the complete reducibility of any unitary representation 
	(see Lemma~\ref{lem:similar-generators}). 
	Also the case $D \le N$ for arbitrary $N \ge 3$ 
	is discussed in \cite[Theorem~6.1]{Formanek-etal-03}.
	See \cite{Weinberger-15} for some further remarks.


\section{Geometric and magnetic anyon models}\label{sec:ham}

	The general definition of the Hilbert space of anyon models
	was discussed in some detail in \cite{MunSch-95},
	and the choice of a self-adjoint Hamiltonian in \cite{BakCanMulSun-93,DelFigTet-97,LunSol-14}.
	See also \cite{Souriau-70,FroMar-89,MueDoe-93}, \cite[Appendix~A]{KorLanSch-99}, 
	\cite{DoeGroHen-99,DoeStoTol-01,GolMaj-04}, \cite[Sec.\,3.7]{Lundholm-17}, \cite{MacSaw-19} 
	for further details concerning the motivations for these definitions.
	For background on fiber bundles and connections we refer to \cite{Nakahara-03,Taubes-11}.

\subsection{Classical configuration space}\label{sec:ham-cspace}

	Let us temporarily consider particles as points in $\R^d$ 
	for general $d \in \N$, before fixing $d=2$.
	The configuration space for $N$ \keyword{\emph{distinguishable} particles} 
	in $\R^d$ is simply the set $(\R^d)^N$ of $N$-tuples of coordinates (distinguished by label/ordering). 
    Let
	\begin{equation}\label{eq:fat-diagonal}
		\bDelta^N 
		:= \{ (\bx_1,\ldots,\bx_N) \in (\R^d)^N : \text{$\exists\ j \neq k$ s.t. $\bx_j = \bx_k$} \}
	\end{equation}
	denote the \keyword{fat diagonal} of this space, 
	i.e.\ the coincidence set of at least two particles.
	The configuration space for $N$ \keyword{\emph{distinct} particles} in $\R^d$
	is then the set $(\R^d)^N \setminus \bDelta^N$ of $N$-tuples of distinct coordinates.
	Consider the natural\footnote{We may regard a tuple $\sx = (\bx_j)_{j=1}^N$ 
	as a map $j \mapsto \bx_j$, so that $\sx.\sigma = \sx \circ \sigma$ indeed acts on the right.} 
	actions of permutations $\sigma \in S_N$ on $(\R^d)^N$,
	\begin{align*}
		(\sigma,\sx) \mapsto \sigma.\sx &= (\bx_{\sigma^{-1}(1)},\ldots,\bx_{\sigma^{-1}(N)}) \\
		(\sx,\sigma) \mapsto \sx.\sigma &= (\bx_{\sigma(1)},\ldots,\bx_{\sigma(N)}).
	\end{align*}
	Identifying the orbits of $S_N$ under this action on distinct particles,
	one obtains the configuration space for $N$ \keyword{\emph{identical} particles} in $\R^d$,
	\begin{equation}\label{eq:config-space}
		\cC^N 
		:= \left( (\R^d)^N \setminus \bDelta^N \right) \Big/ S_N
		= \left\{ X=\{\bx_1,\ldots,\bx_N\} \subset \R^d : |X|=N \right\} ,
	\end{equation}
	the space of $N$-point subsets of $\R^d$
    (sometimes denoted $\eP_N(\R^d)$).
	For reference, we fix a particular point in $\cC^N$,
	$$
		\sx_0 = (1\be_1,2\be_1,\ldots,N\be_1)
		\quad \mapsto \quad
		X_0 = \{1\be_1,2\be_1,\ldots,N\be_1\},
	$$
	where $\be_1$ is the first unit basis vector in $\R^d$
	(for $d=2$ this fixes a real axis in $\C \cong \R^2$).
	This choice of base point will provide us with a reference ordering of the particles,
	and corresponds naturally to the arrangement of strands in Figure~\ref{fig:braids}.

	Consider now the continous motion of identical particles in $\R^d$.
	Given a continuous path $\gamma\colon [0,1] \to \cC^N$
	from a point $X=\gamma(0)$ to a point $Y=\gamma(1)$, 
	there is an induced action on $\cC^N$:
	$$
		(\gamma,X) \mapsto \gamma.X = Y
	$$
	In the same way we obtain a (trivial) action of continuous loops $\gamma$ s.t. $\gamma.X=X$.
	Compositions $(\gamma,\eta) \mapsto \gamma\eta$
	of paths and loops are associative under these actions.
	A \keyword{continuous exchange} of identical particles 
	with initial configuration $X \in \cC^N$ is a loop $\gamma$ in $\cC^N$ based at $X$.
	In the case that we are interested in exchanges modulo homotopy equivalence,
	which is the case for free particles\footnote{The idea is that any non-topological information (holonomy) of loops would require a magnetic field. See \cite{Lundholm-23} for further discussion.}, 
	the relevant group of loops is the \keyword{fundamental group}
	$\pi_1(\cC^N) = \pi_1(\cC^N,X_0)$.

\begin{theorem}[Configuration space topology]\label{thm:config-space}
	The space $\cC^N$ is path-connected and its fundamental group is
	$$
		\pi_1(\cC^N) = \begin{cases}1, & d=1, \\ B_N, & d=2, \\ S_N, & d\ge 3.\end{cases}
	$$
\end{theorem}
	The proof amounts to Artin's correspondence between geometric and algebraic braid groups 
	\cite{Artin-25,Artin-47},
	and the simple-connectedness of the configuration space of distinct particles for $d\ge 3$.
	For details we refer to e.g. \cite[Theorem~1.8]{Birman-74}, \cite{KasTur-08}, or \cite{Weinberger-15}.
	
	The (universal) \keyword{covering space} of $\cC^N$ is a fiber bundle
	$$
		\tcC^N \to \cC^N \qquad \text{with fiber $\pi_1(\cC^N)$},
	$$
	and
	may be defined as the space of continuous paths in $\cC^N$
	based at $X_0$ and up to homotopy equivalence:
	\begin{align*}
		\gamma \sim \gamma' \quad \text{iff} \quad
		\exists F\colon [0,1]\times[0,1] \to \cC^N \ \text{cont. s.t.} \ 
		&F(0,\cdot) = \gamma, F(1,\cdot) = \gamma', \\
		&F(s,0) = \gamma(0), F(s,1) = \gamma(1) \ \forall s \in [0,1].
	\end{align*}
	Thus, with the usual composition of paths,
	the action of paths and loops $\gamma$ in $\cC^N$ lifts to 
	a faithful action on $\tcC^N$,
	$(\gamma,\tX) \mapsto \gamma.\tX$, with
	$\gamma.\tX = \gamma'.\tX$ for homotopic paths $\gamma \sim \gamma'$ in $\cC^N$
	(with the same endpoints),
	where the image of the action of loops $\gamma$ at $X \in \cC^N$ 
	is exactly the induced group $\pi_1(\cC^N,X)$ 
	of equivalence classes $[\gamma]$ of homotopic loops based at $X$.
	The canonical projection to the endpoint of a path is denoted 
	$\tilde\pr\colon \tcC^N \to \cC^N$,
    $[\gamma] \mapsto \gamma(1)$.
	By definition $\tcC^N$ is simply connected, $\pi_1(\tcC^N,\tX_0) = 1$.
	
	Hence, taking the reference points 
	$X_0 \mapsfrom \tX_0 \leftrightarrow (X_0,1)$, 
    corresponding to the trivial path from $X_0$,
	an arbitrary point $\tX \in \tcC^N$ may be expressed
	$$
		\tX = \tilde\gamma_0.\tX_0 = \gamma_0.\tX_0 = [\gamma_0].\tX_0 
	$$
	for some path $\tilde\gamma_0$ in $\tcC^N$ based at $\tX_0$, 
	or equivalently, using its projection
	$\gamma_0 = \tilde\pr \circ \tilde\gamma_0$ in $\cC^N$ based at 
	$X_0 = \tilde\pr(\tX_0)$ and ending at $X = \tilde\pr(\tX)$,
	or using simply the equivalence class $[\gamma_0]$ of homotopic such paths.
	
	Any two points $\tX,\tX' \in \tcC^N$ with the same projection 
	$X=\tilde\pr(\tX)=\tilde\pr(\tX')$ differ only by an element of 
	$\pi_1(\cC^N) = \pi_1(\cC^N,X_0)$,
	since
	$$
		\tX = [\gamma_0].\tX_0 = [\gamma_0'\gamma_0'^{-1}\gamma_0].\tX_0
		= ([\gamma_0'][\gamma]).\tX_0 = \tX'.[\gamma],
	$$
	$[\gamma] = [\gamma_0'^{-1}\gamma_0] \in \pi_1(\cC^N,X_0)$,
	where we have defined an action of equivalence classes of homotopic 
	loops based at $X_0$ on the right,
	$$
		(\tX,[\gamma]) \mapsto \tX.[\gamma] := [\gamma_0\gamma].\tX_0.
	$$
	In other words, we first do $\gamma$ at $X_0$ 
	and then the path $\gamma_0$ taking us to $X = \tilde\pr(\tX)$.
	The corresponding representation $L\colon \pi_1(\cC^N) \to \mathrm{Diff}(\tcC^N)$ 
	on the left, s.t. 
	$L_{[\gamma][\eta]} = L_{[\gamma]}L_{[\eta]}$, is
	$$
		L_{[\gamma]}(\tX) := \tX.[\gamma^{-1}] = [\gamma_0\gamma^{-1}].\tX_0.
	$$
	The group $\pi_1(\cC^N)$ then acts in this way as a group of 
	diffeomorphisms of $\tcC^N$,
	and $\cC^N = \tcC^N / \pi_1(\cC^N)$.
	By writing each point $\tX \in \tcC^N$ 
	using a simple path $\gamma_0$ from $X_0$ to $X$
	and a loop $\gamma$ at $X_0$,
	we have a way of representing 
	$\tX = ([\gamma_0].\tX_0).[\gamma]$ as a pair 
	$(X,[\gamma]) \in \cC^N \times \pi_1(\cC^N)$
	(this is not a canonical idenfication but rather a choice of
	local trivialization of the bundle).
	The corresponding set of points $(X,1) \leftrightarrow [\gamma_0].\tX_0$,
	with $\gamma_0$ a simple path from $X_0$ to $X$,
	constitutes a fundamental domain in $\tcC^N$ in bijection to $\cC^N$.
	
	Further, one may interchange the above actions via the conjugation of paths
	$$
		\Ad_{\eta}[\gamma] := [\eta\gamma\eta^{-1}],
	$$
	for compatible paths $\eta,\gamma$ in $\cC^N$.
	Namely, given $\tX = [\gamma_0].\tX_0$ and $[\gamma] \in \pi_1(\cC^N,X)$,
	$$
		[\gamma].\tX = [\gamma\gamma_0].\tX_0 
		= [\gamma_0(\gamma_0^{-1}\gamma\gamma_0)].\tX_0
		= \tX.\Ad_{\gamma_0^{-1}}[\gamma],
	$$
	with $\Ad_{\gamma_0^{-1}}[\gamma] \in \pi_1(\cC^N,X_0)$.
	Conversely, if $[\gamma] \in \pi_1(\cC^N,X_0)$ then
	$$
		\tX.[\gamma] = \Ad_{\gamma_0}[\gamma].\tX
	$$
	with $\Ad_{\gamma_0}[\gamma] \in \pi_1(\cC^N,X)$.
	
	For $d=1$, trivially $\tcC^N \cong \cC^N$ since no continuous exchange is possible
	and the particles may be ordered canonically along the real line according to the choice 
	$X_0 \leftrightarrow \sx_0 \leftrightarrow \tX_0$.
	For $d \ge 3$ we have $\tcC^N \cong \cC^N \tilde{\times} S_N$
	in the following precise sense.  The projection
	$\tilde\pr\colon\tcC^N\to\cC^N$ is the principal $S_N$-cover associated
	with the monodromy isomorphism
	$\pi_1(\cC^N)\cong S_N$: over every evenly covered open set
	$U\subset\cC^N$ it is isomorphic to $U\times S_N$, and on overlaps
	the local products are glued by right multiplication with locally constant
	elements of $S_N$ \cite[Section~1.3]{Hatcher-02}.  We use
	$\tilde{\times}$ only as a mnemonic for this generally non-canonical
	locally trivial product, not as standard notation (compare $N=2$, which
	yields a doubling of a relative half-space).
	Equivalently, each point $X$ in the configuration space has $N!$ different representatives
	corresponding to each permutation $\sigma \in S_N$, 
	and we may therefore reconstruct the space of distinct particles as
	precisely the covering space
	$$
		(\R^d)^N \setminus \bDelta^N \cong \tcC^N 
		\cong \{(X,\sigma) : X \in \cC^N, \sigma \in S_N \},
	$$
	with the reference point 
	$\tX_0 \leftrightarrow (X_0,1) \leftrightarrow \sx_0$,
	the ordered $N$-tuple of particles.

	For concreteness, we will from now on fix $d=2$ so that the relevant 
	group is the braid group $B_N$. In this case we deal with an infinite
	set of representatives for each point of $\cC^N$:
	$$
		\tcC^N \ni \tX \leftrightarrow (X,b), \ X \in \cC^N, \ b \in B_N.
	$$
	It is now convenient to consider two projections:
	$$
		\begin{array}{rcccl}
		\tcC^N & \xrightarrow{\tilde\pr} & \cC^N & \xleftarrow{\pr} & (\R^2)^N \setminus \bDelta^N \\
		(X,b) & \mapsto & X=\{\bx_1,\ldots,\bx_N\} & \mapsfrom & \sx = (\bx_1,\ldots,\bx_N)
		\end{array}
	$$
	In fact there is also a canonical projection $\hat\pr: B_N \to S_N$ 
	defined by taking the quotient with the additional relations $\{\sigma_j^2=1\}$ of $S_N$,
	and thus we may in similarity to the higher-dimensional case consider
	$$
		\begin{array}{rcccccl}
		\tcC^N & \xrightarrow{\hat\pr} & \cC^N \times S_N &\cong& (\R^2)^N \setminus \bDelta^N & \xrightarrow{\pr} & \cC^N \\
		(X,b) & \mapsto & (X,\hat\pr(b)) &\leftrightarrow& \sx = (\bx_1,\ldots,\bx_N) & \mapsto & X=\{\bx_1,\ldots,\bx_N\} \\
		(X_0,b) & \mapsto & (X_0,\hat\pr(b)) &\leftrightarrow& \hat\pr(b^{-1}).\sx_0 = \sx_0.\hat\pr(b) & \mapsto & X_0
		\end{array}
	$$
	so that the composition is $\tilde\pr = \pr \circ \hat\pr$.
	Inverses may then be defined,
	$$
		\pr^{-1}_0(X) := (X,1),
		\qquad
		\tilde\pr^{-1}_0(X) := (X,1),
	$$
	mapping $\cC^N$ to fundamental domains in $(\R^2)^N$ respectively $\tcC^N$.
	
\begin{example}\label{ex:2-anyon-cspace}
	In the case $N=2$ we may identify
	$$
		\cC^2 \cong \cC_0 \times \R^+\times\R^2,
	$$
	where the configuration space of the relative angle 
	$\cC_0 = S^1/\Z_2 = [0,\pi)_{\text{per.}}$ is a circle, 
	$\pr^{-1}_0(\cC_0) = [0,\pi)$ is a single-cover of $\cC_0$ and a half-circle,
	$\pr^{-1}(\cC_0) = [0,2\pi)_{\text{per.}}$ is a double-cover and a full circle,
	while the full covering space $\tcC_0 = \R$ is the real line
	containing the fundamental domain $\tilde\pr^{-1}_0(\cC_0) = [0,\pi)$.
	Points on $\tcC_0$ are indexed by $\cC_0$ and
	$\pi_1(\cC_0) = \Z$, the winding number of the loop relative to a base point such as $0$.
\end{example}

\subsection{Hamiltonian and Hilbert space}\label{sec:ham-space}

	One means of quantizing a classical system is to find a formulation in terms of a 
	Poisson algebra of observables and then look for representations
	as operators on a Hilbert space
	(see e.g.\ \cite{Thiemann-07,Lundholm-17} for introductions).	
	We consider quantizations of the classical non-relativistic kinetic energy,
	i.e.\ the Hamiltonian function\footnote{For simplicity 
	we put the mass equal to $1/2$ throughout, 
	as well as Planck's constant $\hbar=1$.}
	\begin{equation}\label{eq:Hamiltonian}
		T = \sum_{j=1}^N \bp_j^2,
	\end{equation}
	and thus seek representations of the momenta 
	$\ssp = (\bp_j)_{j=1}^N = (p_{jk})_{j=1,\ldots,N,k=1,2}$ 
	as differentiation operators.
	The important point here is that 
	the momenta are conjugate to the configuration variables 
	$X = \{\bx_1,\ldots,\bx_N\}$ 
	of \emph{identical} particles, which take values on the manifold $\cC^N$.
	Locally, i.e.\ on any topologically trivial open
	subset $\Omega \subseteq \cC^N$ the particles remain distinguishable,
	and may thus be given representatives $\sx = (\bx_1,\ldots,\bx_N)$ 
	on coordinate charts $\pr^{-1}_0(\Omega) \subseteq \R^{2N}$, 
	by the correspondence above.
	We may then use that,
	locally w.r.t $\sx$ in these charts, 
	the classical canonical Poisson algebra
	$$
		\{x_{jk},p_{j'k'}\} = \delta_{jj'}\delta_{kk'},
		\qquad j,j' \in \{1,\ldots,N\}, \ k,k' \in \{1,2\},
	$$
	can be given a Schr\"odinger representation as the quantum CCR algebra
	$$
		{\textstyle\frac{1}{i}}[\hx_{jk},\hp_{j'k'}] = \delta_{jj'}\delta_{kk'} \1,
		\qquad j,j' \in \{1,\ldots,N\}, \ k,k' \in \{1,2\},
	$$
	represented on a Hilbert space $\cH_\Omega = L^2(\Omega;\cF)$, 
	where $\cF$ is a representation Hilbert space for any internal degrees of freedom
	(observables other than these $x$'s and $p$'s),
	and then the different charts $\Omega$ 
	are to be patched together to a representation on all of the configuration
	space manifold $\cC^N$.
	The natural geometric framework \cite{MueDoe-93,DoeGroHen-99,DoeStoTol-01},
	\cite[Remark~3.29]{Lundholm-17} 
	for this problem is to consider a hermitian fiber
	bundle $E \to \cC^N$ over the configuration space, with fiber $\cF$,
	endowed with a connection
	$\cA$ and thus a covariant derivative $\nabla^{\cA}$.
	We are in this work, as a first step, interested in ideal, free anyon models so that
	the pure effects of statistics may be isolated. Physically, this means that
	upon restricting the particles to topologically trivial subsets 
	$\Omega \subseteq \cC^N$ 
	they should
	behave as ideal, free and distinguishable particles with the usual momenta
	$\hbp_j = -i\nabla_{\bx_j}$
	and the usual free 
	kinetic energy $\hT = \sum_{j=1}^N (-\Delta_{\bx_j})$
	and no interactions.
	Geometrically, this means that the connection $\cA$ should be (locally) flat.
	Furthermore, for simplicity we only consider finite-dimensional spaces $\cF$.
	We have the following convenient classification of such bundles:

\begin{theorem}[Classification theorem of flat bundles]\label{thm:flat-bundle}
	There is a 1-to-1 correspondence between flat hermitian vector bundles 
	$E_\rho \to \cC^N$ with fiber $\cF=\C^D$
	and unitary representations of $B_N$ on $\cF$, $\rho\colon B_N \to \sU(\cF)$,
	up to conjugation (similarity) in $\sU(\cF)$.
\end{theorem}
	
	In other words, given $N$ and fiber $\cF$,
	the \keyword{moduli space of flat connections} is exactly
	\begin{equation}\label{eq:character-variety}
		\mathrm{Hom}(B_N,\sU(\cF))/\sU(\cF)
	\end{equation}
	where the action of $\sU(\cF)$ on representations $\rho \in \mathrm{Hom}$ 
	is the adjoint one,
	$$
		(S,\rho) \mapsto \Ad_S\rho = S\rho S^{-1}\colon b \mapsto S\rho(b)S^{-1}.
	$$
	A detailed proof of this well-known classification theorem is given in \cite[Theorem~13.2]{Taubes-11},
	and a further discussion of the correspondence in \cite[Chapter~5]{Michiels-13}
	as well as \cite{MacSaw-19}.
	This correspondence motivates the following definition of an arbitrary model
	of ideal, free anyons.

\begin{definition}[Geometric anyon model]\label{def:anyon-model}
	By a \keyword{geometric $N$-anyon model} we mean here a 
	flat hermitian vector bundle $E_\rho \to \cC^N$,
	or equivalently, a unitary representation $\rho\colon B_N \to \sU(\cF)$.
	Similarly, a \keyword{geometric many-anyon model}
	is a family of geometric anyon models
	$\rho_N\colon B_N \to \sU(\cF_N)$ for some sequence of $N \ge 1$
	(possibly with restrictions on which $N$ are allowed; 
	multiples, odd/even, etc.), where $B_1$ is the trivial group.
	
	Further, the \keyword{Hilbert space} of a geometric $N$-anyon model 
	is defined as the space of
	$L^2$ global sections on the bundle $E_\rho \to \cC^N$,
	i.e. the $L^2$-closure of smooth global sections,
	$$
		\cH_\rho^N := \overline{\left\{ \Psi \in \Gamma(\cC^N,E_\rho) : \int_{\cC^N} |\Psi|_\cF^2 < \infty \right\}}^{L^2},
	$$
	where $\Gamma(M,E)$ denotes the set of smooth sections on the bundle $E \to M$.
\end{definition}

\begin{example}
	Using the map $\hat\pr\colon B_N \to S_N$ and the canonical representation 
	$\hat\rho\colon S_N \to \sU(N)$ of $\sigma_j$ acting as permutation matrices 
	($z=1$ in \eqref{eq:Burau-sigma}),
	we may construct a bundle (the \keyword{canonical permutation bundle})
	with fiber $\cF = \C^N$ and the canonical open covering
	$\{\Omega_\sigma\}_{\sigma \in S_N}$ of $\R^{2N} \setminus \bDelta^N \to \cC^N$ 
	with locally constant transition functions
	$$
		t_{\sigma\sigma'}\colon \Omega_\sigma \cap \Omega_{\sigma'} \to \sU(N),
		\quad X \mapsto \hat\rho(\sigma\sigma'^{-1}),
		\quad \sigma,\sigma' \in S_N.
	$$
	The connection is locally trivial but its holonomy nontrivial due to the transition functions.
\end{example}

	The space $\Gamma(\cC^N,E_\rho)$ of smooth global sections
	can equivalently be characterized as the space $C^\infty_\rho(\tcC^N;\cF)$ of smooth
	\keyword{$\rho$-equivariant functions} on the covering 
	space $\tcC^N$, i.e. functions 
	$\Psi_\rho\colon \tcC^N \to \cF$ such that, 
	given any two points $\tX = [\gamma_0].\tX_0$ and $\tY = \tilde\gamma.\tX$ in $\tcC^N$ 
	connected by a path $\tilde\gamma$ in $\tcC^N$
	and such that its projection $\gamma = \pr \circ \tilde\gamma$ in $\cC^N$,
	based at the point $X = \gamma(0) \in \cC^N$,
	is a closed loop, $\gamma.X = \gamma(1) = X$,
	we have
	\begin{equation}\label{eq:Psi-equivariant}
		\Psi_\rho(\tY) = \rho(b^{-1}) \Psi_\rho(\tX),
	\end{equation}
	where $b = \Ad_{\gamma_0^{-1}} [\gamma] \in B_N$ is the equivalence class at $X_0$ to which the loop
	$\gamma$ belongs, i.e. $\tY = \tX.b = L_{b^{-1}}(\tX)$.
	Thus, we have an isomorphism $\cH_\rho^N \cong L^2_\rho(\tcC^N;\cF)$ 
	(see \cite[Lemma~2.4]{MunSch-95}), 
	where the inner product of the latter space is defined
	$$
		\langle \Phi_\rho,\Psi_\rho \rangle_{L^2_\rho} 
		= \int_{\cC^N} \langle\Phi_\rho(\tX),\Psi_\rho(\tX)\rangle_\cF \,dX.
	$$
	In the integrand $\tX$ denotes any pre-image of $X \in \cC^N$ under the projection 
	$\tilde\pr\colon \tcC^N \to \cC^N$.
	This integrand is indeed well defined 
	thanks to the equivariance condition \eqref{eq:Psi-equivariant}, 
	because if $\tX' = \tX.b$, $b \in B_N$, then
	$$
		\langle\Phi_\rho(\tX'),\Psi_\rho(\tX')\rangle_\cF 
		= \langle\rho(b^{-1})\Phi_\rho(\tX),\rho(b^{-1})\Psi_\rho(\tX)\rangle_\cF 
		= \langle\Phi_\rho(\tX),\Psi_\rho(\tX)\rangle_\cF,
	$$
	by the unitarity of the representation $\rho$.
	The measure $dX$ on $\cC^N$
	is derived from the usual Lebesgue measure $d\sx$ on $\R^{2N}$
	via the projection $\pr$, such that
	$\int_{\Omega} dX = \frac{1}{N!}\int_{\pr^{-1}(\Omega)} d\sx$.
		
	The condition \eqref{eq:Psi-equivariant} 
	is in fact used in the proof of Theorem~\ref{thm:flat-bundle}
	to explicitly construct
	the flat bundle $E_\rho$ associated to $\rho$; cf. \cite[Chapter~13.9.3]{Taubes-11}.
	A $\rho$-equivariant function may be considered a relation 
	$\Psi_\rho \subseteq \tcC^N \times \cF$
	with a unique $v \in \cF$ to each $\tX \in \tcC^N$ and such that 
	if $(\tX,v) \in \Psi_\rho$ then
	$(\tX.b,\rho(b^{-1})v) \in \Psi_\rho$ for all $b \in B_N$.
	Then the bundle $E_\rho \to \cC^N$ is precisely
	the set of orbits in $\tcC^N \times \cF$ of the right action
	$$
		(\tX,v).b := (\tX.b,\rho(b^{-1})v),
		\qquad \tX \in \tcC^N, v \in \cF, b \in B_N,
	$$
	with projection $[(\tX,v)] \mapsto \tilde\pr(\tX)$.
	Conversely, given a smooth global section $\Psi \in \Gamma(\cC^N,E_\rho)$,
	there is to each $\tX$ a unique $v \in \cF$ such that
	$\Psi(\tilde\pr(\tX)) = [(\tX,v)]$, which defines a $\rho$-equivariant
	function $\Psi_\rho(\tX) := v$
	(again, see \cite[Lemma~2.4]{MunSch-95}).
	
	Given an $N$-anyon model $\rho$ and locally flat bundle $E_\rho \to \cC^N$ we now
	have a means of defining a covariant derivative $\nabla^\rho\Psi$ 
	on smooth sections $\Psi \in \Gamma(\cC^N,E_\rho)$
	which on any topologically trivial subset $\Omega \subseteq \cC^N$
	reduces to the ordinary derivative on $\Psi\colon \Omega \to \cF$
	$$
		\nabla\Psi(X) = (\nabla_j\Psi)_{j=1,\ldots,N} = (\partial\Psi/\partial x_{jp})_{j=1,\ldots,N}^{p=1,2}
	$$
	and thus a kinetic energy operator 
	$\hT_\rho = (-i\nabla^\rho)^*(-i\nabla^\rho)$
	which on $\Omega$ reduces to the ordinary one for distinguishable particles
	$$
		\hT_\rho\Psi(X) = \sum_{j=1}^N (-i\nabla_j)^2\Psi = -\sum_{j=1}^N \sum_{p=1,2} \partial^2 \Psi/\partial x_{jp}^2.
	$$
	The corresponding quadratic form
	\begin{equation}\label{eq:def-kinetic-bundle}
		T_\rho[\Psi] := \int_{\cC^N} |\nabla^\rho\Psi|_{\cF^{2N}}^2 \,dX
	\end{equation}
	may be extended to the $L^2$-sections for which the integral is finite,
	by taking the closure of smooth sections with compact support on $\cC^N$
	(see below).
	
	In the $\rho$-equivariant setting, a function $\Psi_\rho \in C^\infty_\rho(\tcC^N;\cF)$
	has a derivative at any point $\tX \in \tcC^N$, defined by
	\begin{equation}\label{eq:equivariant-derivative}
		\nabla\Psi_\rho = \left( \frac{\partial\Psi_\rho}{\partial x_{j1}}, \frac{\partial\Psi_\rho}{\partial x_{j2}} \right)_{j=1,\ldots,N},
		\quad
		\frac{\partial\Psi_\rho}{\partial x_{jp}} 
		= \lim_{t \to 0} \frac{1}{t}(\Psi_\rho(\gamma_{jp}^t.\tX)-\Psi_\rho(\tX))
	\end{equation}
	for $j\in\{1,\ldots,N\}$ and $p\in\{1,2\}$,
	where $\gamma_{jp}^t$ is the path 
	$[0,1] \ni s \mapsto \{\bx_1,\ldots,\bx_j + st\be_p,\ldots,\bx_N\}$ in $\cC^N$
	defining the corresponding tangent vector.
	The index $j$ as it stands here is artificial and depends on which 
	ordering of particles we choose, and we may use the projection 
	$\hat\pr\colon \tX \mapsto \sx$ before applying the path 
	on the $j$th entry of the tuple $\sx$ 
	and then project again to $\cC^N$ with $\pr$.
	Hence the vector in \eqref{eq:equivariant-derivative} depends in fact on the pre-image
	of $X$ in $B_N$.
	In any case, this ambiguity is removed in the permutation-invariant sum
	\begin{equation}\label{eq:kinetic-sum}
		|\nabla\Psi_\rho|^2_{\cF^{2N}} 
		= \sum_{j=1}^N \sum_{p=1,2} \left| \frac{\partial\Psi_\rho}{\partial x_{jp}} \right|_\cF^2.
	\end{equation}
	The kinetic energy
	\begin{equation}\label{eq:def-kinetic-equivariant}
		T_\rho[\Psi_\rho] := \int_{\cC^N} |\nabla\Psi_\rho(\tX)|_{\cF^{2N}}^2 \,dX
	\end{equation}
	is then again well defined because of equivariance,
	\begin{equation}\label{eq:diff-equiv}
		\frac{\partial\Psi_\rho}{\partial x_{jp}}(\tX.b) 
		= \lim_{t \to 0} \frac{1}{t}(\Psi(\gamma_{jp}^t.(\tX.b))-\Psi(\tX.b))
		= \rho(b^{-1})\frac{\partial\Psi_\rho}{\partial x_{j'p}}(\tX),
		\quad j' = \hat\pr(b)[j],
	\end{equation}
	\begin{multline*}
		\langle\nabla\Phi(\tX'),\nabla\Psi(\tX')\rangle_{\cF^{2N}} 
		= \sum_{j=1}^N \sum_{p=1,2} \left\langle \rho(b^{-1}) \frac{\partial\Phi_\rho}{\partial x_{j'p}}(\tX) , \rho(b^{-1}) \frac{\partial\Psi_\rho}{\partial x_{j'p}}(\tX) \right\rangle_\cF  \\
		= \langle\rho(b^{-1})\nabla\Phi(\tX),\rho(b^{-1})\nabla\Psi(\tX)\rangle_{\cF^{2N}} 
		= \langle\nabla\Phi(\tX),\nabla\Psi(\tX)\rangle_{\cF^{2N}},
	\end{multline*}
	and is identical to \eqref{eq:def-kinetic-bundle} by the above correspondence.
	We consider the closure of this form \eqref{eq:def-kinetic-equivariant}
	on $\Psi_\rho \in C^\infty_{\rho,c}$, 
	the $\rho$-equivariant smooth functions on $\tcC^N$
	such that $\tilde\pr(\supp\Psi_\rho) \subseteq \cC^N$ is compact,
	which then defines a Sobolev subspace 
	$H^1_\rho \subseteq L^2_\rho$
	of anyon wave functions with finite expected kinetic energy.
	This is also a Hilbert space with the inner product
	$$
		\langle\Phi,\Psi\rangle_{H^1_\rho} 
		:= \int_{\cC^N} \left( \langle\Phi(\tX),\Psi(\tX)\rangle_\cF + \langle\nabla\Phi(\tX),\nabla\Psi(\tX)\rangle_{\cF^{2N}} \right)\,dX.
	$$
	
\begin{definition}[Anyon Hamiltonian]\label{def:anyon-Hamiltonian}
	We define the kinetic energy operator $\hT_\rho$ on the full 
	$N$-anyon Hilbert space
	$\cH_\rho^N$ as the unique self-adjoint operator corresponding 
	to the closed non-negative quadratic form $T_\rho$ in \eqref{eq:def-kinetic-equivariant},
	with operator domain $\cD(\hT_\rho) \subseteq H^1_\rho \subseteq \cH^N_\rho$
	(see e.g. \cite[Theorem~2.14]{Teschl-14}).
	Considering $\hT_\rho$ as initially defined 
    on the minimal domain
    $C^\infty_{\rho,c}$, 
	with supports away from diagonals $\bDelta^N$,
	this is the \keyword{Friedrichs extension}.
\end{definition}
\begin{remark}\label{rmk:extensions}
	One may consider other extensions, 
	such as the \keyword{Krein extension}\footnote{This
	is the smallest (in terms of energy) non-negative extension while Friedrichs is the largest;
    see \cite{AloSim-80}.
	The extension theory for the 2-anyon problem is actually somewhat analogous to that of two 
	bosons in 3D, admitting a circle of extensions, 
	with the functions of the Friedrichs extension vanishing as $r^\alpha$
	and those of Krein behaving like $r^{-\alpha}$ as the relative distance $r \to 0$.
	Indeed, just as in 3D, we anticipate that there could be 
	good physical reasons to consider
	these other extensions whenever anyons arise as emergent particles 
	along with other interactions. 
	Non-Friedrichs extensions are sometimes referred to as soft-core anyons.},
	however comparing to the two-anyon case where the extension theory has been
	carried through in detail 
	\cite{Grundberg-etal-91,ManTar-91,BorSor-92,AmeBak-95,CorOdd-18,CorFer-21,BorCorFer-24}
	these other extensions would all correspond to introducing additional 
	attractive interactions between the anyons.
	Thus taking the Friedrichs extension is in alignment with our wish in this work
	to isolate the effects of statistics and study purely ideal anyons.
	See \cite{DelFigTet-94,AtaGirLun-25} for possible approaches to the other extensions in the many-anyon case,
	and \cite{LunSol-14} for some further mathematical details 
	in the abelian Friedrichs case.
	Note that the issue here concerns \keyword{point interactions} supported at $\bDelta^N$ while we
	could also have introduced ordinary scalar interactions by adding 
	pair (or higher order) potentials to the Hamiltonian \eqref{eq:Hamiltonian}.
	Both these types of modifications would alter the behavior for 
	distinguishable particles as well.
\end{remark}
	
	In the many-anyon context one may consider the group
	$B_\infty$ defined as the direct limit of the $N$-strand groups with respect to
	the inclusion maps $B_N \hookrightarrow B_{N+1}$ sending
	$\sigma_j \mapsto \sigma_j$.
	It is also natural to consider sequences of representations
	$\rho_N\colon B_N \to \sU(\cF_N)$ such that the respective inclusions 
	$B_N \hookrightarrow B_{N+1}$ and 
	$\C\rho_N(B_N) \hookrightarrow \C\rho_{N+1}(B_{N+1})$ 
	(the group algebras) commute \cite[Definition~2.1]{RowWan-12},
	although not all representations are of this type, such as the Burau representations
	\cite[Sec.~6.1]{DelRowWan-16}.
	
	For geometric many-anyon models we may also consider the \keyword{Fock space}
	$$
		\cH^\infty := \bigoplus_{N=0}^\infty \cH_{\rho_N}^N,
	$$
	where $\cH_{\rho_0}^0 := \cF_0 = \C$ is the vacuum sector.
	However, the relationship between $N$-particle sectors is not nearly as 
	straightforward for anyons as for bosons and fermions, 
	which admit standard creation and annihiliation operators 
	with (anti)commutation relations.
	For an irreducible abelian model every $\cF_N$ is one-dimensional and
	$\rho_N$ is a character of $B_N$, so this direct sum is the 
	collection of abelian $N$-anyon sectors as defined above 
	with statistics parameters $\alpha_N \in [0,2)$.
	By inclusion $B_N \hookrightarrow B_{N+1}$ we may require 
	$\alpha_N = \alpha$ independent of $N$.
	A reducible abelian model is,
	for each $N$, an orthogonal direct sum of such characters.  
	At the bosonic and fermionic characters $\alpha \in \{0,1\}$ 
	the monodromy of particle exchange is trivial and
	the construction
	reduces to the standard symmetric and antisymmetric Fock spaces,
	respectively. 
	For a general abelian phase the tensor products and
	creation operators are braided rather than symmetric or antisymmetric,
	which makes the construction highly nontrivial and requires creation to
	be specified more precisely by e.g. bringing new particles in from
	infinity along certain trajectories, or otherwise ordering the particles
	consistently.
	An explicit construction for abelian anyons and comparison with ordinary 
	Bose and Fermi Fock spaces is given in \cite{GolMaj-04}.
	In the non-abelian case the
	new feature is that the fusion fibers themselves grow with $N$ and the
	representations $\rho_N$ must be chosen compatibly across sectors.
	Recall that in order to have a non-abelian model it is necessary that 
	$\dim \cF_N \ge N-2$ 
	(Theorem~\ref{thm:small-reps-abelian}), and we cannot therefore 
	simply fix a finite-dimensional $\cF$ as we take $N \to \infty$,
	as one may do for bosons and fermions.

\subsection{Statistics transmutation}\label{sec:ham-transmutation}
		
	Given $N$ and fiber $\cF$,
	let the \keyword{trivial} or \keyword{bosonic bundle} be the bundle with trivial geometry (also defined for any configuration space), 
	taking the trivial representation
	$\rho_+(\sigma_j) := +\1$ for all $j$, thus $E_{\rho_+} = \cC^N \times \cF$.
	We note that any function $\Psi_+ \in \cH_+^N$
	of the standard \keyword{bosonic $N$-body Hilbert space}
	\begin{equation}\label{eq:bosonic-Hilbert-space}
		\cH_+^N := L^2_\sym((\R^2)^N;\cF) :=
		\left\{ \Psi \in L^2(\R^{2N};\cF) : \Psi(\sx.\sigma) = \Psi(\sx) \,\forall \sigma \in S_N \right\},
	\end{equation}
	may be considered an equivariant function $\Psi_{\rho_+} \in \cH^N_{\rho_+}$
	of the bosonic bundle. 
	Namely, by the canonical map $\hat\pr\colon B_N \to S_N$,
	we have in $\tcC^N$ pre-images of $\cC^N$ labeled by 
	permutations $\hat\pr(b)$.
	A symmetric function $\Psi_+ \in \cH_+^N$ thus extends to a function
	$\Psi_{\rho_+}(\tX.b) := \Psi_+(\sx.\hat\pr(b)) = \Psi_+(\sx)$, 
	$\sx = \hat\pr(\tX)$.
	Conversely, any function $\Psi_{\rho_+}\colon \tcC^N \to \cF$ with symmetry
	\begin{equation}\label{eq:Psi-invariant}
		\Psi_{\rho_+}(\tX) = \Psi_{\rho_+}(\tX.b) = \Psi_{\rho_+}(X),
		\quad b \in B_N,
	\end{equation}
	where $\tX$ is any pre-image in $\tcC^N$ of a point $X \in \cC^N$,
	may be identified with a function in $\cH_+^N$
	by defining $\Psi_+(\sx) := \Psi_{\rho_+}\bigl(\pr(\sx)\bigr)$ on 
	$\R^{2N} \setminus \bDelta^N$
	and using that
	$L^2_\sym(\R^{2N} \setminus \bDelta^N) = L^2_\sym(\R^{2N})$
	by the vanishing measure of $\bDelta^N$ in $\R^{2N}$.
	Note however that to make it a probability distribution in $L^2(\R^{2N})$ 
	(i.e. on the \emph{distinguishable} configuration space as is conventional)
	it is necessary to normalize 
	$\Psi_+ := \Psi_{\rho_+}/\sqrt{N!}$.
	We note for example that, given any $\rho$-equivariant $\Psi_\rho$, 
	the scalar function $|\Psi_\rho|_\cF$ (amplitude) is $\rho_+$-equivariant,
	$$
		|\Psi_\rho(\tX.b)|_\cF = |\rho(b^{-1})\Psi_\rho(\tX)|_\cF = |\Psi_\rho(\tX)|_\cF,
	$$
	which is to say that $|\Psi_\rho|_\cF^2$ defines a probability distribution on $\cC^N$
	for $L^2_\rho$-normalized $\Psi_\rho$.
	Furthermore, we have that $H^1_0(\R^{2N} \setminus \bDelta^N) = H^1(\R^{2N})$ 
	(see e.g. \cite[Lemma~3]{LunSol-14} and, more generally, \cite[Appendix~B]{LarLunNam-19})
	so that we may as well identify the \keyword{bosonic Sobolev spaces}
	$$
		H^1_{\rho_+} \cong H^1_\sym := \overline{C^\infty_c(\R^{2N}\setminus\bDelta^N) \cap L^2_\sym}^{H^1(\R^{2N})}.
	$$
	Thus, after taking the closure, the functions do not have to vanish at 
	$\bDelta^N$, and indeed for example the Gaussian
	$\Psi(\sx) = e^{-|\sx|^2} \in H^1_\sym$.
	
	There is also an analogous version for fermions, where the
	\keyword{fermionic bundle} is defined by
	$\rho_-(\sigma_j) := -\1$ for all $j$, i.e. 
	$\rho_- = \sign \circ \hat\pr \otimes \1$.
	Instead of $\cH_+^N$ we then consider the \keyword{fermionic $N$-body Hilbert space}
	\begin{equation}\label{eq:fermionic-Hilbert-space}
		\cH_-^N := L^2_\asym((\R^2)^N;\cF) :=
		\left\{ \Psi \in L^2(\R^{2N};\cF) : \Psi(\sx.\sigma) = (\sign\sigma)\Psi(\sx) \,\forall \sigma \in S_N \right\},
	\end{equation}
	Taking $\Psi_{\rho_-}(\tX.b) := \Psi_-(\sx.\hat\pr(b)) = \sign \hat\pr(b^{-1}) \Psi_-(\sx)$,
	we obtain a section of the fermionic bundle,
	and vice versa $\Psi_-(\sx) := \Psi_{\rho_-}(\tX)$ (modulo normalizations)
	with any lift $\tX \mapsto \sx \mapsto X$
	defines an antisymmetric function.
	We define analogously the \keyword{fermionic Sobolev space} 
	$$
		H^1_{\rho_-} \cong H^1_\asym := \overline{C^\infty_c(\R^{2N}\setminus\bDelta^N) \cap L^2_\asym}^{H^1(\R^{2N})},
	$$
	and in this case we will see 
	(by means of the Hardy inequality in Section~\ref{sec:repulsion}) that the functions
	with finite kinetic energy necessarily vanish on $\bDelta^N$,
	as had been expected from \eqref{eq:Pauli}.
		
	Note that to any flat hermitian vector bundle $E \to \cC^N$ with fiber $\cF$, 
	or representation $\rho\colon B_N \to \sU(\cF)$, 
	there is also the associated flat principal bundle $P \to \cC^N$ with fiber $\sU(\cF)$
	(and vice versa).
	Again the correspondence to $\rho$-equivariant functions carries over to $P$, 
	which may be defined as the
	set of orbits in $\tcC^N \times \sU(\cF)$ of the right action
	$$
		(\tX,g).b := (\tX.b,\rho(b^{-1})g),
		\qquad \tX \in \tcC^N, g \in \sU(\cF), b \in b_N,
	$$
	with the projection $[(\tX,g)] \mapsto \tilde\pr \tX = X$.
	Using a similar identification as for $\Gamma(\cC^N,E)$ then,
	the set of smooth global sections 
	$\Gamma(\cC^N,P)$ is the set of smooth functions
	$u_\rho\colon \tcC^N \to \sU(\cF)$ that are $\rho$-equivariant,
	\begin{equation}\label{eq:u-equivariant}
		u_\rho(\tX.b) = \rho(b^{-1})u_\rho(\tX).
	\end{equation}
	
	Recall that a hermitian vector bundle of rank $n$ 
	is \keyword{trivial} iff there exist
	$n$ global sections which are pointwise orthonormal, thus defining a frame
	$\{u_1(X),\ldots,u_n(X)\}$ in $\cF$ at each $X$,
	and, equivalently, iff there exists a global section of the principal bundle,
	so that $P \cong \cC^N \times \sU(\cF)$.
	Thus, triviality of $E$ and $P$ is equivalent to the existence of $u_\rho$ 
	satisfying \eqref{eq:u-equivariant}.
	
	Now, given an anyon model $\rho\colon B_N \to \sU(\cF)$,
	assume that there exists a smooth function $u_\rho\colon \tcC^N \to \sU(\cF)$ 
	satisfying \eqref{eq:u-equivariant},
	i.e.\ a global section of $P$ ($\Leftrightarrow P$ topologically trivial). 
    Then, given any symmetric function 
	$\Psi_+ \in \cH_+^N$
	one obtains a $\rho$-equivariant function $\Psi_\rho \in \cH_\rho^N$ by taking
	\begin{equation}\label{eq:Psi-transmutation}
		\Psi_\rho(\tX) := u_\rho(\tX)\Psi_+(\tX).
	\end{equation}
	Conversely, given a $\rho$-equivariant function $\Psi_\rho \in \cH_\rho^N$,
	the function 
	\begin{equation}\label{eq:Psi-transmutation-inverse}
		\Psi_+(\tX) := u_\rho(\tX)^{-1}\Psi_\rho(\tX)
	\end{equation}
	satisfies
	$\Psi_+(\tX.b) = \Psi_+(\tX)$,
	where $\tX$ is any pre-image in $\tcC^N$ of $X \in \cC^N$,
	and may therefore be identified with a function in $\cH_+^N$.
	
	Considering the kinetic energy, we may write using \eqref{eq:Psi-transmutation}
	\begin{equation}\label{eq:kinetic-magnetic}
		T_\rho[\Psi_\rho] 
		= \int_{\cC^N} |\nabla(u_\rho\Psi_+)(\tX)|_{\cF^{2N}}^2 \,dX
		= \int_{\cC^N} |\nabla^\cA\Psi_+(\tX)|_{\cF^{2N}}^2 \,dX
		= \int_{\cC^N} |\nabla^\cA\Psi_+(X)|_{\cF^{2N}}^2 \,dX,
	\end{equation}
	where the covariant derivative
	$$
		\nabla^\cA = \nabla + \cA,
	$$
	is given in terms of the connection $\cA\colon \tcC^N \to \gu(\cF)^{2N}$
	($\gu(\cF)$ is the Lie algebra of $\sU(\cF)$), 
	$$
		\cA(\tX) := u_\rho(\tX)^{-1}\nabla u_\rho(\tX).
	$$
	This expression descends to $\hat\pr(\tcC^N) = \R^{2N}\setminus\bDelta^N$ 
	thanks to the symmetry
	$$
		\cA_j(\tX.b) = u_\rho(\tX)^{-1}\rho(b)\nabla_{j'}[\rho(b^{-1})u_\rho](\tX) = \cA_{j'}(\tX)
	$$
	and thus the integrand to $\cC^N$, thanks to the permutation invariance 
	of the sum, as in \eqref{eq:kinetic-sum}.
	
	We call the above operation \keyword{statistics transmutation} as it takes
	us from one model of anyons to another, to the cost of introducing the non-trivial
	connection or \keyword{gauge potential} $\cA$.
	In fact in the abelian case $\cF=\C$, $\gu(\cF) = i\R$, 
	one has $\cA_j = i\bA_j$
	for magnetic vector potentials $\bA_j\colon \R^{2N}\setminus\bDelta^N \to \R$ 
	with magnetic field
	$$
		\curl \bA_j = \nabla u_\rho^\dagger \wedge \nabla u_\rho + u_\rho^\dagger \nabla \wedge \nabla u_\rho = \0,
	$$
	as follows due to differentiating $u_\rho^\dagger u_\rho = 1$.
	We have thus represented the geometric anyon model $\rho$ in this way as a 
	\keyword{magnetic anyon model} using bosons and magnetic potentials. 
	Also in the non-abelian case one could talk about non-abelian ``magnetic'' 
	gauge fields with zero curvature,
	$$
		D\omega_\cA = d\omega_\cA + \omega_\cA \wedge \omega_\cA = 0, 
		\qquad \omega_\cA = \sum_{j,k} \cA_{jk} dx_{jk} = u_\rho^{\dagger}du_\rho.
	$$
	One may do the corresponding transformations 
	\eqref{eq:Psi-transmutation}-\eqref{eq:Psi-transmutation-inverse}-\eqref{eq:kinetic-magnetic}
	also with a fermionic function as reference,
	and indeed there are good physical reasons for doing so.
	In systems
	whose microscopic particles are electrons, a fermionic reference
	preserves their underlying antisymmetry while the statistical gauge
	factor accounts for flux attachment; this is the natural starting point
	of composite-fermion descriptions of fractional quantum Hall states;
	see, e.g., \cite{Forte-92,Jain-07} and the more recent derivations
	\cite{LunRou-16,LamLunRou-22}.
	For simplicity we will stick to the bosonic case in all that follows
	as it will be seen not to be a loss in generality.

\begin{definition}[Transmutable anyon model]
	A geometric $N$-anyon model, 
	determined by a braid group representation
	$\rho\colon B_N \to \sU(\cF)$, will be called \keyword{transmutable} 
	if its associated flat principal bundle 
	$P \to \cC^N$ is trivial, i.e. if it admits a smooth 
	global section in $\Gamma(\cC^N,P)$
	or a $\rho$-equivariant function $u_\rho\colon \tcC^N \to \sU(\cF)$.
	The corresponding transmuted model with bosons and the gauge potential 
	$\cA = u_\rho^\dagger\nabla u_\rho$ in the kinetic energy \eqref{eq:kinetic-magnetic}
	will be called a \keyword{magnetic anyon model} corresponding to $\rho$.
	The geometric model formulation is also referred to as the \keyword{anyon gauge} 
	and the magnetic model as the \keyword{magnetic gauge}.
\end{definition}

	Bosons are then obviously transmutable with the global section $u_\rho = \1$.
	Naturally this leads to the following question which we have not seen 
	duly discussed in the literature (among the exceptions we find 
	\cite{Bloore-80,Dowker-85,SudImbImb-88,HorMorSud-89,FroMar-89,ImbImbSud-90,DoeGroHen-99,MunSch-95}, 
	as well as very recently \cite{MacSaw-19} 
	which take an alternative approach via graph configuration spaces):
\begin{question}
	Which (geometric) anyon models are transmutable? 
\end{question}

	Only in the abelian case do we find a complete answer:

\begin{theorem}
	All abelian anyon models on $\R^2$ are transmutable.
\end{theorem}
	A proof was given by \cite{Dowker-85} and \cite[Theorem~2.22]{MunSch-95} 
	(see also e.g. \cite[Example~5]{MacSaw-19})
	and 
	involves an observation that torsion in homology $H_1(B_N,\Z)$ 
	is the same as torsion in cohomology $H^2(B_N,\Z)$,
	with an explicit global section given below in Section~\ref{sec:ham-abelian}, 
	which was also implicitly used in the earlier works \cite{Wu-84b,FroMar-91}.
	
	Since a fermionic model is abelian it is thus transmutable into a bosonic model.
	Furthermore, two-anyon models are trivially transmutable since they are also abelian
	(though possibly of higher rank and reducible).
	Also recall that for rank $D=\dim \cF < N-2$ and $N > 6$, 
	any representation of $B_N$ is necessarily abelian 
	(Theorem~\ref{thm:small-reps-abelian}) and therefore the bundle trivializes 
	to a sum of line bundles.
	Geometrically, the moduli space \eqref{eq:character-variety} 
	of possible anyon models or 2D quantum statistics 
	splits into components of non-transmutable models 
	(i.e.\ isomorphism classes of topologically non-trivial bundles) 
	and a component of mutually transmutable models 
	(which in turn may a priori consist of several connected components of different 
	flat connections on a topologically trivial bundle; see however \cite[Section~3.3]{MacSaw-19})
	including the circle $\rho^\alpha \otimes \1_\cF$, $\alpha \in [0,2)$, 
	of abelian models and 
	the identity point $\rho^{\alpha=0}=\rho_+$ of the bosonic/trivial bundle.
	
	The classification of isomorphism classes of higher-rank flat hermitian
	vector bundles on $\cC^N$ for $N > 2$
	is to our knowledge an open problem.
	However, as pointed out in \cite[p.10]{MunSch-95},
	one possible way to obtain transmutable models given non-transmutable ones
	is to simply take the Whitney sum of the bundles:

\begin{definition}[Whitney sum]
	Given two geometric $N$-anyon models, determined by representations
	$\rho_1\colon B_N \to \sU(\cF_1)$ resp.\ $\rho_2\colon B_N \to \sU(\cF_2)$,
	their \keyword{Whitney sum} is the geometric $N$-anyon model determined by
	$\rho = \rho_1 \oplus \rho_2\colon B_N \to \sU(\cF_1 \oplus \cF_2)$.
	This is the bundle $E_\rho = E_{\rho_1} \oplus E_{\rho_2}$.
\end{definition}

	Namely, as noted in \cite[p.10]{MunSch-95}, 
	the bundles trivialize under the taking of
	their $k$-fold Whitney sum $\oplus^k \rho$ for large enough $k$,
	because of the following finiteness theorem for the cohomology ring
	$H^*(\cC^N,\Z) = H^*(B_N,\Z)$
	\cite{Arnold-70,Arnold-14,Vainshtein-78}:

\begin{theorem}[{Arnold's Cohomology Theorems \cite[p.201]{Arnold-14}}]
	
	\item{Finiteness:}
	With the exception of $H^0 \cong H^1 \cong \Z$,
	the cohomology groups $H^j(B_N,\Z)$ are all finite.
	
    \item{Vanishing:}
    $H^j(B_N,\Z) = 0$ for $j=2$ and for $j \ge N$.
	
	\item{Recurrence:} 
	$H^j(B_{2N+1},\Z) = H^j(B_{2N},\Z)$ for all $j$ and $N$.
	
	\item{Stability:}
	$H^j(B_N,\Z) = H^j(B_{2j-2},\Z)$ for all $N \ge 2j-2$.
\end{theorem}

\begin{theorem}
	Given any flat vector bundle $E \to \cC^N$, 
	there exists a $k \in \N$ such that 
	the $k$-fold Whitney sum $\bigoplus^k E$ is trivial.
\end{theorem}
\begin{proof}
	Hermitian vector bundles are classified up to stable equivalence by their Chern classes,
	$c_j \in H^{2j}$, 
	which trivialize in $\oplus^k E$ for large enough $k$ if they are only torsion.
	Furthermore, a hermitian vector bundle of sufficiently high rank is trivial 
	iff it is stably trivial \cite{AntEbe-12}.
\end{proof}

\begin{remark}
	As a relevant comparison, note that with the analogous definitions, 
	a pair of \emph{spinless} fermions in dimension $d=3$ is \emph{not} directly transmutable to bosons. 
	Namely, in relative coordinates
	$$
		\cC^2(\R^3) \cong \cC_0 \times \R^+ \times \R^3,
	$$
	where the relative angles are on $\cC_0 \cong S^2/{\sim}$, the sphere with
	antipodal points identified. Here $\tcC_0 \cong S^2$ and hence a global
	section $u_\rho\colon \tcC_0 \to \sU(1)$ is a smooth function on the
	sphere with the equivariance relation
	$$
		u_\rho(-\bx) = \rho(\tau)u_\rho(\bx), \qquad \bx \in S^2,
	$$
	where $\tau$ is the generator of $\pi_1(\cC_0)=\Z_2$ and $\rho$ the representation.
	For bosons ($\rho_+=1$) we can take $u_{\rho_+} = 1$ while for fermions 
	($\rho_-(\tau)=-1$)
	there is no such function by the Borsuk-Ulam theorem
	\cite[Chapter~2]{Matousek-03}
	(any continuous and odd complex-valued function on the sphere has a zero somewhere).
	However, taking the Whitney sum, $\rho_2 = \rho_- \oplus \rho_-$, 
	there is e.g. the global section
	$$
		u_{\rho_2}\colon S^2 \to \sU(2),
		\quad
		u_{\rho_2}(x,y,z) = 
		  x \begin{bmatrix} 0 & 1 \\ 1 & 0 \end{bmatrix}
		+ y \begin{bmatrix} 0 &-i \\ i & 0 \end{bmatrix}
		+ z \begin{bmatrix} 1 & 0 \\ 0 &-1 \end{bmatrix},
		\quad
		x^2+y^2+z^2=1,
	$$
	that trivializes the corresponding bundle.%
	\footnote{Thus from this perspective, one could say that spin emerges 
	as a trivialization.
	Also, in $d=3$ there is 
	the notion of \keyword{dyons} \cite{Schwinger-69} which uses
	magnetic monopoles to transmute statistics, however this goes beyond the
	scope of the present article.}
\end{remark}

\subsection{Abelian anyons}\label{sec:ham-abelian}

	The fiber for an irreducible abelian model $\rho^\alpha(\sigma_j^{-1}):=e^{i\alpha\pi}$, 
	$\alpha \in (-1,1]$, is
	$\cF = V^{\balpha^N}_c = \C$, where the total charge $c = \balpha^{\times N}$
	corresponds to the total normalized magnetic flux of all particles (mod 2).
	The convention we use here on the sign of $\alpha$ (conjugation symmetry) 
	will help to recover the
	familiar conventions with the equivariance condition \eqref{eq:Psi-equivariant}
	actually being an action of $B_N$ on the right.
	
	A global section is given by the analytic continuation to $\tcC^N$ of
	$u_{\rho^\alpha}(\sx) := U(\sx)^{\alpha}$, where we use 
	$z_j := x_{j1} + ix_{j2} \in \C$
	and the \keyword{boson-fermion transmutation}
	$$
		U\colon L^2_{\sym/\asym} \to L^2_{\asym/\sym}, 
		\qquad (U\Psi)(\sx) := \prod_{1 \le j<k \le N} \frac{z_j-z_k}{|z_j-z_k|} \Psi(\sx)
		= \exp\bigl({\textstyle\sum_{j<k}} i\phi_{jk}\bigr) \Psi(\sx),
	$$
	where $\phi_{jk}$ is the phase of $z_j - z_k$ which is initially 
	defined for $(-\pi,\pi]$ and extends to a multivalued function.
	In fact, the combination
	$$
		\frac{1}{\pi} \sum_{j<k}\phi_{jk} 
		= \frac{1}{2\pi} \sum_{j \neq k}\phi_{jk} 
		:= \frac{1}{2\pi} \sum_{j \neq k} d\phi_{jk}[\tX] 
		= \frac{1}{2\pi} \int_\gamma \sum_{j \neq k} d\phi_{jk},
		\qquad \tX = [\gamma].\tX_0,
	$$
	is well defined on $\tcC^N$ and maps braids to winding numbers $B_N \to \Z$,
	hence abelianizes the braid group (see \cite{Dowker-85,MunSch-95}).
	We then obtain $U(\tX.{\sigma_j}) = -U(\tX)$ and 
	$u_{\rho^\alpha}(\tX.\sigma_j) = \rho^\alpha(\sigma_j^{-1})u_{\rho^\alpha}(\tX)$.
	Thus by the above recipe,
	the abelian model may be converted to a magnetic model on $\cH_+$ 
	with gauge potential
	$\cA = u_\rho^{-1}\nabla u_\rho = i\bA_\alpha$:
	$$
		\bA_{\alpha,j}(\sx) 
		= \alpha \sum_{k \neq j} \nabla_j \phi_{jk}
		= \alpha \sum_{k \neq j} \frac{(\bx_j-\bx_k)^\perp}{|\bx_j-\bx_k|^2},
		\qquad
		\bx^\perp = (x,y)^\perp := (-y,x).
	$$
	Using that $\nabla^\perp \cdot \bx^\perp/|\bx|^2 = \Delta \ln|\bx| = 2\pi \delta(\bx)$,
	its magnetic field w.r.t. $\bx_j \in \R^2$ is
	$$
		\curl \bA_{\alpha,j}(\bx_j) = 2\pi\alpha\sum_{k \neq j} \delta(\bx_j - \bx_k).
	$$
	The corresponding kinetic energy for bosonic
	$\Psi_+ = u_{\rho^\alpha}^{-1}\Psi_{\rho^\alpha}$ is
	$$
		\hT_\alpha := \sum_{j=1}^N (-i\nabla_j + \bA_{\alpha,j})^2,
		\qquad
		T_\alpha[\Psi_+] = \frac{1}{N!} \int_{\R^{2N}} \sum_{j=1}^N |(-i\nabla_j + \bA_{\alpha,j})\Psi_+|^2.
	$$
	In this formulation there is the possibility to extend the
	statistics parameter $\alpha$ to all of $\R$
	and consider the gauge-equivalent magnetic models
	\begin{equation}\label{eq:gauge-equivalence}
		-i\nabla+\bA_{\alpha+2n} = U^{-2n} (-i\nabla+\bA_\alpha) U^{2n}, 
		\quad
		\cD(\hT_{\alpha+2n}) = U^{-2n} \cD(\hT_\alpha), 
		\qquad n \in \Z.
	\end{equation}
	For $\alpha=0$ the domain of the Friedrichs extension is 
	$\cD(\hT_{\rho_+}) = H^2_\sym = L^2_\sym \cap H^2(\R^{2N})$ which defines free bosons,
	while for $\alpha=1$ it is 
	$\cD(\hT_{\rho_-}) = U^{-1} H^2_\asym = U^{-1} L^2_\asym \cap H^2(\R^{2N})$ 
	which defines free fermions in a bosonic representation
	(see \cite[Section~2.2]{LunSol-14}).
	
	In the reducible case $\cF = \C^D$, we take a basis of joint eigenvectors s.t.
	$$
		\rho(\sigma_j^{-1}) = S \diag(e^{i\gamma_1\pi}, \ldots, e^{i\gamma_D\pi}) S^{-1},
	$$
	$S \in \sU(D)$, 
	and the corresponding global section is the analytic continuation to $\tcC^N$ of
	$$
		u_\rho(\sx) = S \diag\left( U(\sx)^{\gamma_1}, \ldots, U(\sx)^{\gamma_D} \right) S^{-1}.
	$$
	Thus, if $\{v_n\}$ is a basis in $\cF$ of joint unit eigenvectors such that
	$\rho(\sigma_j^{-1})v_n = e^{i\gamma_n\pi}v_n$, then we may consider the 
	subspace in $H^1_\rho$ consisting of functions
	\begin{equation}\label{eq:abelian-magnetic}
		\Psi(\tX) = \sum_{n=1}^D U(\tX)^{\gamma_n} \Phi_n(X) v_n,
	\end{equation}
	where $\Phi_n \in C^\infty_{c,\sym}(\R^{2N} \setminus \bDelta^N)$, for which
	$$
		\int_{\cC^N} |\Psi|^2 = \sum_{n=1}^D \frac{1}{N!} \int_{\R^{2N}} |\Phi_n|^2,
		\qquad
		T_\rho[\Psi] = \sum_{n=1}^D T_{\gamma_n}[\Phi_n],
	$$
	i.e. $(\Phi_n) \in (\cH_+^N)^D$ models a collection of 
	$N$-anyon wave functions in the magnetic representation
	with statistics parameters $\alpha=\gamma_n$, $n \in \{1,\ldots,D\}$.
	
	For a glimpse of how such abelian models may arise in a FQHE context, see
	e.g.\ the brief introduction \cite[Chapter~2]{Tournois-20},
	and \cite{LunRou-16} for a more precise application
	(accounting for certain ambiguities \cite{Forte-91,Forte-92} 
	in the original derivation \cite{AroSchWil-84}).
	Another recent realization is via quantum impurity problems, polarons and angulons
	\cite{Yakaboylu-etal-19,Yakaboylu-etal-20}.

\subsection{Fibonacci anyons}\label{sec:ham-fib}

	We consider the representation $\rho_N\colon B_N \to \sU(\cF^N)$
	on the splitting space given in Definition~\ref{def:splitting-rep}.
	Imposing that we should have exactly $N$ anyons, as well as irreducibility,
	we may thus take $\cF^N_c := V^{\tau^N}_c$, where
	there is a remaining choice to be made for the total charge $c \in \{1,\tau\}$:
	$\cF^N_1 \cong \C^{\Fib(N-1)}$ and 
	$\cF^N_\tau \cong \C^{\Fib(N)}$.
	
	Presently we do not know if this model is transmutable or not.
	However, there is a correspondence to FQHE and CFT 
	(Read-Rezayi states \cite{ReaRez-99,CapGeoTod-01,ArdKedSto-05,Lundholm-16})
	which is similar to that for the Ising model discussed below, 
	and which seems to suggest transmutability
	\cite{ArdSch-07,TouArd-20}.

\subsection{Ising anyons}\label{sec:ham-ising}

	Here we could work with two types of particles 
	(the Ising anyon $\sigma$ and the Majorana fermion $\psi$)
	explicitly present in the system
	and having dynamics,
	but to simplify matters slightly we consider 
	$N$ anyons of only the more interesting type $\sigma$,
	with fiber $\cF^N = V^{\sigma^N}_\sigma \cong \C^{2^{(N-1)/2}}$ for $N$ odd and
	$\cF^N_c = V^{\sigma^N}_c \cong \C^{2^{N/2-1}}$ for $N$ even, 
	where either $c=\sigma$ or $c=\psi$.
	
	The corresponding model, also including the fermions $\psi$,
	arises in the FQHE context, supposedly as an effective model
	for quasiparticles modeled by a family of electron (fermion) wave functions known as 
	Moore-Read or Pfaffian states \cite{MooRea-91,NayWil-96}. 
	Another representation for these states are as 
	correlators of a CFT related to the Ising model. 
	This story is beyond the scope of the present article 
	(see e.g.\ \cite{HanHerSimVie-17} and \cite[Chapter~3.4]{Tournois-20} for recent reviews),
	but it suffices to note that, given positions $\bx_j \in \R^2$, $j=1,\ldots,N$ 
	lifting to $\tX \in \tcC^N$ there exists a family of states $\{\Psi_n(\tX)\}_{n=1}^D$
	with values in an $M$-fermion Hilbert space $\C^D \cong \cF_N \subseteq L^2_\asym(\R^{2M})$
	and some constant $C>0$ 
	such that in the limit $M \to \infty$ 
	$$
		\langle \Psi_n(\tX), \Psi_m(\tX) \rangle_{\cF_N} = C\delta_{nm} + O(e^{-|\bx_j-\bx_k|})
	$$
	if $|\bx_j-\bx_k| \to \infty$ \cite{GurNay-97,BonGurNay-11}.
	Thus the model appears to be transmutable by the existence of this trivialization, 
	at least in some limit and for $N=3$ or 4, and $D=2$.
	However the connection in this family of models is the Berry connection induced
	from the embedding $\cF_N \hookrightarrow L^2_\asym(\R^{2M})$ 
	and which is not necessarily locally flat.
	
	Another way to obtain its transmutability would be to use that it is
	expected, as part of the CFT/Chern--Simons correspondence, to arise from
	subrepresentations of NACS (see below) with gauge group $G=\sSU(2)$.
	More specifically, 
	the Ising fusion and braiding data are realized, up to the
	customary abelian sector and phase conventions, by level-$2$
	$\sSU(2)$ Chern--Simons theory; see
	\cite{MooRea-91,Nayak-etal-08,Kitaev-06}.

\subsection{NACS}\label{sec:ham-nacs}

	There is a certain family of representations which arise
	in the context of \keyword{Non-Abelian Chern--Simons (NACS)} theory 
	\cite{Kohno-87,GuaMarMin-90,Verlinde-91,Lo-93,LeeOh-94,BakJacPi-94,ManTroMus-13a,ManTroMus-13b}, 
	These may be explicitly defined as magnetic models on the trivial bundle \cite{Kohno-87}, 
	and are hence transmutable.
	Typically one formulates the models in a holomorphic gauge for convenience.
	
	Take a compact Lie group $G$ (for example $G=\sSU(2)$) and a unitary representation 
	$\rho_1\colon G \to \sU(\cF_1)$ (for example $\cF_1=\C^2$), 
	such that the generators $L_a$ of the Lie algebra $\mathfrak{g}$ of $G$
	are mapped to operators $\hat{L}_a \in \cL(\cF_1)$.
	For an operator-valued one-form $A$ on a manifold $M$ and a parametrized
	path $\gamma\colon[0,1]\to M$, we use
	$\exp_\eP\int_\gamma A$ to denote the endpoint $U(1)$ of the unique
	solution to the ODE initial-value problem
	\begin{equation}\label{eq:def-ordered-exp}
		\frac{dU}{dt}=U(t)A_{\gamma(t)}(\dot\gamma(t)),
		\qquad U(0)=\1.
	\end{equation}
	Equivalently, it is the convergent Dyson series in which factors are
	ordered with later parameter values to the right; this fixes our
	ordering convention, which is essential when the values of $A$ do not
	commute \cite[Chapter~10]{Nakahara-03}.
	Writing the left Maurer--Cartan form on $G$ as
	$\theta=g^{-1}dg=\sum_a L_a\theta^a$, and taking a path $\gamma$ from the
	identity to $g$, we may write
	$$
		g = \exp_\eP \int_\gamma \sum_a L_a \,\theta^a,
		\qquad
		\rho_1(g) = \exp_\eP \int_\gamma \sum_a \hat{L}_a \,\theta^a.
	$$
	Let $\cF_N=\bigotimes^N \cF_1$, with $\rho_1$ acting canonically on each of the factors
	$$
		\rho_1^{(j)}(g) = \exp_\eP \int_\gamma \sum_a \hat{L}_a^{(j)} \,\theta^a = \1 \otimes \rho_1(g) \otimes \1,
		\qquad
		\hat{L}_a^{(j)} = \1 \otimes \hat{L}_a \otimes \1.
	$$
	Defining
	$$
		U_\rho(\tX) := \exp_\eP \int_\Gamma \frac{1}{4\pi\kappa} \sum_{j \neq k} \sum_a \hat{L}_a^{(j)} \hat{L}_a^{(k)} d\log(z_j-z_k),
		\qquad \tX = [\Gamma].\tX_0,
	$$
	then $\cF_N$ carries a representation $\rho\colon B_N \to \sU(\cF_N)$ 
	with a globally defined section
	$$
		u_\rho(\tX) := U_\rho(\tX)/|U_\rho(\tX)|_{\cF_N}.
	$$
	The parameter $4\pi\kappa-2 \in \Z$ 
	is known as the level of the representation,
	and $U_\rho$ formally solves the Knizhnik-Zamolodchikov equations 
	\cite{KniZam-84}, \cite{Kohno-87}, \cite[Eq.~(39)]{LeeOh-94}.
	
	The representation $\rho$ will typically be reducible into irreducible blocks (`conformal blocks')
	$\rho_n$ and one could say that, while $\rho_n$ may not be transmutable individually,
	their sum $\rho$ is.
	
\subsection{Local kinetic energy}\label{sec:ham-energy}

	Given a subset $\Omega \subseteq \R^2$ 
	we define the \keyword{local configuration space} of identical particles on $\Omega$
	$$
		\cC^N(\Omega) := \left( \Omega^N \setminus \bDelta^N \right) \Big/ S_N
		= \eP_N(\Omega),
	$$
	the set of $N$-point subsets of $\Omega$.
	Let now $\Omega$ be open and simply connected (homotopic to $\R^2$).
	Similarly to the full space we may again define the covering space
	$$
		\tcC^N(\Omega) \to \cC^N(\Omega)
	$$
	as the space of continuous paths in $\cC^N(\Omega)$ modulo homotopy equivalence.
	We choose as reference base point for these paths a point 
	$X_0(\Omega) \in \cC^N(\Omega)$,
	which is exactly the ordered set of points $X_0 \in \cC^N(\R^2)$ as before, 
	only translated and scaled to fit the domain $\Omega$.
	Denote a path in $\cC^N(\R^2)$ which performs this scaling and 
	translation by $\gamma_\Omega$, so that $X_0(\Omega) = [\gamma_\Omega].X_0$
	and $\tX_0(\Omega) := [\gamma_\Omega].\tX_0$.
	With the corresponding scaling and translation
	$[\gamma_\Omega \gamma_\Omega'^{-1}]$
	to match the reference points 
	$X_0(\Omega)$ and $X_0(\Omega')$
	there is then an inclusion
	$$
		\tcC^N(\Omega') \hookrightarrow \tcC^N(\Omega)
		\quad \text{if} \quad
		\Omega' \subseteq \Omega.
	$$ 
	By the homotopy equivalence $\Omega \sim \R^2$ we have
	$$
		\pi_1(\cC^N(\Omega),X_0(\Omega)) \cong \pi_1(\cC^N(\R^2),X_0) = B_N,
	$$
	$$
		[\gamma].\tX_0(\Omega) 
		= [\gamma\gamma_\Omega].\tX_0 
		= [\gamma_\Omega\gamma_\Omega^{-1}\gamma\gamma_\Omega].\tX_0 
		= \tX_0.\Ad_{\gamma_\Omega^{-1}}[\gamma],
		\quad [\gamma] \in \pi_1(\cC^N(\Omega),X_0(\Omega)),
	$$
	$$
		\tX_0(\Omega).[\eta] = [\gamma_\Omega\eta].\tX_0 
		= [\gamma_\Omega\eta\gamma_\Omega^{-1}\gamma_\Omega].\tX_0 
		= \Ad_{\gamma_\Omega}[\eta].\tX_0(\Omega),
		\quad [\eta] \in \pi_1(\cC^N(\R^2),X_0).
	$$
	The actions of paths and loops in $\tcC^N$ thus carry over to $\tcC^N(\Omega)$
	straightforwardly, considering there only paths contained in $\cC^N(\Omega)$.

\begin{definition}[Function space]
	Given an $N$-anyon model $\rho\colon B_N \to \sU(\cF)$ 
	and a simply connected 
	open domain $\Omega \subseteq \R^2$,
	the space $L^2_\rho(\tcC^N(\Omega);\cF)$
	of \keyword{ideal $N$-anyon wave functions on $\Omega$} is the closure
	of the space of $\rho$-equivariant smooth functions 
	$\Psi \in C^\infty_{\rho,c}(\tcC^N(\R^2);\cF)$ 
	with $\tilde\pr(\supp \Psi) \subseteq \cC^N(\R^2)$ compact,
	\begin{equation}\label{eq:rho-equiv-local}
		\Psi(\tX.b^{-1}) = \rho(b)\Psi(\tX),
		\qquad b \in B_N,
	\end{equation}
	w.r.t. the norm $\|\cdot\|_{L^2_\rho}$ induced by the inner product
	on $\cC^N(\Omega)$
	\begin{equation}\label{eq:L2-local}
		\langle \Phi,\Psi \rangle_{L^2_\rho} 
		:= \int_{\cC^N(\Omega)} \langle\Phi(\tX),\Psi(\tX)\rangle_\cF \,dX
		= \frac{1}{N!}\int_{\Omega^N} \langle\Phi(\tX),\Psi(\tX)\rangle_\cF \,d\sx.
	\end{equation}
\end{definition}

\begin{definition}[Kinetic energy]
	Given an $N$-anyon model $\rho\colon B_N \to \sU(\cF)$
	and a simply connected domain $\Omega \subseteq \R^2$,
	we define a \keyword{Sobolev space} $H^1_\rho(\tcC^N(\Omega);\cF)$
	of wave functions $\Psi$ with finite \keyword{expected kinetic energy on $\Omega$}
	\begin{equation}\label{eq:kinetic-local}
		T_\rho^\Omega[\Psi] := \int_{\cC^N(\Omega)} |\nabla\Psi(\tX)|_{\cF^{2N}}^2 \,dX.
	\end{equation}
	This is the Hilbert space with inner product
	\begin{equation}\label{eq:H1-local}
		\langle\Phi,\Psi\rangle_{H^1_\rho} 
		:= \int_{\cC^N(\Omega)} \left( \langle\Phi(\tX),\Psi(\tX)\rangle_\cF + \langle\nabla\Phi(\tX),\nabla\Psi(\tX)\rangle_{\cF^{2N}} \right)\,dX
	\end{equation}
	obtained by taking the closure of smooth $\rho$-equivariant functions 
	on $\tcC^N(\R^2)$ w.r.t.\ the corresponding norm,
	$$
		H^1_\rho(\tcC^N(\Omega);\cF) := \overline{C^\infty_{\rho,c}(\tcC^N;\cF)}^{H^1_\rho}.
	$$
	We may then define the \keyword{$N$-anyon ground-state energy on $\Omega$}, $N \ge 2$,
	$$
		E_N(\Omega) := \inf \left\{ \int_{\cC^N(\Omega)} |\nabla\Psi(\tX)|_{\cF^{2N}}^2 \,dX : 
		\Psi \in H^1_\rho(\tcC^N(\Omega);\cF), \int_{\cC^N(\Omega)} |\Psi|_\cF^2 \,dX = 1 \right\}.
	$$
	For convenience we also define $E_0(\Omega) := 0$ and the one-particle energy
	$$
		E_1(\Omega) := \inf \left\{\int_\Omega |\nabla\Psi|_{\cF^2}^2 : 
		\Psi \in H^1(\Omega;\cF), \int_\Omega |\Psi|_\cF^2 = 1 \right\}
		= 0.
	$$
	If we need to indicate the anyon model $\rho$ we write $E^\rho_N(\Omega)$.
\end{definition}

	Again, the $\rho$-equivariance \eqref{eq:rho-equiv-local} 
	of functions assures that the integrands
	in \eqref{eq:L2-local}, \eqref{eq:kinetic-local} and \eqref{eq:H1-local} 
	are well defined. We will also need the corresponding one-body density:
	
\begin{definition}[One-body density]\label{def:density}
	Let $\rho\colon B_N \to \sU(\cF)$ be an $N$-anyon model and $\Omega \subseteq \R^2$ 
	a simply connected open subset.
	For $\Psi \in L^2_\rho(\tcC^N(\Omega);\cF)$ a normalized
	$N$-body wave function on $\Omega$ we
	define its \keyword{one-body density} on $\Omega$ as the function
	$\varrho_\Psi \in L^1(\Omega)$, where $\varrho_\Psi(\bx)$ at $\bx \in \Omega$
	is the marginal of the probability distribution $|\Psi|_\cF^2$ at 
	$X = \{\bx\} \cup X'$,
	\begin{equation}\label{eq:def-1b-density}
		\varrho_\Psi(\bx) := \int_{\cC^{N-1}(\Omega \setminus \{\bx\})} |\Psi(\{\bx\}\cup X')|_\cF^2 \,dX'
		= \frac{1}{(N-1)!} \int_{\Omega^{N-1}} |\Psi(\bx,\sx')|_\cF^2 \,d\sx'.
	\end{equation}
	It may equivalently be defined via its integral over a subset $\Omega' \subseteq \Omega$
	$$
		\int_{\Omega'} \varrho_\Psi = \int_{\cC^N(\Omega)} \sum_{j=1}^N \1_{\{\bx_j \in \Omega'\}} |\Psi|_\cF^2 \,dX,
	$$
	(the expected number of particles on $\Omega'$)
	and the Lebesgue differentiation theorem, taking $\Omega' = B_\eps(\bx)$, $\eps \to 0$.
\end{definition}

\begin{definition}[Energy density]\label{def:energy-density}
	Similarly, we define for $\Psi \in H^1_\rho(\tcC^N(\Omega);\cF)$
	a \keyword{local expected kinetic energy on $\Omega' \subseteq \Omega$}
	$$
		T^{\Omega' \subseteq \Omega}_\rho[\Psi] 
		:= \int_{\Omega'} T_{\rho,\Psi}(\bx) \,d\bx
		= \int_{\cC^N(\Omega)} \sum_{j=1}^N \1_{\{\bx_j \in \Omega'\}} |\nabla_j\Psi|_{\cF^2}^2 \,dX,
	$$
	and its corresponding \keyword{kinetic energy density} $T_{\rho,\Psi} \in L^1(\Omega)$,
	$$
		T_{\rho,\Psi}(\bx) := \int_{\cC^{N-1}(\Omega \setminus \{\bx\})} |\nabla_\bx \Psi(\{\bx\}\cup X')|_{\cF^2}^2 \,dX'.
	$$
\end{definition}

	Occasionally we shall suppress the $\cF$ on the norms to make the notation less heavy.


\section{Statistical repulsion}\label{sec:repulsion}

	Given an $N$-anyon model $\rho\colon B_N \to \sU(\cF)$ 
	with associated exchange operators $\{U_p\}_{p=0}^{N-2}$, we may 
	define an \keyword{exchange parameter} for $p$ \keyword{\emph{enclosed}} particles
	$$
		\beta_p := \min \{ \beta \in [0,1] : \text{$e^{i\beta\pi}$ or $e^{-i\beta\pi}$ is an eigenvalue of $U_p$} \},
	$$
	and for $n \in \{2,3,\ldots,N\}$ \keyword{\emph{involved}} particles 
	(i.e. worst case of up to $n-2$ enclosed)
	$$
		\alpha_n := \min_{p \in \{0,1,2,\ldots,n-2\}} \beta_p.
	$$
	To a simple two-particle exchange $U_0 = \rho(\sigma_1)$ 
	is then associated the parameter
	$\alpha_2 = \beta_0$.	
	In the case of irreducible abelian anyons, with exchange phases 
	$U_p = e^{i(2p+1)\alpha\pi}$, we have
	$$
		\alpha_n = \min\limits_{p \in \{0, 1, \ldots, n-2\}} \min\limits_{q \in \Z} |(2p+1)\alpha - 2q|.
	$$
	
	The goal of this section is to derive explicit lower bounds
	for the local kinetic energy in terms of the above model-dependent parameters
	which quantify the strength of repulsion among the particles due to their
	statistics (exchange phases).

\subsection{Relative configuration space}\label{sec:repulsion-rel}

	Let $\Omega \subseteq \R^2$ be a convex open set.
	Consider the symmetric exchange of two particles inside $\Omega$ with the 
	remaining $N-2$ particles fixed on $\Omega$.
	It is convenient to switch to relative and center-of-mass coordinates for
	the pair $(\bx_1,\bx_2)$:
	$$
		\br = \bx_1-\bx_2, \qquad
		\bR = {\textstyle\frac{1}{2}}(\bx_1+\bx_2) \in \Omega, \qquad
		\sx' = (\bx_3,\ldots,\bx_N) \in \Omega^{N-2},
	$$
	$$
		\sx(\br;\bR;\sx') := \bigl( \bR+{\textstyle\frac{1}{2}}\br, \bR-{\textstyle\frac{1}{2}}\br, \bx_3, \ldots, \bx_N \bigr) 
		\in \R^{2N}
	$$
	To keep all points in $\Omega$ 
	we must restrict the values of $\br \in \R^2$ to the subset
	$$
		\Omega_\rel(\bR;\sx') = \left\{ \br \in \R^2 : \sx(\br;\bR;\sx') \in \Omega^N \setminus \bDelta^N \right\}.
	$$
	For these we have a surjection
	$$
		X(\br;\bR;\sx') := \pr(\sx(\br;\bR;\sx'))
		= \bigl\{ \bR+{\textstyle\frac{1}{2}}\br, \bR-{\textstyle\frac{1}{2}}\br, \bx_3, \ldots, \bx_N \bigr\} 
		\in \cC^N(\Omega).
	$$
	However, $\Omega_\rel$ is not simply connected.
	For example, in the case $\Omega = \R^2$ we have
	$$
		\Omega_\rel = \R^2 \setminus\{\0, \pm 2(\bx_3-\bR), \ldots, \pm 2(\bx_N-\bR) \}
	$$
	and
	$$
		X(\pm\be_1; {\textstyle\frac{3}{2}}\be_1; 3\be_1,\ldots,N\be_1) = X_0
		= X_0(\R^2),
	$$
	while if $\bar\Omega = Q_0 = [0,1]^2$ is the unit square then
	$\Omega_\rel \subseteq [-1,1]^2$ 
	is a rectangle centered at $\0$ with up to $2N+1$ points removed
	symmetrically around the origin (one point being the origin).
	
	Note that several configurations $(\br;\bR;\sx')$ map to
	the same point $X(\br;\bR;\sx')$,
	and we have for example the antipodal symmetry
	$$
		X(-\br;\bR;\sx') = X(\br;\bR;\sx'), 
		\qquad \br \in \Omega_\rel(\bR;\sx').
	$$ 
	Upon fixing $\bR$ and $\sx'$ and a base point $\br_0 \in \Omega_\rel$ we may
	extend $\br \mapsto X(\br;\bR;\sx')$ to a smooth map 
	on the covering space 
	$\tilde\Omega_\rel(\bR;\sx') \to \Omega_\rel(\bR;\sx')$
	$$
		\tilde\br \mapsto \tX(\tilde\br;\bR;\sx') \mapsto X(\br;\bR;\sx')
	$$
	such that $\tX(\tilde\br;\bR;\sx'.\sigma) = \tX(\tilde\br;\bR;\sx').b$
	for $\sigma \in S_{N-2}$ and some $b \in B_N$. 
	By composing with $\Psi \in C^\infty_{\rho,c}$
	we then have a smooth map
	\begin{equation}\label{eq:rel-func-def}
		\psi\colon \tilde\Omega_\rel \to \cF, \qquad
		\psi(\tilde\br) := \Psi\bigl(\tX(\tilde\br,\bR,\sx')\bigr)
	\end{equation}
	such that
	$$
		\partial_k \psi(\tilde\br) 
		= \sum_{j,q} \frac{\partial\Psi}{\partial x_{jq}} \frac{\partial x_{jq}}{\partial r_k}(\tilde\br)
		= \pm\frac{1}{2}\left( \frac{\partial\Psi}{\partial x_{1k}} - \frac{\partial\Psi}{\partial x_{2k}} \right)(\tX(\tilde\br)),
	$$
	where the sign depends on the order of exchange of particles 1 and 2 done in $\tilde\br$.
	Namely, the action of any $\beta \in \pi_1(\Omega_\rel,\br_0)$ 
	carries over to a corresponding one on $\tX$,
	$$
		\tX(\tilde\br.\beta;\bR;\sx') = \tX(\tilde\br;\bR;\sx').b,
		\qquad \text{some}\ b=b(\beta) \in B_N,
	$$
	so that, with the transformation rule \eqref{eq:diff-equiv},
	$$
		\partial_k \psi(\tilde\br.\beta) 
		= \pm\frac{1}{2} \rho(b^{-1})\left( \frac{\partial\Psi}{\partial x_{1'k}} - \frac{\partial\Psi}{\partial x_{2'k}} \right)(\tX(\tilde\br)),
		\quad j' = \hat\pr(b)[j],
	$$
	and hence
	$$
		|\nabla\psi(\tilde\br.\beta)|^2 = \frac{1}{4}|\nabla_1\Psi-\nabla_2\Psi|^2
	$$
	for any $\beta \mapsto b \mapsto \hat\pr(b)$ wich 
	leaves invariant the partition 
	$\{1,2\},\{3,4, \ldots, N\}$.
	
	Furthermore, given any $r=|\br|>0$ such that the full circle $rS^1$ 
	is contained in $\Omega_\rel$, 
	we may consider the corresponding lift of this circle to a subset 
	in $\tilde\Omega_\rel$ isomorphic to $\R$.
	By the convexity of $\Omega$,
	this will be possible on any of the annuli $A_p \subseteq \Omega_\rel$ 
	defined by $r \in I_p$,
	$$
		I_0 = (r_0:=0,r_1),\ I_1 = (r_1,r_2),\ \ldots,\ I_{M-1} = (r_{M-1},r_M),\ I_M = (r_M,r_{\max}),
		\ 0 \le M \le N-2,
	$$
	with $r_{\max} := \dist(\0,\partial\overline{\Omega_\rel}) = 2\dist(\bR,\partial\overline{\Omega})$,
	and having written the relative distances to the other particles in increasing order
	$$
		0 \le r_1 = 2|\bx_{2+\sigma(1)}-\bR| \le r_2 = 2|\bx_{2+\sigma(2)}-\bR| 
		\le \ldots \le r_{N-2} = 2|\bx_{2+\sigma(N-2)}-\bR|
	$$
	for a suitable permutation $\sigma \in S_{N-2}$.
	
	Comparing now
	to the two-particle relative configuration space
	$\tcC_0 \to \cC_0$ in Example~\ref{ex:2-anyon-cspace},
	we have the curve $\R \to S^1 \to S^1/{\sim} = \cC_0$
	(with $\bx \sim -\bx$) 
	of relative angles
	$$
		\varphi \mapsto \be(\varphi) := (\cos\vphi,\sin\vphi) 
		\mapsto \{\pm(\cos\vphi,\sin\vphi)\},
	$$
	which after fixing a base point $\tilde\be(0) \in \tcC_0$ 
	lifts to a unique curve $\R \to \tcC_0$
	$$
		\vphi \mapsto \tilde\be(\vphi) := [\gamma_\vphi].\tilde\be(0), \qquad
		\gamma_\vphi(t) := \{\pm\be(\vphi t)\},
		\ t \in [0,1].
	$$
	In the same way there is for $r \in I_p$ a corresponding smooth curve 
	$\R \to \tcC^N(\Omega) \to \cC^N(\Omega)$
	obtained by taking $\tilde\br = r\tilde\be(\vphi)$ with the remaining
	coordinates fixed, i.e.
	$$
		\vphi \mapsto \tX(r\tilde\be(\vphi);\bR;\sx') \mapsto X(r\be(\vphi);\bR;\sx').
	$$
	Given a base point on the curve, for example
	$$
		\tX(r\tilde\be(0);\bR;\sx') = [\eta].\tX_0(\Omega) = [\eta\gamma_\Omega].\tX_0
	$$
	for some fixed path $\eta$ from $X_0(\Omega)$ to $X(r\be(0);\bR;\sx')$
	(possibly via $X(\br_0;\bR;\sx')$),
	this curve can also be represented as
	$$
		\vphi \mapsto \tX(r\tilde\be(\vphi);\bR;\sx') = [\Gamma_\vphi].\tX(r\tilde\be(0);\bR;\sx')
	$$
	using the path in $\cC^N(\Omega)$
	$$
		\Gamma_\vphi(t) := \bigl\{ \bR+{\textstyle\frac{1}{2}}r\gamma_\vphi(t), \bR-{\textstyle\frac{1}{2}}r\gamma_\vphi(t), \bx_3, \ldots, \bx_N \bigr\},
		\ t \in [0,1].
	$$
	We may then define a smooth function $v\colon \R \to \cF$ as the restriction of 
	$\Psi$ to this particular curve,
	\begin{equation}\label{eq:angular-func-def}
		v(\vphi) := \psi(r\tilde\be(\vphi)) = \Psi(\tX(r\tilde\be(\vphi);\bR;\sx')),
	\end{equation}
	with
	$$
		v'(\vphi) 
		= \nabla \psi \cdot r\tilde\be'(\vphi)
		= \pm\left[ \frac{1}{2}\left( \frac{\partial\Psi}{\partial x_{11}} - \frac{\partial\Psi}{\partial x_{21}} \right) (-r\sin\vphi)
		+ \frac{1}{2}\left( \frac{\partial\Psi}{\partial x_{12}} - \frac{\partial\Psi}{\partial x_{22}} \right) (r\cos\vphi) \right].
	$$
	We also compute
	$$
		\partial_r \psi( r\tilde\be(\vphi) ) = \nabla \psi \cdot \be(\vphi)
		= \pm\left[ \frac{1}{2}\left( \frac{\partial\Psi}{\partial x_{11}} - \frac{\partial\Psi}{\partial x_{21}} \right) \cos\vphi
		+ \frac{1}{2}\left( \frac{\partial\Psi}{\partial x_{12}} - \frac{\partial\Psi}{\partial x_{22}} \right) \sin\vphi \right],
	$$
	which combines with the previous expression to yield
	\begin{equation}\label{eq:psi-derivative}
		|\nabla \psi|_{\cF^2}^2 
		= |\partial_r \psi( r\tilde\be(\vphi) )|_\cF^2 + \frac{1}{r^2}|v'(\vphi)|_\cF^2
		= \frac{1}{4}|\nabla_1\Psi - \nabla_2\Psi|_{\cF^2}^2.
	\end{equation}
	
	A two-particle exchange on $\cC_0$ with winding number $n \in \Z$ 
	yields $\be(n\pi) = (-1)^n \be(0)$ and
	$$
		\tilde\be(n\pi) = [\gamma_{n\pi}].\tilde\be(0) 
		= [\gamma_\pi]^n.\tilde\be(0) = \tilde\be(0).[\gamma_\pi]^n,
	$$
	with $[\gamma_\pi]$ the generator of the respective braid group 
	$\pi_1(\cC_0) = \Z$ (base point $\{\pm\be(0)\}$).
	With our corresponding curve on $\cC^N(\Omega)$
	$$
		\tX(r\tilde\be(n\pi);\bR;\sx') 
		= [\Gamma_{n\pi}].\tX(r\tilde\be(0);\bR;\sx')
		= \tX(r\tilde\be(0);\bR;\sx').b^n,
	$$
	with $b = \Ad_{[\eta\gamma_\Omega]^{-1}}([\Gamma_\pi]) \in B_N$
	the generator of exchange.
	In this case we perform continuous two-particle exchanges which are not 
	necessarily simple but may, depending on $r$, involve other particles.
	Topologically we have a situation similar to \eqref{eq:exchange-braid}, 
	i.e. $[\Gamma_\pi] \sim b \sim \Sigma_p$  where there are $p$ enclosed particles 
	if $r \in I_p$.
	Thus the $\rho$-equivariance of $\Psi \in C^\infty_\rho$ implies
	$$
		v(n\pi) = \Psi(\tX(r\tilde\be(n\pi);\bR;\sx'))
		= \Psi(\tX(r\tilde\be(0);\bR;\sx').b^n)
		= \rho(b^{-n})v(0)
	$$
	where $\rho(b) \sim \rho(\Sigma_p) = U_p$, 
	the corresponding two-anyon exchange operator.

\subsection{Poincar\'e inequality}\label{sec:repulsion-Poincare}

	The statistical repulsion between a pair of anyons boils down to the 
	following simple version of the Poincar\'e inequality
    (Wirtinger's inequality).

\begin{lemma}[Poincar\'e inequality for semi-periodic functions]\label{lem:Poincare}
	Let $\psi \in H^1([0,\pi];\cF)$ satisfy the semi-periodic boundary condition
	\begin{equation}\label{eq:Poincare-bc}
		\psi(\pi) = U \psi(0),
	\end{equation}
	for some $U \in \sU(\cF)$, then
	\begin{equation}\label{eq:Poincare-ineq}
		\int_0^\pi |\psi'(\vphi)|_\cF^2 \,d\vphi 
		\ge \lambda_0(U)^2 \int_0^\pi |\psi(\vphi)|_\cF^2 \,d\vphi,
	\end{equation}
	where 
	\begin{equation}\label{eq:Poincare-lambda}
		\lambda_0(U) := \inf \{ \lambda \in [0,1] : 
			\text{$e^{i\lambda\pi}$ or $e^{-i\lambda\pi}$ 
			is an eigenvalue of $U$} \}.
	\end{equation}
\end{lemma}
\begin{proof}
	We consider the operator $D$ on the space $L^2([0,\pi];\cF)$
	with $D\psi(\vphi) = -i\psi'(\vphi)$ and domain given by functions
	$\psi \in H^1([0,\pi];\cF)$ satisfying the b.c. \eqref{eq:Poincare-bc}.
	This is a self-adjoint operator and its spectrum is given explicitly by
	$\lambda_{n,k} = \mu_n + 2k \in \R$, $n \in \{1,\ldots,\dim\cF\}$,
	$k \in \Z$, 
	with corresponding orthonormal eigenfunctions 
	$u_{n,k}(\vphi) = \pi^{-1/2}e^{i\lambda_{n,k}\vphi}v_n$,
	where $U v_n = e^{i\mu_n \pi} v_n$, $\mu_n \in (-1,1]$,
	a basis in $\cF$ of eigenvectors of $U$.
	By the spectral theorem, the l.h.s. of \eqref{eq:Poincare-ineq} is then
	$$
		\|D\psi\|^2
		= \sum_{n,k} \lambda_{n,k}^2
			\bigl|\langle u_{n,k},\psi\rangle\bigr|^2
		\ge \inf_{n,k} \lambda_{n,k}^2 \|\psi\|^2
		= \lambda_0(U)^2 \|\psi\|^2,
	$$
	according to the definition \eqref{eq:Poincare-lambda}.
	This quadratic-form identity is valid for every $\psi\in\cD(D)$; 
	no assumption that $\psi\in\cD(D^2)$ is required.
\end{proof}

	For completeness we also state a version suitable for magnetic
	anyon models modeled using bosons:

\begin{lemma}[Gauge-transformed Poincar\'e inequality]\label{lem:Poincare-gauge}
	Let $\psi \in H^1([0,\pi];\cF)$ satisfy the periodic boundary condition
	\begin{equation}\label{eq:Poincare-bc-gauge}
		\psi(\pi) = \psi(0),
	\end{equation}
	and assume $\cA \in C([0,\pi];\gu(\cF))$, where $\gu(\cF)$ are the
	anti-hermitian operators on $\cF$. 
	Then
	\begin{equation}\label{eq:Poincare-ineq-gauge}
		\int_0^\pi \left|\left(\frac{d}{d\varphi} + \cA\right)\psi(\varphi)\right|_\cF^2 d\varphi 
		\ge \lambda_0(U)^2 \int_0^\pi |\psi(\varphi)|_\cF^2 \,d\varphi,
	\end{equation}
	where 
	$U = \exp_\eP (\int_0^\pi \cA) \in \sU(\cF)$ 
	is the path-ordered exponential of $\cA$.
\end{lemma}
\begin{proof}
	Define the function $v(\vphi) = U(\vphi) \psi(\vphi)$,
	where $U(\vphi) := \exp_\eP (\int_0^\vphi \cA(t) dt)$,
	i.e., by definition \eqref{eq:def-ordered-exp}, $U(0) = \1$ and
	$U'(\vphi) = U(\vphi)\cA(\vphi)$ for all $\vphi \in [0,\pi]$.
	Then
	$v'(\vphi) = U(\vphi)\bigl(\psi'(\vphi)+\cA(\vphi)\psi(\vphi)\bigr)$,
	so the l.h.s.
	of \eqref{eq:Poincare-ineq-gauge} is $\int_0^\pi |v'|^2$,
	and furthermore $v$ satisfies the b.c.
	$$
		v(\pi) = U(\pi)\psi(\pi) = U\psi(0) = U v(0).
	$$
	Now we may apply Lemma~\ref{lem:Poincare} to $v$ and
	use $|v| = |\psi|$ to conclude the lemma.
\end{proof}

\subsection{Diamagnetic inequality}\label{sec:repulsion-diamag}

	Consider estimates due to a polar decomposition of $\cF$:

\begin{lemma}[Angular diamagnetic inequality]\label{lem:diamag-circle}
	Let $\psi \in H^1([0,\pi];\cF)$, 
	$|\psi|_\cF = \left( \sum_k |\psi_k|^2 \right)^{1/2}$ where $\psi_k$ are
	the components in an ON basis of $\cF$. Then
	$|\psi|_\cF \in H^1([0,\pi];\R_+)$, 
	\begin{equation}\label{eq:diamag-ineq-pwc}
		|\psi'(\vphi)|_\cF^2 \ge \bigl||\psi|_\cF'(\vphi)\bigr|^2
	\end{equation}
	pointwise a.e. on $[0,\pi]$, and thus
	\begin{equation}\label{eq:diamag-ineq-circle}
		\int_0^\pi |\psi'(\vphi)|_\cF^2 \,d\vphi 
		\ge \int_0^\pi \bigl||\psi|_\cF'(\vphi)\bigr|^2 \,d\vphi.
	\end{equation}
\end{lemma}
\begin{proof}
	This is a direct consequence of the usual diamagnetic inequality for functions
	with values in $\C \cong \R^2$
	(see e.g. \cite[Theorem 6.17]{LieLos-01}), i.e. pointwise a.e.
	$$
		|(u_1,u_2)'|^2 = |u_1'|^2 + |u_2'|^2 \ge |\sqrt{u_1^2 + u_2^2}'|^2,
	$$
	where the r.h.s. is zero if $u_1^2+u_2^2=0$.
	Applied inductively,
	$$
		|u_1'|^2 + \ldots + |u_{n-1}'|^2 + |u_n'|^2 
		\ge |\sqrt{u_1^2+\ldots+u_{n-1}^2}'|^2 + |u_n'|^2 
		\ge |\sqrt{u_1^2+\ldots+u_n^2}'|^2,
	$$
	etc., extends to $\cF \cong \C^D \cong \R^{2D}$.
	Thus, pointwise a.e.
	$$
		|\psi'|_\cF^2 = \sum_k |\psi_k'|^2 \ge \bigl||\psi|_\cF'\bigr|^2,
	$$
	and since $|\psi|_\cF \in L^2$ and $\psi' \in L^2$ we also have
	$|\psi|_\cF \in H^1$ and the bound \eqref{eq:diamag-ineq-circle}.
\end{proof}

\begin{lemma}[Spatial diamagnetic inequality]\label{lem:diamag}
	Let $\Omega \subseteq \R^2$ be a simply connected Lipschitz domain. 
	Let $\Psi \in H^1_\rho(\tcC^N(\Omega);\cF)$, 
	$|\Psi|_\cF = \left( \sum_k |\psi_k|^2 \right)^{1/2}$, then
	$|\Psi|_\cF \in H^1_+(\tcC^N(\Omega);\R_+)$ extends to a function
	$|\Psi|_\cF \in H^1_\sym(\Omega^N;\R_+)$, 
	\begin{equation}\label{eq:diamag-ineq-pw}
		|\nabla\Psi(\tX)|_{\cF^{2N}}^2 
		\ge \big| \nabla|\Psi|_\cF(\tX) \big|_{\R^{2N}}^2
	\end{equation}
	pointwise a.e. on $\cC^N(\Omega)$, and
	\begin{equation}\label{eq:diamag-ineq}
		\int_{\cC^N(\Omega)} |\nabla\Psi(\tX)|_{\cF^{2N}}^2 \,dX
		\ge \int_{\cC^N(\Omega)} \big| \nabla|\Psi|_\cF(\tX) \big|_{\R^{2N}}^2 \,dX
		= \frac{1}{N!} \int_{\Omega^N} \big| \nabla|\Psi|_\cF(\sx) \big|_{\R^{2N}}^2 \,d\sx.
	\end{equation}
\end{lemma}
\begin{proof}
	The inequality \eqref{eq:diamag-ineq-pwc} may be replaced by a derivative in any 
	direction,
	$$
		|\partial_{j,k}\Psi|_\cF^2 \ge \big| \partial_{j,k}|\Psi|_\cF \big|^2.
	$$
	We therefore have a pointwise a.e. inequality on 
	$\tcC^N(\Omega)$ as above,
	$$
		|\nabla\Psi|_{\cF^{2N}}^2 
		= \sum_{j,k} |\partial_{j,k}\Psi|_\cF^2 
		\ge \sum_{j,k} \big| \partial_{j,k}|\Psi|_\cF \big|^2
		= \big| \nabla|\Psi|_\cF \big|_{\R^{2N}}^2.
	$$
	Furthermore both sides of this inequality are invariant 
	under the action of $B_N$ due to the $\rho$-equivariance of $\Psi$,
	and thus descend to $\cC^N(\Omega)$.
	Hence we have a function $|\Psi|_\cF \in H^1_{\rho_+}(\cC^N(\Omega);\R_+)$ satisfying
	the bound \eqref{eq:diamag-ineq}. 
	This extends to a function $|\Psi|_\cF \in H^1_\sym(\Omega^N\setminus\bDelta^N;\R_+)$.
	We may also consider $|\Psi|_\cF$ as a function in 
	$L^2_\sym(\Omega^N) = L^2_\sym(\Omega^{2N} \setminus \bDelta^N)$ 
	and its distributional derivative $\nabla|\Psi|_\cF \in H^{-1}(\Omega^N)$.
	Since $H^1(\R^{2N} \setminus \bDelta^N) = H^1(\R^{2N})$
	this function extends to $H^1_\sym(\Omega^N)$;
	cf. \cite[Lemmas~3 and 4]{LunSol-14} and more generally \cite[Appendix~B]{LarLunNam-19}.
	In fact $\Psi \in H^1_\rho$ is the limit of $\Psi_n \in C^\infty_{\rho,c}$ and
	the inequality \eqref{eq:diamag-ineq} for such functions
	implies for the limit 
	$|\Psi|_\cF \in H^1(\Omega^N) \cap H^1_0(\R^{2N} \setminus \bDelta^N) = H^1(\Omega^N)$.
\end{proof}

	The diamagnetic inequality shows that $\infspec \hT_\rho \ge \infspec \hT_{\rho_+}$,
	i.e.\ bosons always provide a lower bound to the energy, 
	since given any
	$\Psi \in H^1_\rho(\tcC^N(\Omega);\cF)$ 
	or sequence $\Psi_n \in C^\infty_\rho$ converging to $\Psi$
	we obtain
	$$
		T_\rho^\Omega[\Psi_n] 
		\ge T_{\rho_+}^\Omega\bigl[|\Psi_n|_\cF\bigr] 
		\ge E_N^{\rho_+}(\Omega),
	$$
	and thus $E_N^\rho(\Omega) \ge E_N^{\rho_+}(\Omega)$.
	Furthermore, we obtain by finiteness of the integral
	$$
		T_{\rho_+}^\Omega\bigl[|\Psi|_\cF\bigr] = T_0[\Phi] 
		= \frac{1}{(N-1)!} \int_{\Omega^{N-1}} \int_{\Omega} |\nabla_\bx \Phi(\bx,\sx')|^2 d\bx \,d\sx',
		\quad \Phi(\sx) = |\Psi(X)|_\cF,
	$$
	the Sobolev embedding $|\Psi(\cdot,\sx')|_\cF \in L^p(\Omega)$, $2 \le p < \infty$, 
	for a.e. $\sx' \in \Omega^{N-1}$
	\cite[Theorems~8.5 and 8.8]{LieLos-01}.

\subsection{Hardy inequality}\label{sec:repulsion-Hardy}

	We are now ready to prove our first main result concerning the dynamics of the anyon gas, 
	which extends the many-body Hardy inequality \cite[Theorem~1.3]{LarLun-16}
	for irreducible abelian anyons in the magnetic representation
	to arbitrary geometric anyon models.

\begin{theorem}[\keyword{Hardy inequality for non-abelian anyons}]\label{thm:Hardy}\mbox{}\\
	Let $\rho\colon B_N \to \sU(\cF)$ be an $N$-anyon model 
	with exchange parameters $(\beta_p)_{p\in \{0,\ldots,N-2\}}$ and 
	$\alpha_N = \min_{p \in \{0,\ldots,N-2\}} \beta_p$, and let
	$\Omega \subseteq \R^2$ be open and convex. Then, for any 
	$\Psi \in H^1_\rho(\tcC^N(\Omega);\cF)$,
	\begin{align}
		\int_{\cC^N(\Omega)} |\nabla\Psi|_{\cF^{2N}}^2 \, dX
		&\ge \frac{1}{N} \int_{\cC^N(\Omega)} \biggl| \sum_{j=1}^N \nabla_j \Psi \biggr|_{\cF^2}^2 dX \\
		&\qquad + \frac{4}{N} \int_{\cC^N(\Omega)} \sum_{j < k} \biggl(
			\bigl| \partial_{r_{jk}} |\Psi|_\cF \bigr|^2 + 
			\sum_{p=0}^M \beta_p^2 \,\1_{A_p}(\br_{jk}) \frac{|\Psi|_\cF^2}{r_{jk}^2}
			\biggr) dX \\
		&\ge \pi (j_{\alpha_N}')^2 \frac{4}{N} \int_{\cC^N(\Omega)} \sum_{j < k} 
			\frac{ \1_{\Omega \circ \Omega}(\bx_j, \bx_k) }{ 4\pi\delta(\bX_{jk})^2 }
			\, |\Psi|_\cF^2 \, dX, 
	\end{align}
	where 
	$0 \le M \le N-2$ and 
	the relative annuli $A_p$ (defined in Section~\ref{sec:repulsion-rel})
	depend on the relative positions of all particles with the pair $j<k$ singled out,
	the reduced support
	$$
		\1_{\Omega \circ \Omega}(\bx_j, \bx_k) := \1_{B_{2\delta(\bX_{jk})}(0)}(\br_{jk})
		= \sum_{p=0}^M \1_{A_p}(\br_{jk})
	$$
	defined in terms of pairwise coordinates and distances
	$$
		\br_{jk} := \bx_j-\bx_k, \qquad
		\bX_{jk} := (\bx_j+\bx_k)/2, \qquad
		r_{jk} := |\br_{jk}|, \qquad
		\delta(\bx) := \dist(\bx, \partial\Omega),
	$$
	while $j_\nu'$ denotes the first positive zero of the derivative 
	of the Bessel function $J_\nu$, 
	satisfying
	\begin{equation}\label{eq:Bessel-bounds}
		\sqrt{2\nu} \le j_\nu' \le \sqrt{2\nu(1+\nu)},
		\qquad j_0' := 0.
	\end{equation}
\end{theorem}
\begin{proof}
	We adapt the strategy of the proof of \cite[Theorem~1.3]{LarLun-16}
	to functions living on the covering space.	
	By definition of the space $H^1_\rho$, w.l.o.g.\ $\Psi \in C^\infty_{\rho,c}$ 
	where $\tilde\pr(\supp\Psi)$ is compact in $\cC^N$, 
	i.e. $\Psi$ is zero close to $\bDelta^N$.
	We may also use invariance of $|\nabla\Psi|^2 = \sum_j |\nabla_j\Psi|^2$ 
	under the action of $B_N$
	and thus of $S_N$ to replace $N!$ copies of the integral 
	$\int_{\cC^N(\Omega)}$ by $\int_{\Omega^N}$.
	
	Now use the parallelogram identity
	\begin{equation}\label{eq:parallelogram}
		\sum_{j=1}^n |\bz_j|^2 
			= \frac{1}{n} \sum_{1 \le j<k \le n} |\bz_j - \bz_k|^2
			+ \frac{1}{n} \biggl| \sum_{j=1}^n \bz_j \biggr|^2, 
	\end{equation}
	valid for any $\sz = (\bz_j)_j \in \cF^{n}$, $n \in \N$, 
	and Hilbert space $\cF$,
	to rewrite the kinetic energy
	\begin{align}\label{eq:derivative-split}
		\int_{\Omega^N} \sum_{j=1}^N |\nabla_j \Psi|^2\, d\sx 
		=& \frac{1}{N} \sum_{1\le j<k \le N} \int_{\Omega^{N-2}}\int_{\Omega^2} |\nabla_j\Psi-\nabla_k\Psi|^2 \, d\bx_j d\bx_k \prod_{l\neq j, k}d\bx_l \\
		&+ \frac{1}{N} \int_{\Omega^N} \biggl| \sum_{j=1}^N \nabla_j\Psi \biggr|^2 d\sx.
	\end{align}
	We consider the term
	$$
		\int_{\Omega^{N-2}} \int_{\Omega^2} |\nabla_1\Psi - \nabla_2\Psi|^2 d\bx_1 d\bx_2 \,d\sx'
		= \int_{\Omega^{N-2}} \int_\Omega \int_{\Omega_\rel(\bR;\sx')} |\nabla_1\Psi - \nabla_2\Psi|^2 d\br \,d\bR \,d\sx'
	$$
	where, following the strategy of Section~\ref{sec:repulsion-rel}, 
	we make the 1-to-1 coordinate transformation to relative coordinates 
	$\bR \in \Omega$, $\br \in \Omega_\rel$ 
	with measure $d\bx_1 d\bx_2 = d\br d\bR$.
	Given a configuration $\sx' \in \Omega^{N-2}$ and $\bR \in \Omega$, 
	by \eqref{eq:psi-derivative} it remains
	then to study the integral
	$$
		\int_{\Omega_\rel(\bR;\sx')} |\nabla_1\Psi - \nabla_2\Psi|^2 d\br
		= 4 \int_{\Omega_\rel(\bR;\sx')} |\nabla \psi(\tilde\br)|^2 d\br
		= 4 \sum_{p=0}^M \int_{A_p} |\nabla \psi(\tilde\br)|^2 d\br
		+ 4 \int_{A^c} |\nabla \psi(\tilde\br)|^2 d\br
	$$
	with the smooth function $\psi$ in \eqref{eq:rel-func-def} defined on the covering space
	$\tilde\Omega_\rel$ and $A^c := \Omega_\rel \setminus \overline{\cup_p A_p}$. 
	On each annulus $A_p$ we write in terms of polar coordinates
	$$
		\int_{A_p} |\nabla \psi(\tilde\br)|^2 d\br 
		= \int_{r=r_p}^{r_{p+1}} \int_{\vphi=0}^{2\pi} \left( |\partial_r \psi|^2 + \frac{1}{r^2} |\partial_\vphi \psi|^2 \right) d\vphi \,r dr.
	$$
	
	For the first term we use the pointwise diamagnetic inequality \eqref{eq:diamag-ineq-pwc}
	$$
		|\partial_r \psi| \ge \bigl|\partial_r|\psi|\bigr|
	$$
	while in the second we may at fixed $r \in I_p$ identify 
	$|\partial_\vphi \psi|^2 = |v'(\vphi)|^2$ for the smooth function \eqref{eq:angular-func-def}
	of the relative angle $\vphi \in \R$ which obeys the semi-periodicity
	$$
		v(n\pi) = \rho(b^{-1})^n v(0),
		\qquad n \in \Z, 
	$$
	with $\rho(b^{-1}) \sim U_p^{-1}$.
	We stress that the choice of a base point 
	$\tX(r\tilde\be(0);\bR;\sx') = [\eta].\tX_0(\Omega) = [\eta\gamma_\Omega].\tX_0$ 
	for the curve used in the definition of $v$ (respectively $\psi$) is arbitrary
	since if we instead took 
	$[\eta'].\tX_0(\Omega) = ([\eta].\tX_0(\Omega)).[\gamma_\Omega^{-1}\eta^{-1}\eta'\gamma_\Omega]$ 
	they differ by the braid $B = [\gamma_\Omega^{-1}\eta^{-1}\eta'\gamma_\Omega]$
	and thus, as long as it preserves our division of the particles, we have
	$$
		\partial_\vphi[ \Psi(\tX(r\tilde\be(\vphi);\bR;\sx').B) ]
		= \rho(B^{-1}) \partial_\vphi[ \Psi(\tX(r\tilde\be(\vphi);\bR;\sx')) ],
	$$
	leaving $|v'(\vphi)|^2$ unchanged. In the same way we have
	$v(\vphi+\pi) = \rho(b^{-1})v(\vphi)$ and thus
	$$
		\int_0^{2\pi} |v'|^2
		= 2\int_0^{\pi} |v'|^2 
		\ge 2 \lambda_0(U_p^{-1})^2 \int_0^{\pi} |v|^2
		= \lambda_0(U_p)^2 \int_0^{2\pi} |v|^2,
	$$
	by the Poincar\'e inequality, Lemma~\ref{lem:Poincare}.
	Therefore
	$$
		\int_{A_p} |\nabla u(\tilde\br)|^2 d\br 
		\ge \int_{r=r_p}^{r_{p+1}} \int_0^{2\pi} \left( \bigl|\partial_r|u|\bigr|^2 + \frac{\lambda_0(U_p)^2}{r^2} |u|^2 \right) d\vphi \,r dr 
	$$
	
	The remaining term can be bounded using the diamagnetic inequality,
	$$
		\int_{A^c} |\nabla \psi|^2 \ge \int_{A^c} |\partial_r \psi|^2 
		\ge \int_{A^c} \bigl|\partial_r|\psi|\bigr|^2.
	$$
	Summarizing,
	$$
		\int_{\Omega^N} |\nabla_1\Psi - \nabla_2\Psi|^2 \,d\sx
		\ge 4 \int_{\Omega^N} \left( \bigl|\partial_r|\Psi|^2\bigr| + \sum_{p=0}^M \1_{A_p}(\br)\beta_p^2 \frac{|\Psi|^2}{|\br|^2} \right) d\sx,
	$$
	which, after bounding uniformly
	$\lambda_0(U_p) = \beta_p \ge \alpha_N$ 
	in terms of the exchange parameters of the anyon model,
	simplifies further with the remaining support
	$\sum_{p=0}^M \1_{A_p}(\br) = \1_{B_{r_{\max}}(0)}(\br) = \1_{\Omega \circ \Omega}(\bx_1,\bx_2)$.
	
	The other terms in \eqref{eq:derivative-split}
	involving the pairs $(\bx_j,\bx_k)$ are exactly analogous in terms of
	corresponding relative coordinates $(\br_{jk},\bR_{jk})$.
	After collecting these we may finally
	use again the $B_N$-invariance of the collection of 
	all the terms to write the integrals
	on $\cC^N(\Omega)$.
	This proves the first bound of the theorem.
	
	The second bound of the theorem follows exactly as in \cite{LarLun-16}, 
	passing again to polar coordinates on $\Omega_\rel$ and
	considering $u(r) = |\Psi(X)|_\cF$ in
	$$
		\int_0^{r_{\max}} \left( |u'|^2 + \frac{\alpha_N^2}{r^2}|u|^2 \right) rdr
		\ge \lambda \int_0^{r_{\max}} |u|^2 \,rdr.
	$$
	The minimizer of the Rayleigh quotient satisfies the Bessel equation
	$$
		-u''(r) - u'(r)/r + \nu^2 u(r) / r^2 = \lambda u(r), 
		\quad u(0) = 0, \ u'(r_{\max}) = 0,
	$$
	i.e. $u(r)=J_\nu(j_\nu' r/r_{\max})$, 
	with the eigenvalue $\lambda = (j_\nu')^2/r_{\max}^2$,
	$\nu = \alpha_N \in [0,1]$.	
	The bounds \eqref{eq:Bessel-bounds} for $j_\nu'$ are given in 
	\cite[Appendix A]{LarLun-16}.
\end{proof}

	We note that Theorem~\ref{thm:Hardy} implies some simpler inequalities on the full plane $\Omega=\R^2$:
\begin{corollary}
    Let $\rho$ be as in Theorem~\ref{thm:Hardy}, and
	$\Psi \in H^1_\rho$. Then
	\begin{equation}\label{eq:Hardy-nn}
            T_\rho[\Psi] \ge \frac{2\alpha_2^2}{N} \sum_{j=1}^N \frac{1}{N!} \int_{\R^{2N}} \frac{|\Psi|_\cF^2}{|\bx_j-\bx_{\mathrm{nn}(j)}|^2} \,d\sx,
	\end{equation}
    where $\mathrm{nn}(j)$ denotes the nearest neighbor of particle $j$ (depending on $\sx$).
    Furthermore,
	\begin{equation}\label{eq:Hardy-simple}
		T_\rho[\Psi] \ge C_N \sum_{j<k} \frac{1}{N!} \int_{\R^{2N}} \frac{|\Psi|_\cF^2}{|\bx_j-\bx_k|^2} \,d\sx,
		\qquad
		C_N = \frac{4}{N} \max\left\{ \alpha_N^2, \frac{\alpha_2^2}{N-1} \right\}.
	\end{equation}
\end{corollary}
\begin{proof}
    For \eqref{eq:Hardy-nn}, we consider the first inequality of Theorem~\ref{thm:Hardy} and note that for any fixed $j$ in the sum over particle numbers, for one of the terms $k = \mathrm{nn}(j)$ is its nearest neighbor, which corresponds to $p=0$ and $\beta_{p=0} = \alpha_2$ in the sum over annuli. 
    Thus, the desired inequality remains after dropping the other terms which are non-negative.
    For \eqref{eq:Hardy-simple}, we simply note that
    $$
         |\bx_j-\bx_{\mathrm{nn}(j)}|^{-2} \ge (N-1)^{-1} \sum_{k \neq j} |\bx_j-\bx_k|^{-2},
    $$
    and combine that with the first inequality where we keep all the terms in $k$ and $p$ but use that $\beta_p \ge \alpha_N$ for all $p$.
\end{proof}

\begin{proposition}[Counterexamples for $\alpha_2 = 0$, global case]\label{prop:Hardy-counterex}
	A Hardy inequality of the type \eqref{eq:Hardy-simple} cannot hold 
	with any positive constant $C_N$
	in the case that $\rho\colon B_N \to \sU(D)$ is abelian with $\alpha_2=0$
	or $\rho\colon B_3 \to \sU(2)$ is the non-abelian Burau representation given in Section~\ref{sec:phases-burau} with $w=1$
	($\alpha_2=\alpha_3=0$).
\end{proposition}
\begin{proof}
	For $\rho\colon B_N \to \sU(D)$ abelian with $\alpha_2=0$, 
	we may take a joint unit eigenvector
	$v_0 \in \C^D$ such that $\rho(\sigma_j)v_0 = v_0 \ \forall j$ and consider 
	the state $\Psi = \Phi_\eps v_0 \in H^1_\sym$,
	where $\Phi_\eps \in H^1_\sym \cap C^\infty_c(\R^{2N} \setminus \bDelta^N)$ 
	approximates in $H^1$ the product state $\Phi_0(\sx) = e^{-|\sx|^2} = \prod_{j=1}^N e^{-|\bx_j|^2}$ 
	as $\eps \to 0$
    (see e.g. \cite[Lemma~3]{LunSol-14} and, more generally, \cite[Appendix~B]{LarLunNam-19}).
	Thus, $\Psi$ is $\rho$-equivariant,
	$$
		\Psi(\tX.\sigma_j^{-1}) = \Phi_\eps(X)v_0 = \rho(\sigma_j)\Psi(\tX),
	$$
	$\int_{\cC^N} |\Psi|^2 = (N!)^{-1} \int_{\R^{2N}} |\Phi_\eps|^2 \to C < \infty$, and
	$$
		T_\rho[\Psi] = \int_{\cC^N} |\nabla\Psi(\tX)|^2 \,dX 
		= \frac{1}{N!} \int_{\R^{2N}} |\nabla\Phi_\eps|^2
		\to \frac{1}{N!} \int_{\R^{2N}} |\nabla\Phi_0|^2 < \infty,
	$$
	while, due to the non-integrability of the inverse-square potential,
	$$
		\int_{\R^{2N}} \frac{|\Psi|^2}{|\bx_1-\bx_2|^2} d\sx 
		= \int_{\R^{2N}} \frac{|\Phi_\eps|^2}{|\bx_1-\bx_2|^2} d\sx \to \infty
		\ \ \text{as} \ \ \eps \to 0.
	$$
	
	For $N=2$, any representation is abelian and thus the above counterexample applies
	if and only if $\alpha_2 = \beta_0 = 0$.
	
	For $N=3$, we consider the reduced unitarized Burau representation with 
	$\cF=\C^2$ and $w=1$:
	$$
		\rho(\sigma_1) = \frac{1}{2} \begin{bmatrix} 1 & -\sqrt{3} \\ -\sqrt{3} & -1 \end{bmatrix},
		\qquad
		\rho(\sigma_2) = \frac{1}{2} \begin{bmatrix} 1 & \sqrt{3} \\ \sqrt{3} & -1 \end{bmatrix},
		\qquad
		\rho(\sigma_1\sigma_2\sigma_1) = \begin{bmatrix} -1 & 0 \\ 0 & 1 \end{bmatrix}.
	$$
	Note that $\rho(\sigma_j)^2 = \1$ and thus it descends to a non-abelian representation of 
	$S_3 = \{1,\sigma_1,\sigma_2,\sigma_1\sigma_2,\sigma_2\sigma_1,\sigma_1\sigma_2\sigma_1\}$.
	The eigenvectors corresponding to eigenvalue $1$ are
	$$
		v_1 := \frac{1}{2} \begin{bmatrix} \sqrt{3} \\ -1 \end{bmatrix},
		\quad
		v_2 := \frac{1}{2} \begin{bmatrix} \sqrt{3} \\ 1 \end{bmatrix},
		\quad \text{resp.} \quad
		v_2-v_1 = \begin{bmatrix} 0 \\ 1 \end{bmatrix},
	$$
	and therefore we cannot simply take a constant eigenvector as in the abelian case.
	Instead, take two disjoint domains $\Omega_1$ and $\Omega_2$ in $\R^2$ and define
	$\Psi(\tX) = \Psi(\hat\pr(\tX))$ by
	\begin{align*}
		\Psi(\bx_1,\bx_2,\bx_3)
		&:= \sum_{\sigma \in S_3} f(\bx_{\sigma(1)},\bx_{\sigma(2)}) g(\bx_{\sigma(3)}) \rho(\sigma) v_1 \\
		&= 2f(\bx_1,\bx_2)g(\bx_3)v_1 + 2f(\bx_2,\bx_3)g(\bx_1)(-v_2) + 2f(\bx_3,\bx_1)g(\bx_2)(v_2-v_1),
	\end{align*}
	where $f \in H^1_\sym(\R^4;\C)$ supported in $\Omega_1^2$ 
	and $g \in H^1(\R^2;\C)$ supported in $\Omega_2$.
	We also used that $\rho(\sigma_1)v_2 = v_1-v_2$ and $\rho(\sigma_2)v_1 = v_2-v_1$.
	It then follows that the three terms above are pairwise orthogonal 
	in $L^2_\rho$ and
	$$
		\int_{\cC^3} |\Psi|^2 = 2 \int_{\Omega_1^2} |f(\bx_1,\bx_2)|^2 d\bx_1 d\bx_2 \int_{\Omega_2} |g(\bx_3)|^2 d\bx_3.
	$$
	Furthermore, 
	\begin{align*}
		\int_{\cC^3} |\nabla\Psi|^2 =& 2 \int_{\Omega_1^2} \left(
			|\nabla_1 f(\bx_1,\bx_2)|^2 + |\nabla_2 f(\bx_1,\bx_2)|^2 \right) d\bx_1 d\bx_2 
			\int_{\Omega_2} |g(\bx_3)|^2 d\bx_3 \\
		&+ 2\int_{\Omega_1^2} |f(\bx_1,\bx_2)|^2
			\int_{\Omega_2} |\nabla g(\bx_3)|^2 d\bx_3.
	\end{align*}
	and $\Psi$ is $\rho$-equivariant by definition,
	$$
		\Psi(\tX.\sigma_j^{-1}) = \sum_{\sigma \in S_3} f(\bx_{\sigma_j\sigma(1)},\bx_{\sigma_j\sigma(2)}) g(\bx_{\sigma_j\sigma(3)}) \rho(\sigma_j^2\sigma) v_1
		= \rho(\sigma_j)\Psi(\tX).
	$$
	Again, by fixing $g$ and taking $f \in C^\infty_c(\Omega_1 \setminus \bDelta^2)$ 
	to approximate a product state
	$u(\bx_1)u(\bx_2)$ where, e.g., $u$ is the g.s. of the Dirichlet laplacian
	on $\Omega_1$, we find that $T_\rho[\Psi]$ stays uniformly bounded while
	$$
		\int_{\R^6} \frac{|\Psi|^2}{|\bx_1-\bx_2|^2} 
		\ge \int_{\Omega_1^2} \frac{|f(\bx_1,\bx_2)|^2}{|\bx_1-\bx_2|^2} d\bx_1 d\bx_2 \int_{\Omega_2} |g|^2
	$$
	grows unboundedly.
	This contradicts the validity of the Hardy inequality \eqref{eq:Hardy-simple}
	for any positive value of the constant $C_N$.
\end{proof}

\begin{proposition}[Counterexamples for $\alpha_2 = 0$, local case]\label{prop:Hardy-counterex-local}
	In the case that $\rho\colon B_N \to \sU(D)$ is abelian with
	$\rho(\sigma_j^{-1}) \sim \diag(e^{i\gamma_n\pi})_{n=1}^D$ we have
	$$
		E_N^{\rho}(\Omega) \le \min_{n \in \{1,\ldots,D\}} E_N^{\rho^{\gamma_n}}(\Omega),
	$$
	with
	\begin{equation}\label{eq:abelian-bound}
		E_N^{\rho^\alpha}([0,1]^2) \le 2\pi N(N-1) \alpha \frac{1 + \frac{20}{3}\pi(N-2)\alpha}{(1-2\pi\alpha N)^2},
		\quad \text{if}\ 0 \le \alpha < (2\pi N)^{-1},
	\end{equation}
	and thus for $\Omega=[0,1]^2$ (and in fact for any bounded convex $\Omega$)
	the Hardy inequality of Theorem~\ref{thm:Hardy} necessarily trivializes 
	for $\alpha_2=\beta_0=0$.
\end{proposition}
\begin{proof}
	Use the abelian magnetic representation \eqref{eq:abelian-magnetic} to estimate
	for any $1 \le n \le D$
	$$
		E_N^\rho(\Omega) \le \inf \left\{ T_{\gamma_n}^\Omega[\Phi_n] : \Phi_n \in C^\infty_{c,\sym}(\R^{2N} \setminus \bDelta^N), \int_{\cC^N(\Omega)} |\Phi_n|^2 = 1 \right\} 
		= E_N^{\rho^{\gamma_n}}(\Omega).
	$$
	The bound \eqref{eq:abelian-bound} for the abelian energy on 
	$\Omega=[0,1]^2$ was proved in \cite[Lemma 3.2]{LunSei-17}
	using a Dyson-type ansatz \cite{Dyson-57}.
	In fact, if $\gamma_n=0$ and $\Omega \subseteq \R^2$ bounded
	we may take a sequence of $\Phi_\eps \in C^\infty_{c,\sym}(\R^{2N} \setminus \bDelta^N)$
	approximating the constant function $\Phi=(|\Omega|^N/N!)^{-1/2}$, such that
	$$
		\int_{\cC^N(\Omega)} |\Phi_\eps|^2 = 1, \qquad
		T_0^\Omega[\Phi_\eps] = \frac{1}{N!} \int_{\Omega^N} |\nabla\Phi_\eps|^2 \to 0
        \ \text{as} \ \eps \to 0,
	$$
	yielding $E_N^{\rho}(\Omega) = 0$.
\end{proof}

\begin{remark}\label{rem:Hardy-dom}
    In the Hardy inequality \eqref{eq:Hardy-simple},
    the constant
    $C_N > 0$ for $\alpha_2 \neq 0$, which has important consequences for the form domain $H^1_\rho$ due to the singular nature of the potential, requiring $|\Psi|$ to vanish on $\bDelta^N$.
    In fact, after statistics transmutation to bosons or fermions, any abelian anyon model with $\alpha_2 \neq 0$ has the same form domain as fermions;
    see \cite[Appendix A.1]{AtaGirLun-25}.
    Similar observations were made in \cite{RouYan-24} with slightly different methods.
\end{remark}

\begin{remark}\label{rem:Hardy-opt}
	The question concerning the 
	behavior of the optimal constant in
	\eqref{eq:Hardy-simple} 
	for $N \to \infty$ and $\alpha_* = 0$
	is an interesting but difficult open problem
	(in fact open even for fermions with $\alpha_*=1$ \cite{FraHofLapSol-24}).
	It is a consequence of the uncertainty principle that in the ground state 
	we must also consider exchanges $U_p$ 
	with $p>0$ enclosed particles 
	(this was ignored in some earlier scenarios 
	\cite{FroMar-88}, \cite[p.212-213]{FroMar-89}
	resp.\ \cite{GolMaj-04} for anyonic exclusion).
	One may note that Theorem~\ref{thm:Hardy} seems to premiere clustering states where some
	of the annuli $A_p$ can have a higher weight if $\beta_p < \beta_0$.
	This has been discussed to some extent in \cite{LunSol-13b,Lundholm-16,LarLun-16,Lundholm-17},
	and it would be even more relevant for point-attractive anyons 
	(cf.\ Remark~\ref{rmk:extensions}).
	Actually only the latter form of the Hardy inequality relies on the regularity 
	assumptions on $\Psi$ at $\bDelta^N$ implied by the Friedrichs extension.
	See also the remarks following Lemma~\ref{lem:E-from-Hardy} below.
\end{remark}

\subsection{Scale-covariant energy bounds}\label{sec:repulsion-scale}
	
	The goal of this subsection is to use the positivity of the anyonic energy
	due to the repulsion from the Hardy inequality for just a few particles to derive 
	positivity and in fact a quadratic growth for large numbers of particles.
	We are here guided by the scale-covariant method introduced in \cite{LunSei-17}
	and formulated quite generally in \cite{LarLunNam-19}
	(see also \cite[Sec.\,5.5]{Lundholm-17}):

\begin{lemma}[{\keyword{Covariant energy bound}; \cite[Lemma~4.1]{LarLunNam-19}}]
	\label{lem:covariant-energy} 
	Assume that to any $n\in \mathbb{N}_0$ and any cube $Q\subset \R^d$ there is associated a non-negative number (`energy') $e_n(Q)$ satisfying the following properties, for some constant $s>0$: 
	\begin{itemize}
		\item (scale-covariance)
			$e_n(\lambda Q) = \lambda^{-2s} e_n(Q)$ for all $\lambda > 0$;
		\item (translation-invariance)
			$e_n(Q+\bx) = e_n(Q)$ for all $\bx \in \R^d$;
		\item (superadditivity)
			For any collection of disjoint cubes $\{Q_j\}_{j=1}^J$ such that their union is a cube, 
			$$e_n\Bigl(\bigcup_{j=1}^J Q_j\Bigr) \ge \min_{\{n_j\} \in \mathbb{N}_0^J \, s.t.\, \sum_j n_j = n}\ \sum_{j=1}^J e_{n_j}(Q_j) ;$$
		\item (a priori positivity) There exists $q \ge 0$ such that $e_n(Q) > 0$ for all $n \ge q$.
	\end{itemize}
	There then exists a constant $C>0$ independent of $n$ and $Q$ such that
	\begin{equation}\label{eq:covariant-energy-bound}
		e_n(Q) \ge C |Q|^{-2s/d} n^{1+2s/d}, \quad \forall n\ge q.
	\end{equation}
\end{lemma}

	Note that scale-covariance with $s=1$ as well as translation-invariance
	follows directly from the definition of the anyonic energy $E_N(\Omega)$:

\begin{lemma}[Covariance]\label{lem:covariance}
	Given $\Omega \subseteq \R^2$ open and simply connected, 
	$\bx \in \R^2$ and $\lambda>0$, we have
	\begin{equation}\label{eq:energy-covariance}
		E_N(\lambda\Omega + \bx) = \lambda^{-2} E_N(\Omega).
	\end{equation}
\end{lemma}
\begin{proof}
	Denote $\Omega_\lambda := \lambda\Omega + \bx$.
	For $X = \{\bx_1,\ldots,\bx_N\} \in \cC^N(\Omega)$ we have 
	$X_\lambda := \lambda X + \bx \in \cC^N(\Omega_\lambda)$,
	and for $\tX \in \tcC^N(\Omega)$ an equivalence class of paths from 
	$X_0(\Omega)$ to $X=\tilde\pr(\tX)$ we have
	$\tX_\lambda := \lambda \tX + \bx \in \tcC^N(\Omega_\lambda)$
	a corresponding translated and scaled equivalence class of paths from 
	$X_0(\Omega_\lambda)$ to $X_\lambda = \tilde\pr(\tX_\lambda)$.
	Given $\Psi_\lambda \in H^1_\rho(\Omega_\lambda)$ we define 
	$\Psi(\tX) := \lambda^{N}\Psi_\lambda(\tX_\lambda)$. 
	Then, using $dX_\lambda = \lambda^{2N}dX$, we have an $L^2_\rho$-isomorphism
	$$
		\int_{\cC^N(\Omega)} |\Psi(\tX)|_\cF^2 \,dX
		= \int_{\cC^N(\Omega_\lambda)} |\Psi_\lambda(\tX_\lambda)|_\cF^2 \,dX_\lambda,
	$$
	and, with the corresponding rescaling in \eqref{eq:equivariant-derivative},
	$$
		\int_{\cC^N(\Omega)} |\nabla\Psi(\tX)|_{\cF^{2N}}^2 \,dX
		= \lambda^2 \int_{\cC^N(\Omega_\lambda)} |\nabla\Psi_\lambda(\tX_\lambda)|_{\cF^{2N}}^2 \,dX_\lambda,
	$$
	and \emph{vice versa} by inversion 
	$(\lambda,\bx) \mapsto (\lambda^{-1},-\bx)$.
	Applied to a sequence of minimizers yields \eqref{eq:energy-covariance}.
\end{proof}

\begin{lemma}[Superadditivity]\label{lem:superadditivity}
	For $K \ge 2$, let $\{\Omega_k\}_{k=1}^K$ be a collection of disjoint,
	convex subsets of $\R^2$. For any 
	$\vec{n} \in \N_0^K$ with $\sum_k n_k = N$,
	let $\1_{\vec{n}}$ denote the characteristic function of the subset of
	$\cC^N(\R^2)$ where exactly $n_k$ of the points $X=\{\bx_1,\ldots,\bx_N\}$
	are in $\Omega_k$, for all $1 \le k \le K$. Let
	$$
		W(X) := \sum_{\vec{n}} \sum_{k=1}^K E_{n_k}(\Omega_k) \1_{\vec{n}}(X),
	$$
	and assume $\Omega := \overline{\cup_k \Omega_k}^\circ$ is also convex. We then have
	\begin{equation}\label{eq:superadd-W}
		\int_{\cC^N(\Omega)} |\nabla\Psi|^2 \ge \int_{\cC^N(\Omega)} W|\Psi|^2
	\end{equation}
	for any $\Psi \in H^1_\rho(\tcC^N(\Omega))$, and in particular
	\begin{equation}\label{eq:superadd-energy}
		E_N(\Omega) \ge \min_{\vec{n}} \sum_{k=1}^K E_{n_k}(\Omega_k).
	\end{equation}
\end{lemma}
\begin{proof}
	Using that $\1 = \sum_{\vec{n}} \1_{\vec{n}}$ on $\cC^N(\Omega)$, 
	we have for any $\Psi \in C^\infty_{\rho,c}(\tcC^N)$
	\begin{equation}\label{eq:superadd-start}
		\int_{\cC^N(\Omega)} |\nabla\Psi|^2 \,dX 
		= \sum_{\vec{n}} \int_{\cC^N(\Omega)} |\nabla\Psi|^2 \1_{\vec{n}} \,dX,
	\end{equation}
	where for any $\tX = [\gamma].X_0(\Omega) \in \tcC^N(\Omega)$
	$$
		|\nabla\Psi(\tX)|^2 \1_{\vec{n}} 
		= \sum_{k=1}^K \sum_{j : \bx_j \in \Omega_k} |\nabla_j\Psi(\tX)|^2 \1_{\vec{n}}.
	$$
	Hence for each $\vec{n}$, the integral to be considered is
	\begin{equation}\label{eq:superadd-n}
		\sum_{k=1}^K \int_{\cC^{N-n_k}(\Omega \setminus \Omega_k)} \int_{\cC^{n_k}(\Omega_k)} \sum_{j : \bx_j \in \Omega_k} |\nabla_j\Psi(\tX)|^2 \,dX_k \1_{\vec{n}} dX_k^c.
	\end{equation}
	Fixing the set of points $X_k^c \in \cC^{N-n_k}(\Omega \setminus \Omega_k)$ and considering 
	the set $X_k \in \cC^{n_k}(\Omega_k)$ s.t. $X = X_k \cup X_k^c$, 
	we have by $\tX = [\gamma].\tX_0(\Omega)$ a path $\gamma$ in $\Omega$ where $n_k$ 
	of the points at $X_0^N(\Omega)$
	move into $\Omega_k$ and $N-n_k$ into $\Omega \setminus \Omega_k$.
	Any such path can equivalently be taken via first acting on $\tX_0^N(\Omega)$ with a braid $b \in B_N$,
	then selecting a partition of particles at $X_0^N(\Omega)$, say $\{1,\ldots,n_k\},\{n_k+1,\ldots,N\}$,
	then moving the former set of points via $X_0^{n_k}(\Omega_k)$ 
	to $X_k$ and the latter via $X_0^{N-n_k}(\Omega_{k'})$ to $X_k^c$ for some $k'\neq k$.
	Note that any additional encircling of the points in $X_k$ around the fixed set $X_k^c$ 
	in this process may be taken care of by writing the corresponding action of
	$[\gamma_k] \in \pi_1(\cC^{n_k}(\Omega \setminus X_k^c),X_k)$ on $\tX_k \mapsto X_k$ 
	as $\tX_0^N(\Omega).b$ for some $b \in B_N$
	(actually in the subgroup of the pure braid group which keeps the indices fixed).
	Furthermore, any $\tX_k \in \tcC^{n_k}(\Omega_k)$ can be represented as 
	$\tX_k = [\gamma_k].\tX_0^{n_k}(\Omega_k)$
	for a path $\gamma_k$ of $n_k$ points in $\Omega_k$ and using
	$\tX_0^{n_k}(\Omega_k).\beta$ if it is a loop $\beta \in B_{n_k} \hookrightarrow B_N$.
	In this way we construct a surjective map
	$$
		\tcC^{n_k}(\Omega_k) \times \tcC^{N-n_k}(\Omega \setminus \Omega_K) \times B_N
		\ni (\tX_k,\tX_k^c,b) 
		\mapsto \tX \in \tilde\pr^{-1}(\supp \1_n) \subseteq \tcC^N(\Omega).
	$$
	
	Hence, if we define for a fixed $X_0^{n_k}(\Omega_k)$ and $X_k^c$ 
	with fixed lift to $\tX \in \tcC^N(\Omega)$ the function
	$\Psi_k(\tX_k) := \Psi(\tX)$
	we have
	$$
		\Psi_k(\tX_k.\beta) = \Psi(\tX.\beta) = \rho(\beta^{-1})\Psi_k(\tX_k).
	$$
	Note that it is also defined in a neighborhood of $\cC^{n_k}(\Omega_k)$.
	By making a smooth cut-off on $\Omega_k^c$ away from the points $X_k^c$
	we may consider this a smooth function on 
	$\tcC^{n_k}(\R^2)$ which is $\rho$-equivariant and with compact support 
	$\tilde\pr(\supp\Psi_k)$. 
	It is therefore in $H^1_\rho(\tcC^{n_k}(\Omega_k);\cF)$.
	Further,
	$$
		\sum_{j=1}^{n_k} |\nabla_j \Psi_k(\tX_k)|^2 
		= \sum_{j : \bx_j \in \Omega_k} |\nabla_j \Psi(\tX)|^2,
	$$
	for $\tX \in \supp \1_{\vec{n}}$.
	Then
	$$
		\int_{\cC^{n_k}(\Omega_k)} \sum_{j=1}^{n_k} |\nabla_j\Psi_k(\tX_k)|^2 \,dX_k
		\ge E_{n_k}(\Omega_k) \int_{\cC^{n_k}(\Omega_k)} |\Psi_k|^2 \,dX_k
	$$
	and after putting this back into the original integral 
	\eqref{eq:superadd-n} and \eqref{eq:superadd-start},
	and using that $|\Psi_k|^2 = |\Psi|^2$,
	we obtain the bound \eqref{eq:superadd-W} 
	in terms of the scalar and $B_N$-symmetric potential $W$.
\end{proof}

	In the sequel we denote simply by $E_N := E_N([0,1]^2)$ the $N$-anyon
	energy on the unit square with a given anyon model.
	Our starting point for positivitity is the following bound originally derived for
	abelian anyons in \cite{LarLun-16}:

\begin{lemma}[{Bound for $E_N$ directly from Hardy; \cite[Lemma~5.3]{LarLun-16}}]\label{lem:E-from-Hardy}
	For $\nu > 0$ let $j_\nu'$ denote the first positive zero of the
	derivative of the Bessel function $J_\nu$, satisfying \eqref{eq:Bessel-bounds}.
	There exists a function $f\colon [0,(j_1')^2] \to \R_+$ satisfying
	$$
		t/6 \le f(t) \le 2\pi t
		\qquad \text{and} \qquad
		f(t) = 2\pi t \bigl( 1-O(t^{1/3}) \bigr) \ \text{as} \ t \to 0,
	$$
	such that
	\begin{equation}\label{eq:E-from-Hardy}
		E_N \ge f\bigl( (j_{\alpha_N}')^2 \bigr) (N-1)_+.
	\end{equation}
\end{lemma}
\begin{proof}
	After splitting the energy $T_\rho = \kappa T_\rho + (1-\kappa) T_\rho$ and
	applying the diamagnetic inequality Lemma~\ref{lem:diamag} to the first part,
	and the Hardy inequality of Theorem~\ref{thm:Hardy} in the later form to the second part, 
	the proof is identical to
	that of \cite[Lemma~5.3]{LarLun-16}, where one expresses
	the bosonic $N$-particle energy 
	as $(N-1)/2$ copies of the two-particle energy, and defines
	(where subsequently $t=j_{\alpha_N}'^2$)
	$$
		f(t) := \frac{1}{2} \sup_{\kappa \in (0,1)} \inf_{\int_{Q_0^2} |\psi|^2 = 1} 
			\int_{Q_0^2} \left( \kappa|\nabla_1|\psi||^2 + \kappa|\nabla_2|\psi||^2 
				+ (1-\kappa)t\frac{\1_{\Omega\circ\Omega}(\bx_1,\bx_2)}{\delta(\bX_{12})^2} |\psi|^2 \right) d\bx_1 d\bx_2.
	$$
	The energy of the corresponding Schr\"odinger operator (with singular potential)
	on $L^2(Q_0^2)$ is then estimated in a standard way using projection on the constant function,
	yielding the bounds for $f(t)$ stated in the lemma 
	(see also Figure~\ref{fig:popcorn}).
\end{proof}

\begin{remark}\label{rmk:Hardy-improvement}
	We could replace the above bound by one which takes more of the particle distribution
	into account. Namely, after making a cut-off around the diagonals 
	on the otherwise singular Hardy term we can take the average
	(i.e.\ expectation in the bosonic g.s.)
	$$
		\sum_{j<k} \int_{Q_0^N} \sum_{p=0}^M \beta_p^2 \,\1_{A_p \cap B_\eps(\0)^c}(\br_{jk}) \frac{1}{r_{jk}^2} d\sx,
	$$
	which is finite and greater than the one used in obtaining the above bounds for $f$.
	However, we will be using the lemma mainly for $N=2$, where nothing is gained by this,
	and use a different route to derive a bound for $E_N$ 
	depending on $E_2$ and thus $\alpha_2=\beta_0$ only.
\end{remark}
\begin{remark}\label{rmk:Hardy-optimal}
	In fact, it was pointed out in \cite[Proposition~4.6]{LunSei-17} that
	Lemma~\ref{lem:E-from-Hardy}
	gives the correct behavior 
	to leading order for $N=2$ and $\alpha_2 \to 0$ 
	as it is matched by the upper bound \eqref{eq:abelian-bound} of Proposition~\ref{prop:Hardy-counterex-local} 
	(for $N=2$ it is effectively abelian)
	to yield
	\begin{equation}\label{eq:E2-asymp}
		E_2 = 4\pi\alpha_2 \bigl(1 + O(\alpha_2^{1/3})\bigr) 
		\qquad \text{as}\ \alpha_2 \to 0.
	\end{equation}
	In other words, the energy per particle is asymptotically $2\pi\alpha_2$.
\end{remark}

	Given that $\alpha_N \le \alpha_2$ for all $N \ge 2$, 
	the bound \eqref{eq:E-from-Hardy} will be trivial if $\alpha_2 = 0$,
	i.e. if $E_2 = 0$.
	On the other hand, if $\alpha_2 > 0$ and thus $E_2>0$
	then it is in fact sufficient to yield
	positivity $E_N>0$ for all $N \ge 2$:

\begin{lemma}[{A priori bounds in terms of $E_2$; \cite[Lemma~4.3]{LunSei-17}}]\label{lem:E-apriori}
	For any $N \ge 3$ we have
	$$
		E_N \ge \frac{\pi^2 \binom{N}{2} \left(\frac{3}{4}\right)^{N-2} E_2}{\left(\pi + 4\sqrt{E_2}\right)^2 + \binom{N}{2} \left(\frac{3}{4}\right)^{N-2} E_2},
	$$
	and thus $E_N > 0$ if $E_2 > 0$.
\end{lemma}

	The proof is exactly analogous to the abelian case,
	\cite[Proposition~2]{Lundholm-13} and \cite[Lemma~4.3]{LunSei-17}, 
	where we split the unit square $Q_0$ in $K=4$ smaller squares and use 
	covariance \eqref{eq:energy-covariance} and superadditivity \eqref{eq:superadd-W},
	with
	$$
		W \ge W_2 := 4E_2 \sum_{k=1}^4 \sum_{\vec{n} :  n_k=2} \1_{\vec{n}},
	$$
	(for each $k$, we keep only the terms in $W$ for which $n_k=2$), and
	$$ 
		\qquad \int_{Q_0^N} W_2 = E_2 \binom{N}{2} \left(\frac{3}{4}\right)^{N-2}.
	$$
	Then we use the diamagnetic inequality of Lemma~\ref{lem:diamag} 
	to estimate the energy $E_N$ in terms of a standard 
	bosonic or distinguishable particle problem, yielding the claimed bound
	(see \cite[Lemma~4.3]{LunSei-17} for details).

\begin{lemma}[{\keyword{Covariant energy bound}}]
	Given a sequence of anyon models $\rho_N\colon B_N \to \sU(\cF_N)$
	with two-anyon exchange parameters $\alpha_2(N) \ge \alpha_2 > 0$ 
	(it could depend on the
	total number $N$ of particles but must then be uniformly bounded from below), 
	the local $N$-anyon energy $E_N(Q)$ on squares $Q$
	satisfies the criteria of Lemma~\ref{lem:covariant-energy}
	with $s=1$ and $q=1$
	and therefore there exists a constant $C=C(\alpha_2)>0$ such that
	$$
		E_N \ge C N^2, \quad N \ge 2.
	$$
\end{lemma}

	For an explicit bound on this constant we will use the original method of
	\cite[Lemma~4.8-4.9]{LunSei-17} and the local exclusion principle.

\subsection{Local exclusion principle}\label{sec:repulsion-exclusion}

	Given $\Psi \in H^1_\rho(\tcC^N(\Omega))$, recall that it defines a corresponding
	one-body density $\varrho_\Psi \in L^1(\Omega)$ 
	and a kinetic energy density $T_{\rho,\Psi} \in L^1(\Omega)$
	(Definitions~\ref{def:density} and \ref{def:energy-density}).
	
\begin{lemma}[{\keyword{Local exclusion principle for non-abelian anyons}}]\label{lem:local-exclusion}
	Given a sequence of anyon models $\rho_N\colon B_N \to \sU(\cF_N)$
	with exchange parameters 
	$\alpha_n = \min_p \beta_p$ (they may all depend on $N$),
	the local $n$-anyon energy $E_n = E_n^{\rho_N}(Q_0)$ on the unit square satisfies
	\begin{equation}\label{eq:local-exclusion-linbound}
		E_n \ge C_4(\rho_N) n, \qquad 4 \le n \le N,
	\end{equation}
	where
	$$
		C_4(\rho_N) := \frac{1}{4} \min\{E_2,E_3,E_4\}
		\ \ge \ 
		\frac{1}{4} \min\{ E_2,0.147 \} \ \ge \ c(\alpha_2)
		= \frac{1}{4}\min\left\{ f(j_{\alpha_2}'^2),0.147 \right\}.
	$$
	Furthermore, for any simply connected open domain $\Omega \subseteq \R^2$ 
	and any square $Q \subseteq \Omega$, any $N \ge 1$
	and $L^2$-normalized $\Psi \in H^1_{\rho_N}(\tcC^N(\Omega);\cF_N)$ 
	with one-particle density $\varrho_\Psi$ on $\Omega$, 
	we have
	\begin{equation}\label{eq:local-exclusion}
		T_{\rho_N}^{Q \subseteq \Omega}[\Psi] \ \ge \ 
		\frac{C(\rho_N)}{|Q|} 
		\left( \int_Q \varrho_\Psi(\bx) \,d\bx \ - 1 \right)_+,
	\end{equation}
	where
	$$
		C(\rho_N) := \max \left\{ C_4(\rho_N), f( j_{\alpha_N}'^2 ) \right\}
		\ge \max \left\{ \frac{1}{4}\min\left\{f( j_{\alpha_2}'^2 ),0.147\right\}, f( j_{\alpha_N}'^2 ) \right\}.
	$$
\end{lemma}
\begin{proof}
	By splitting the unit square $Q_0$ into four equally large squares $Q_k$,
	$|Q_k| = |Q_0|/4$, one may derive from scale-covariance and superadditivity
	the recursive bound
	$$
		E_n \ge 4\min_{m \in \{n/4,n/4+1,\ldots,n\}} E_m.
	$$
	From this and a priori positivity follows the linear bound \eqref{eq:local-exclusion-linbound} 
	(see \cite[Lemma~4.8]{LunSei-17} for details).
	The bound on $C_4(\rho_N)$ in terms of $E_2 \ge f(j_{\alpha_2}'^2)$ 
	follows from Lemma~\ref{lem:E-apriori},
	where the positive root of $(\pi+4\sqrt{x})^2 + \frac{9}{4}x = \frac{9}{4}\pi^2$
	is $x \ge 0.147$. Furthermore one may bound 
	$c(\alpha_2) \ge \frac{1}{4}\min\{ \alpha_2/3, 0.147 \}$ 
	using \eqref{eq:Bessel-bounds} and
	the bounds on $f$ in Lemma~\ref{lem:E-from-Hardy}.
	
	For \eqref{eq:local-exclusion} we use the decomposition of the energy into 
	particles in $Q$ and in $Q^c = \Omega \setminus Q$,
	\begin{align*}
		T_{\rho_N}^{Q \subseteq \Omega}[\Psi] 
		&= \int_{\cC^N(\Omega)} \sum_{j=1}^N |\nabla_j\Psi(\tX)|^2 \1_{\{\bx_j \in Q\}} \,dX \\
		&= \sum_{\vec{n}} \int_{\cC^{N-n_1}(Q^c)} \int_{\cC^{n_1}(Q)} \sum_{j:x_j \in Q} |\nabla_j\Psi(\tX)|^2 dX_1 \1_{\vec{n}} \,dX_1^c
	\end{align*}
	and proceed similarly to the proof of superadditivity (Lemma~\ref{lem:superadditivity}), 
	with $\Omega_1 = Q$ and $\Omega_{k>1} \subseteq Q^c$ 
	(e.g. splitting $Q^c$ along the coordinate axes).
	Thus,
	$$
		T_{\rho_N}^{Q \subseteq \Omega}[\Psi]  
		\ge \sum_{\vec{n}} \int_{\cC^{N-n_1}(Q^c)} E_{n_1}(Q) \int_{\cC^{n_1}(Q)} |\Psi(\tX)|^2 dX_1 \1_{\vec{n}} \,dX_1^c.
	$$
	Now use that from scale-covariance and \eqref{eq:local-exclusion-linbound}
	$$
		E_n(Q) \ge |Q|^{-1} C_4(\rho_N) (n-1)_+ \quad \text{for all} \ 0 \le n \le N,
	$$
	while from our earlier Lemma~\ref{lem:E-from-Hardy}
	$$
		E_n(Q) \ge |Q|^{-1} f( j_{\alpha_n}'^2 ) (n-1)_+ 
		\ge |Q|^{-1} f( j_{\alpha_N}'^2 ) (n-1)_+ 
		\quad \text{for all} \ 0 \le n \le N.
	$$
	Defining the induced $n$-particle probability distribution on $Q$, $n \in \{0,1,\ldots,N\}$,
	\begin{equation}\label{eq:p_n}
		p_n(\Psi;Q) := \sum_{\vec{n} : n_1=n} \int_{\cC^{N-n_1}(Q^c)} \int_{\cC^{n_1}(Q)} |\Psi(\tX)|^2 dX_1 \1_{\vec{n}} \,dX_1^c,
	\end{equation}
	with $\sum_{n=0}^N p_n = 1$ and $\sum_{n=0}^N np_n = \int_Q \varrho_\Psi$,
	we may use the convexity of $x \mapsto (x-1)_+$ to obtain
	\begin{align*}
		T_{\rho_N}^{Q \subseteq Q_0}[\Psi] 
		&\ge \sum_{n=0}^N |Q|^{-1} C(\rho_N) (n-1)_+ p_n 
		\ge |Q|^{-1} C(\rho_N) \sum_n (n p_n - 1)_+,
	\end{align*}
	which proves \eqref{eq:local-exclusion}.
\end{proof}


\section{The ideal non-abelian anyon gas}\label{sec:gas}

\subsection{The homogeneous gas} 

	The (zero-temperature) \keyword{homogeneous anyon gas} on a domain $\Omega$ 
	is defined by taking a ground state $\Psi$
	of the kinetic energy $T_\rho^\Omega$, either with Dirichlet or Neumann boundary conditions.
	As $N \to \infty$ and $|\Omega| \to \infty$ while keeping the mean density $\bar\varrho := N/|\Omega|$ fixed
	(\keyword{the thermodynamic limit})
	one expects on general grounds the one-body density $\varrho_\Psi$ 
	to tend to the constant $\bar\varrho$, 
	regardless of the shape of (reasonable) $\Omega$ 
	and the choice of b.c.
	(at least on average over large enough scales; 
	cf. the almost-bosonic limit \cite{CorDubLunRou-19}).

\begin{definition}[Dirichlet energy]\label{def:Dirichlet}
	Given an $N$-anyon model $\rho\colon B_N \to \sU(\cF)$
	and a simply connected domain $\Omega \subseteq \R^2$,
	we define the \keyword{Dirichlet Sobolev space} $H^1_{\rho,0}(\tcC^N(\Omega);\cF)$,
	of wave functions $\Psi \in L^2_\rho$ vanishing on the boundary $\partial\overline{\Omega}$
	and with finite expected local kinetic energy on $\Omega$,
	as the closure of the space of functions $\Psi \in C^\infty_{\rho,c}(\tcC^N(\Omega))$
	(with $\tilde\pr(\supp\Psi)$ compactly supported in $\cC^N(\Omega)$) 
	in the $H^1_\rho$-norm.
	Taking the infimum of $T_\rho^\Omega[\Psi]$ of such functions defines the 
	\keyword{Dirichlet ground-state energy} on $\Omega$
	$$
		\bar{E}_N(\Omega) := \inf \left\{ \int_{\cC^N(\Omega)} |\nabla\Psi(\tX)|_{\cF^{2N}}^2 \,dX : 
		\Psi \in H^1_{\rho,0}(\tcC^N(\Omega)), \int_{\cC^N(\Omega)} |\Psi|_\cF^2 \,dX = 1 \right\},
		\ \ N \ge 2,
	$$
	also $\bar{E}_0(\Omega) := 0$, and 
	$
		\bar{E}_1(\Omega) := \inf \left\{\int_\Omega |\nabla\Psi|^2 : 
		\Psi \in H^1_0(\Omega;\cF), \int_\Omega |\Psi|^2 = 1 \right\}.
	$
\end{definition}
	
	Since $C^\infty_{\rho,c}(\tcC^N(\Omega)) \subseteq C^\infty_{\rho,c}(\tcC^N(\R^2))$
	we trivially have
	$$
		E_N(\Omega) \le \bar{E}_N(\Omega),
	$$
	where, to distinguish the two,
	we refer to $E_N(\Omega)$ as the \keyword{Neumann energy} on $\Omega$.
	We also note that 
	the one-body energy $E_1$ with fiber $\cF \cong \C^D$ 
	is the same as $E_1$ with $\cF = \C$ or $\cF = \R_+$.
	This follows by decomposing
	$|\nabla\Psi|^2 = \sum_{n=1}^D |\nabla\Psi_n|^2$ w.r.t.\ a basis in $\cF$,
	i.e. $T_0|_{H^1(\Omega;\cF)} = \bigoplus^D T_0|_{H^1(\Omega;\C)}$,
	and using the standard diamagnetic inequality, $|\nabla\Psi| \ge |\nabla|\Psi||$.

\begin{lemma}[Subadditivity]\label{lem:subadd}
	Let $\Omega_k$, $k=1,\ldots,N$ be pairwise disjoint 
	and simply connected subsets of $\R^2$,
	and assume $\Omega := \overline{\cup_k \Omega_k}^\circ$ is also simply connected.
	If $\rho\colon B_N \to \sU(\cF)$ is an arbitrary geometric $N$-anyon model then
	$$
		\bar{E}_N(\Omega) \le \sum_{k=1}^N \bar{E}_1(\Omega_k).
	$$
\end{lemma}
\begin{proof}
	Let $u_k \in C^\infty_c(\Omega_k;\C)$ and $v \in \cF$, $|v|=1$, and 
	consider the subset $\Omega = \prod_{k=1}^N \Omega_k \subseteq \R^{2N}$ where 
	$\bx_j$ is localized on $\Omega_j$ (thus making all particles distinguishable). 
	Fix a representative $\sx_\Omega \in \Omega$.
	This point is in 1-to-1 correspondence to $X_0$ (and $X_0(\Omega)$)
	by means of a simple path $\gamma_\Omega$ in $\cC^N$,
	and any action of $b \in B_N$
	on $X_0$ induces a permutation $\Omega.\hat\pr(b)$ in $\R^{2N}$
	and a lift $\tilde\Omega.b$ in $\tcC^N$.
	Thus, 
	by the disjointness of $\Omega_k$ we can write the set of pre-images in $\tcC^N$
	$$
		\tilde\pr^{-1} \pr(\Omega) = \bigsqcup_{b \in B_N} \tilde\Omega.b.
	$$
	Define for $\tX \in \tilde\Omega$ ($b=1$)
	$$
		\Psi(\tX) := u_1(\bx_1) u_2(\bx_2) \ldots u_N(\bx_N) v
	$$
	and on $\tilde\Omega.b$ for $b \neq 1$
	$$
		\Psi(\tX.b) := \rho(b^{-1})\Psi(\tX).
	$$
	Then we have a $\rho$-equivariant function such that for any
	$\tX \mapsto X = \{\bx_1,\ldots,\bx_N\}$, $\bx_j \in \Omega_j$,
	$$
		|\Psi(\tX)|^2 = |u_1(\bx_1)|^2 \ldots |u_N(\bx_N)|^2, 
	$$
	and
	$$
		|\nabla\Psi(\tX)|^2 = |\nabla u_1(\bx_1)|^2 |u_2(\bx_2)|^2 \ldots |u_N(\bx_N)|^2 + \ldots + |u_1(\bx_1)|^2 \ldots |u_{N-1}(\bx_{N-1})|^2 |\nabla u_N(\bx_N)|^2.
	$$
	Thus, normalizing $\int_{\R^2} |u_k|^2 = 1$ and taking sequences of $u_k$ 
	converging to the respective ground states on $\Omega_k$ realizing $E_1(\Omega_k)$, 
	we obtain the upper bound of the lemma.
\end{proof}

\begin{theorem}[Uniform bounds for the homogeneous anyon gas]\label{thm:homogeneous-gas}
	For any sequence of $N$-anyon models $\rho_N\colon B_N \to \sU(\cF_N)$ 
	with n-anyon exchange parameters $\alpha_n = \alpha_n(N) \in [0,1]$ 
	we have the uniform bounds
	\begin{equation}\label{eq:homogeneous-energy}
		\frac{1}{4} C(\rho_N) N^2 \bigl(1 - O(N^{-1})\bigr)
		\le E_N(Q_0) \le \bar{E}_N(Q_0) 
		\le 2\pi^2 N^2 \bigl(1 + O(N^{-1/2})\bigr),
	\end{equation}
	where
	$$
		C(\rho_N) = \max \left\{ C_4(\rho_N), f(j_{\alpha_N}'^2) \right\}
		\ge \max \left\{ \frac{1}{4}\min\left\{f(j_{\alpha_2}'^2),0.147\right\}, f(j_{\alpha_N}'^2) \right\}.
	$$
\end{theorem}
\begin{remark}
	By scale-covariance, the g.s.\ energy per particle and unit density 
	in the thermodynamic limit $N \to \infty$, $\Omega = [-L/2,L/2]^2$, $L \to \infty$
	at fixed density $\bar\varrho = N/|\Omega|$ is then
	$$
		e(\{\rho_N\}) := \liminf_{N \to \infty} E_N/N^2,
		\quad
		\max \left\{ \frac{1}{4}\min\left\{f( j_{\alpha_2}'^2 ),0.147\right\}, f( j_{\alpha_*}'^2 ) \right\}
		\le e(\{\rho_N\}) \le 2\pi^2,
	$$
	where $\alpha_n$ are the $\liminf$ of $\alpha_n(\rho_N)$ as $N \to \infty$,
    and $\alpha_* = \liminf_{N \to \infty} \alpha_N(\rho_N)$.
\end{remark}
\begin{remark}
	One can improve the upper bound in the case that the anyon model is 
	transmutable and with a small statistics parameter, 
	by treating the magnetic potential as a weak interaction and using 
	techniques for weakly interacting Bose gases \cite{Dyson-57,LieSeiSolYng-05}. 
	This was done for almost-bosonic abelian anyons in
	\cite[Lemma~3.2]{LunSei-17}.
\end{remark}
\begin{proof}
	Given sub- and superadditivity and the local exclusion principle, 
	the method to derive these bounds via splitting into smaller boxes is quite standard;
	we follow \cite[Proposition~3.3 and Lemma~4.9]{LunSei-17}.

	For the lower bound, let $K \in \N$ and split $Q_0$ into $K^2$ smaller squares 
	$Q_k$, $k = 1,\ldots,K^2$, of equal size $|Q_k| = K^{-2}$. 
	By superadditivity Lemma~\ref{lem:superadditivity} and scale-covariance Lemma~\ref{lem:covariance}
	we have for any state $\Psi \in H^1_\rho(\tcC^N(Q_0);\cF_N)$
	\begin{equation}\label{eq:homogeneous-lower-split}
		\int_{\cC^N(Q_0)} |\nabla\Psi|^2 
		\ge \sum_{\vec{n}} \sum_{k=1}^{K^2} K^2 E_{n_k} \int_{\cC^N(Q_0)} \1_{\vec{n}} |\Psi|^2
		= K^2 \sum_{k=1}^{K^2} \sum_{n=0}^N E_n p_n(\Psi;Q_k),
	\end{equation}
	where $p_n$ is the induced $n$-particle probability distribution on $Q_k$,
	defined in \eqref{eq:p_n}.
	Define also the average distribution of particle numbers
	$$
		\gamma_n := K^{-2} \sum_{k=1}^{K^2} p_n(\Psi;Q_k),
	$$
	normalized
	$$
		\sum_{n=0}^N \gamma_n = 1,
		\qquad
		\sum_{n=0}^N n\gamma_n = N/K^2 =: \rho_Q,
	$$
	which is the expected number of particles on any one of the smaller squares.
	Hence the r.h.s. of \eqref{eq:homogeneous-lower-split} is, 
	by Lemma~\ref{lem:local-exclusion} and convexity,
	bounded from below by
	$$
		K^4 \sum_{n=0}^N C(\rho_N) (n-1)_+ \gamma_n
		\ge C(\rho_N) K^4 \left( \sum_{n=0}^N n\gamma_n - 1 \right)_+
		= C(\rho_N) N^2 \rho_Q^{-2} (\rho_Q - 1)_+.
	$$
	Optimizing then $\rho_Q \sim 2$ by our choice of $K$, 
	we take $K := \lceil \sqrt{N/2} \rceil$ and obtain
	$2(1+\sqrt{2/N})^{-2} \le \rho_Q \le 2$, and thus
	$$
		E_N \ge C(\rho_N) N^2 \frac{1}{4} \left( 2(1+\sqrt{2/N})^2 - (1+\sqrt{2/N})^4 \right),
	$$
	which proves the claimed lower bound.
	
	For the upper bound
	let $K$ be the smallest integer such that $N \le K^2$, and again split the
	unit square $Q_0 = [0,1]^2$ into $K^2$ smaller squares $Q_k$ of equal size.
	By subadditivity Lemma~\ref{lem:subadd} and the well-known one-body energy
	$\bar{E}_1(Q_0) = 2\pi^2$ of the Dirichlet ground state 
	$u(x,y) = \sin(\pi x)\sin(\pi y)$, 
	we obtain
	$$
		\bar{E}_N(Q_0) \le \sum_{k=1}^{N} E_1(Q_k) + \sum_{k=N+1}^{K^2} E_0(Q_k)
		= N K^2 2\pi^2 \le N (1+\sqrt{N})^2 2\pi^2
	$$
	by scale-covariance. This proves the claimed upper bound.
\end{proof}
	
	With our results of Corollaries~\ref{cor:Up-fib}-\ref{cor:Up-ising}
	respectively Section~\ref{sec:phases-clifford},
	and bounds for $f$ and $j_\nu'$,
	we have then

\begin{corollary}\label{cor:Fibonacci-gas}
	For Fibonacci anyons the exchange parameters are $\alpha_2 = \beta_0 = 3/5$
	and $\alpha_N = \beta_1 = 1/5$ for $N \ge 3$,
	and hence \eqref{eq:homogeneous-energy} holds with $C(\rho_N) \ge 1/15$ for $N \ge 3$.
\end{corollary}

\begin{corollary}\label{cor:Ising-gas}
	For Ising anyons the exchange parameters are 
	$\alpha_N = \beta_0 = 1/8$ for all $N \ge 2$,
	and hence \eqref{eq:homogeneous-energy} holds with $C(\rho_N) \ge 1/24$ for $N \ge 2$.
\end{corollary}

\begin{corollary}\label{cor:Clifford-gas}
	For Clifford anyons the exchange parameters are $\alpha_N = \beta_0 = 1/4$ for $2 \le N \le 4$
	and $\alpha_N = \beta_3 = 0$ for $N \ge 5$,
	and hence \eqref{eq:homogeneous-energy} holds with $C(\rho_N) \ge 1/48$ for $N \ge 5$.
\end{corollary}

\begin{remark}
	Note that the numerical estimate of $f(j_{\alpha_*}'^2)$ in Figure~\ref{fig:popcorn}
	gives significant improvements to these analytical lower bounds for $C(\rho_N)$:
	$$
		C(\rho_N^\text{Fibonacci}) \gtrsim 0.35, \quad
		C(\rho_N^\text{Ising}) \gtrsim 0.25, \quad
		\text{respectively} \quad 
		C(\rho_N^\text{Clifford}) \gtrsim 0.147/4.
	$$
	Let us make some additional remarks concerning 
	the sharper first form of the Hardy inequality of Theorem~\ref{thm:Hardy}
	and Remark~\ref{rmk:Hardy-improvement}.
	While Ising anyons exhibit a uniform statistical repulsion
	with $\beta_p$ independent of $p$,
	Fibonacci and in particular Clifford anyons could in principle prefer to cluster to
	minimize their repulsion, just like we consider this also a possible preferred behavior 
	(balanced by the uncertainty principle) for abelian anyons
	with $\alpha_* < \alpha$ \cite{LunSol-13b,Lundholm-16}.
	We also note in a similar way that the statistical repulsion of Majorana fermions $\psi$ 
	in the Ising model is potentially weakened to some probability, 
	$U_{\psi,\sigma,\psi} = +\1$ by \eqref{eq:psi-braiding},
	thus allowing for pairing to happen on larger scales.
\end{remark}

\subsection{Lieb--Thirring inequality} 

	In the case that the gas is not homogeneous, such as if an additional external
	scalar potential is added to the Hamiltonian \eqref{eq:Hamiltonian},
	a very powerful bound is given by the Lieb--Thirring inequality,
	which estimates the kinetic energy of $\Psi$ locally at $\bx \in \R^2$ 
	in terms of the homogeneous gas energy at the corresponding density $\varrho_\Psi(\bx)$
	(for fermions that would be the famous Thomas--Fermi appoximation \eqref{eq:TF-approx}).
	The bound combines the uncertainty principle and the exclusion principle
	into a single uniform bound for arbitrary $N$-anyon states $\Psi \in H^1_\rho$.

\begin{theorem}[\textbf{Lieb--Thirring inequality for ideal anyons}]\label{thm:LT} 
	There exists a constant $C>0$ such that 
	for any number of particles $N \ge 1$, 
	for any $N$-anyon model $\rho\colon B_N \to \sU(\cF_N)$
	with 2-particle exchange parameter
	$\alpha_2 \in [0,1]$ (which may depend on $N$), 
	and for any $L^2$-normalized $N$-anyon state $\Psi \in H^1_\rho(\tcC^N;\cF_N)$,
	we have
	\begin{equation}\label{eq:LT}
		T_\rho[\Psi]
		\ge C \alpha_2 \int_{\R^2} \varrho_\Psi(\bx)^2 \, d\bx\,.
	\end{equation}
	Furthermore, given a one-body potential $V\colon \R^2 \to \R$
	and $\hV(X) = \sum_{j=1}^N V(\bx_j)$,
	\begin{equation}\label{eq:LT-potential}
		\infspec (\hT_\rho + \hV) \ge -\frac{1}{4C\alpha_2} \int_{\R^2} V_-(\bx)^2 \,d\bx,
	\end{equation}
	where $V_- := \max\{-V,0\}$.
\end{theorem}

\begin{remark}
	For bosons, with $\alpha_2=0$, the inequality \eqref{eq:LT}
	cannot hold with any better constant since for product states
	\eqref{eq:product-state} the l.h.s.\ is $N\int |\nabla u|^2$
	while $\varrho_\Psi(\bx) = N|u(\bx)|^2$ in the r.h.s., 
	so the ratio cannot be better than $C_\mathrm{GNS}/N \to 0$, where
	$$
		C_\mathrm{GNS} := \inf_{u \in H^1(\R^2) : \int_{\R^2} |u|^2 = 1} \frac{\int_{\R^2} |\nabla u|^2}{\int_{\R^2} |u|^4}
	$$ 
	is the optimal constant of the 
	corresponding Ladyzhenskaya-Gagliardo-Nierenberg-Sobolev inequality for $N=1$.
\end{remark}
\begin{remark}
    In \cite{LarLunNam-19} it was shown that a vanishing condition \eqref{eq:Pauli} on the diagonal implies a Lieb--Thirring (LT) inequality iff the codimension $d < 2s$, the order of the operator $\hat{T}$ (here $s=1$). Thus, $d=2$ is a critical case, and while the hard-core Bose gas has vanishing energy in the dilute limit, the stronger vanishing around diagonals implied by the Hardy inequality is sufficient to produce nontrivial LT for such particles.
\end{remark}
\begin{remark}
	We expect that the optimal constant $C$ in \eqref{eq:LT} is on the order of $2\pi$,
	which is the Thomas-Fermi constant in 2D and is obtained in Weyl's law 
	\eqref{eq:Fermi-energy} for the
	sum of the eigenvalues of the Laplacian on a bounded domain
	(although recall also our remarks after \eqref{eq:TF-approx}).
	The standing conjecture for fermions \cite{LieThi-76} 
	is that the optimal constant is slightly smaller and equal to $C_\mathrm{GNS}$ 
	(which is only known numerically).
	However, proving this even for fermions is a difficult open problem; 
	see e.g. \cite{Frank-20,SeiSol-23} for recent review.
\end{remark}

	We prove Theorem~\ref{thm:LT} as an application of the local exclusion principle, 
	Lemma~\ref{lem:local-exclusion},
	combined with a local version of the uncertainty principle for bosons 
	or distinguishable particles.
	The method was introduced for abelian anyons in \cite{LunSol-13a}
	and goes back to Dyson and Lenard's approach to the proof of 
	the stability of fermionic matter \cite[Lemma~5]{DysLen-67}.
	We could simply replace \cite[Lemma~8]{LunSol-13a} by our Lemma~\ref{lem:local-exclusion},
	however we give for completeness a more immediate proof 
	from \cite{LunNamPor-16} formulated using a covering lemma 
	(see \cite{Lundholm-17} for a more detailed exposition of this local approach 
	to Lieb-Thirring inequalities).
	
\begin{lemma}[{Local uncertainty principle \cite[Lemma~9-10]{LunSol-13a}, \cite[Lemma~14]{LunSol-14}}]\label{lem:local-uncertainty}
	\mbox{}\\
	For arbitrary $d\ge 1$ and $N \ge 1$ let
	$\Psi$ be a function in $H^1(\R^{dN})$, 
	normalized $\int_{\R^{dN}} |\Psi|^2 = 1$,
	and let $Q$ be an arbitrary cube in $\R^d$. 
	There exists a constant $C_1>0$ depending only on $d$, such that
	\begin{equation} \label{eq:local-uncertainty}
		T_0^{Q \subseteq \R^d}[\Psi] 
		= \int_{\R^{dN}} \sum_{j=1}^N |\nabla_j\Psi|^2 \1_{\{\bx_j \in Q\}} d\sx
		\ \ge \ C_1\frac{\int_Q \varrho_{\Psi}^{1+2/d}}{\Bigl(\int_Q \varrho_{\Psi} \Bigr)^{2/d}} 
		- \frac{1}{C_1}\frac{1}{|Q|^{2/d}} \int_Q \varrho_{\Psi},
	\end{equation}
	with the one-body density
	\begin{equation}\label{eq:def-density-Rd}
		\varrho_\Psi(\bx) := 
		\sum_{j=1}^N \int_{\R^{d(N-1)}} |\Psi(\bx_1, \ldots, 
		\bx_{j-1}, \bx, \bx_{j+1}, \ldots, \bx_N)|^2 \prod_{k \neq j}d\bx_k.
	\end{equation}
\end{lemma}

\begin{lemma}[{Covering lemma \cite{LunNamPor-16}}] \label{lem:covering}
	Let $0\le f\in L^1(\R^d)$ be a function with compact support such that 
	$\int_{\R^d} f \ge \Lambda>0$. Then the support of $f$ can be covered by a 
	collection of disjoint cubes $\{Q\}$ in $\R^d$ such that 
	\begin{equation}\label{eq:covering-0}
		\int_{Q} f \le \Lambda, \quad \forall Q
	\end{equation}
	and
	\begin{equation}\label{eq:covering}
		\sum_{Q} \frac{1}{|Q|^{a}} \Biggl( \biggl[\int_{Q} f - q \biggr]_+ 
			- b \int_{Q} f \Biggr) \ge 0
	\end{equation}
	for all $a>0$ and $0\le q< \Lambda 2^{-d}$, where
	$$
		b:= \biggl( 1- \frac{2^dq}{\Lambda} \biggr) \frac{2^{da}-1}{2^{da} + 2^d - 2}>0. 
	$$
\end{lemma}

\begin{proof}[Proof of Theorem~\ref{thm:LT}]
	
	If $\alpha_2=0$ the theorem is trivially true, and otherwise we have $E_2(Q)>0$
	and thus a non-trivial local exclusion principle on squares $Q$
	given by Lemma~\ref{lem:local-exclusion} with $C(\rho) \ge C_4(\rho) \ge c(\alpha_2) > 0$.
	By definition of the space $H^1_\rho$ we may w.l.o.g. assume $\Psi$ is smooth and
	the projection by $\tilde\pr$ of its support to $\cC^N$ contained in the projection 
	by $\pr$ of some large cube $Q_L^N$,
	$Q_L = [-L,L]^2$.

	Let $q=1$ and $\Lambda=5$. 
	If $N\le \Lambda$, then~\eqref{eq:LT} follows immediately from \eqref{eq:local-uncertainty} 
	with $Q=Q_L$, $L \to \infty$. 
	If $N> \Lambda$, then we can apply Lemma \ref{lem:covering} with $f=\varrho_{\Psi}$
	and $a=1$, and obtain a collection of disjoint cubes $\{Q\}$ covering $Q_L$. 
	We use that
	$$
		T_\rho[\Psi] = \sum_{Q} T_\rho^{Q \subseteq Q_L}[\Psi],
	$$ 
	and by the diamagnetic inequality of Lemma~\ref{lem:diamag}, 
	the scalar function $|\Psi|_\cF$ may be viewed as an element of $H^1(\R^{2N})$.  
	We therefore apply
	Lemma~\ref{lem:local-uncertainty} to $|\Psi|_\cF$, obtaining
	$$
		T_\rho[\Psi] \ge T_0\bigl[|\Psi|_\cF\bigr],
		\qquad \varrho_{|\Psi|_\cF} = \varrho_\Psi,
	$$
	and finally combining the bounds \eqref{eq:local-uncertainty} and \eqref{eq:local-exclusion}, 
	we obtain
	\begin{align*}
		(\eps+1) T_\rho[\Psi] 
		&\ge \eps \sum_Q \left[ C_1\frac{\int_Q \varrho_{\Psi}^2}{\int_Q \varrho_{\Psi}} - \frac{1}{C_1}\frac{1}{|Q|} \int_Q \varrho_{\Psi}\right] 
		 + \sum_Q \frac{C(\rho)}{|Q|}  \biggl[ \int_Q \varrho_\Psi(\bx)\,d\bx - 1 \biggr]_+  \\
		&\ge \eps C_1 \frac{\int_{\R^2} \varrho_{\Psi}^{2}}{\Lambda} 
	\end{align*}
	for any fixed constant $\eps>0$ satisfying $\eps \le C_1 C(\rho) b$.
	Taking
	$$
		\eps = C_1 b \min\{\alpha_2/12,0.147/4,(C_1b)^{-1}\}
	$$ 
	proves the first inequality \eqref{eq:LT} with $C = C_1 \min\{C_1 b,1\}/120 > 0$. 
	
	For the second inequality \eqref{eq:LT-potential} we use a well-known equivalence of
	the bounds, namely by Cauchy-Schwarz and simple optimization
	\begin{align*}
		\left\langle \Psi, (\hT_\rho + \hV)\Psi \right\rangle_{L^2_\rho}
		&= T_\rho[\Psi] - \int_{\R^2} V \varrho_\Psi
		\ge T_\rho[\Psi] - \int_{\R^2} |V_-| \varrho_\Psi \\
		&\ge C \alpha_2 \int_{\R^2} \varrho_\Psi^2 - \left( \int_{\R^2} |V_-|^2 \right)^{1/2} \left( \int_{\R^2} \varrho_\Psi^2 \right)^{1/2} 
		\ge -\frac{1}{4C\alpha_2} \int_{\R^2} |V_-|^2,
	\end{align*}
	which is a uniform bound in $N$ provided that $\alpha_2 = \alpha_2(N)$
	stays uniformly bounded away from zero.
\end{proof}

%
%
%
%
%
%



\def\MR#1{} 

\newcommand{\etalchar}[1]{$^{#1}$}


\end{document}